\documentclass[11pt,ruled,linesnumbered,vlined]{article}
\usepackage[T1]{fontenc}
\usepackage[latin9]{inputenc}
\usepackage{color}
\usepackage{verbatim}
\usepackage{float}
\usepackage{mathrsfs}
\usepackage{algorithm2e}
\usepackage{amsmath}
\usepackage{amsthm}
\usepackage{amssymb}
\usepackage{graphicx}
\usepackage{esint}
\usepackage[authoryear]{natbib}
\usepackage[unicode=true,pdfusetitle,
 bookmarks=true,bookmarksnumbered=false,bookmarksopen=false,
 breaklinks=false,pdfborder={0 0 1},backref=false,colorlinks=true]
 {hyperref}
\hypersetup{
 citecolor=blue}
\usepackage{breakurl}

\makeatletter
\newcommand{\lyxaddress}[1]{
\par {\raggedright #1
\vspace{1.4em}
\noindent\par}
}
 \theoremstyle{definition}
  \newtheorem{example}{\protect\examplename}
  \theoremstyle{remark}
  \newtheorem{rem}{\protect\remarkname}
\theoremstyle{plain}
\newtheorem{thm}{\protect\theoremname}
  \theoremstyle{plain}
  \newtheorem{prop}{\protect\propositionname}
  \theoremstyle{plain}
  \newtheorem{cor}{\protect\corollaryname}
  \theoremstyle{plain}
  \newtheorem{lem}{\protect\lemmaname}

\usepackage{stmaryrd}

\usepackage{amsfonts}
\usepackage{fullpage}
\usepackage{subfigure}
\usepackage{comment}
\usepackage{epsfig}\usepackage{authblk}
\usepackage{bbm}

\usepackage[font=footnotesize, labelfont=bf]{caption}

\providecommand{\examplename}{Example}
\providecommand{\remarkname}{Remark}
\providecommand{\theoremname}{Theorem}

  \providecommand{\examplename}{Example}
  \providecommand{\propositionname}{Proposition}
  \providecommand{\remarkname}{Remark}
\providecommand{\theoremname}{Theorem}

\makeatother

  \providecommand{\examplename}{Example}
  \providecommand{\lemmaname}{Lemma}
  \providecommand{\propositionname}{Proposition}
  \providecommand{\remarkname}{Remark}
\providecommand{\corollaryname}{Corollary}
\providecommand{\theoremname}{Theorem}

\begin{document}

\title{On the utility of Metropolis-Hastings with asymmetric acceptance
ratio}

\author{Christophe Andrieu$^{*}$, Arnaud Doucet$^{\dagger}$, Sinan Y\i ld\i r\i m$^{+}$
and Nicolas Chopin$^{\blacklozenge}$}
\maketitle

\lyxaddress{$^{*}$School of Mathematics, University of Bristol, U.K. \\
$^{\dagger}$Department of Statistics, University of Oxford, U.K.\\
$^{+}$Faculty of Engineering and Natural Sciences, Sabanc\i{} University,
Turkey.\\
$^{\blacklozenge}$ENSAE, France.}
\begin{abstract}
{\normalsize{}The Metropolis\textendash Hastings algorithm allows
one to sample asymptotically from any probability distribution $\pi$
admitting a density with respect to a reference measure, also denoted
$\pi$ here, which can be evaluated pointwise up to a normalising
constant. There has been recently much work devoted to the development
of variants of the Metropolis\textendash Hastings update which can
handle scenarios where such an evaluation is impossible, and yet are
guaranteed to sample from $\pi$ asymptotically. The most popular
approach to have emerged is arguably the pseudo-marginal Metropolis\textendash Hastings
algorithm which substitutes an unbiased estimate of an unnormalised
version of $\pi$ for $\pi$ }\citep{Lin_2000,Beaumont_2003,Andrieu_and_Roberts_2009}\textbf{.}{\normalsize{}
Alternative pseudo-marginal algorithms relying instead on unbiased
estimates of the Metropolis\textendash Hastings acceptance ratio have
also been proposed \citep{Neal_2004,Murray_et_al_2006,Nicholls_et_al_2012}.
These algorithms can have better properties than standard pseudo-marginal
algorithms. Convergence properties of both classes of algorithms are
known to depend on the variability (in the sense of the convex order)
of the estimators involved }\citep{Andrieu_and_Vihola_2014}{\normalsize{},
and reduced variability is guaranteed to decrease the asymptotic variance
of ergodic averages and will shorten the ``burn-in'' period, or
convergence to equilibrium, in most scenarios of interest. A simple
approach to reduce variability, amenable to parallel computations,
consists of averaging independent estimators. However, while averaging
estimators of $\pi$ in a pseudo-marginal algorithm retains the guarantee
of sampling from $\pi$ asymptotically, naive averaging of acceptance
ratio estimates breaks detailed balance, leading to incorrect results.
We propose an original methodology which allows for a correct implementation
of this idea. We establish theoretical properties which parallel those
available for standard pseudo-marginal algorithms and discussed above.
We demonstrate the interest of the approach on various inference problems
involving doubly intractable distributions, latent variable models,
model selection, and state-space models. In particular we show that
convergence to equilibrium can be significantly shortened, therefore
offering the possibility to reduce a user's waiting time in a generic
fashion when a parallel computing architecture is available.}{\normalsize \par}
\end{abstract}
{\small{}Keywords: Annealed Importance Sampling; Doubly intractable
distributions; Intractable likelihood; Markov chain Monte Carlo; Reversible
jump Monte Carlo; Sequential Monte Carlo; State-space models.}{\small \par}

\newpage{}

\tableofcontents{}

\newpage{}

\section{Introduction\label{sec: Introduction} }

Suppose we are interested in sampling from a given probability distribution
$\pi$ on some measurable space $(\mathsf{X},\mathcal{X})$. When
it is impossible or too difficult to generate perfect samples from
$\pi$, one practical resource is to use a Markov chain Monte Carlo
(MCMC) algorithm which generates an ergodic Markov chain $\{X_{n},n\geq0\}$
whose invariant distribution is $\pi$. Among MCMC methods, the Metropolis\textendash Hastings
(MH) algorithm plays a central rôle. The MH update proceeds as follows:
given $X_{n}=x$ and a Markov transition kernel $q\big(x,\cdot\big)$
on $(\mathsf{X},\mathcal{X})$, we propose $y\sim q(x,\cdot)$ and
set $X_{n+1}=y$ with probability $\alpha(x,y):=\min\left\{ 1,r(x,y)\right\} $,
where 
\begin{equation}
r(x,y):=\frac{\pi({\rm d}y)q(y,{\rm d}x)}{\pi({\rm d}x)q(x,{\rm d}y)}\label{eq:genericMHacceptratio}
\end{equation}
for $(x,y)\in\mathsf{S}\subset\mathsf{X}^{2}$ (see Appendix \ref{sec: A general framework for MPR and MHAAR algorithms}
for a definition of $\mathsf{S}$) is a well defined Radon\textendash Nikodym
derivative, and $r(x,y)=0$ otherwise. When the proposed value $y$
is rejected, we set $X_{n+1}=x$. We will refer to $r(x,y)$ as the
acceptance ratio. The transition kernel of the Markov chain $\{X_{n},n\geq0\}$
generated with the MH algorithm with proposal kernel $q(\cdot,\cdot)$
is 
\begin{equation}
P(x,{\rm d}y)=q(x,{\rm d}y)\alpha(x,y)+\rho(x)\delta_{x}({\rm d}y),\quad x\in\mathsf{X},\label{eq: MH transition kernel}
\end{equation}
where $\rho(x)$ is the probability of rejecting a proposed sample
when $X_{n}=x$, 
\[
\rho(x):=1-\int_{\mathsf{X}}\alpha(x,y)Q(x,{\rm d}y)
\]
and $\delta_{x}(\cdot)$ is the Dirac measure centred at $x$. Expectations
of functions, say $f$, with respect to $\pi$ can be estimated with
$S_{M}:=M^{-1}\sum_{n=1}^{M}f(X_{n})$ for $M\in\mathbb{N}$, which
is consistent under mild assumptions.

Being able to evaluate the acceptance ratio $r(x,y)$ is obviously
central to implementing the MH algorithm in practice. Recently, there
has been much interest in expanding the scope of the MH algorithm
to situations where this acceptance ratio is intractable, that is,
impossible or very expensive to compute. A canonical example of intractability
is when $\pi$ can be written as the marginal of a given joint probability
distribution for $x$ and some latent variable $z$. A classical way
of addressing this problem consists of running an MCMC targeting the
joint distribution, which may however become very inefficient in situations
where the size of the latent variable is high\textendash this is for
example the case for general state-space models. In what follows,
we will briefly review some more effective ways of tackling this problem.
To that purpose we will use the following simple running example to
illustrate various methods. This example has the advantage that its
setup is relatively simple and of clear practical relevance. We postpone
developments for much more complicated setups to Sections \ref{sec: Pseudo-marginal ratio algorithms for latent variable models}
and \ref{sec: State-space models: SMC and conditional SMC within MHAAR}. 
\begin{example}[\textbf{Inference with doubly intractable models}]
\label{ex:doublyintractable} In this scenario the likelihood function
of the unknown parameter $\theta\in\Theta$ for the dataset $\mathfrak{y}\in\mathsf{Y}$,
$\ell_{\theta}(\mathfrak{y})$, is only known up to a normalising
constant, that is 
\[
\ell_{\theta}(\mathfrak{y})=\frac{g_{\theta}(\mathfrak{y})}{C_{\theta}},
\]
where $C_{\theta}$ is unknown, while $g_{\theta}(\mathfrak{y})$
can be evaluated pointwise for any value of $\theta\in\Theta$. In
a Bayesian framework, for a prior density $\eta(\theta)$, we are
interested in the posterior density $\pi(\theta)$, with respect to
some measure, given by 
\[
\pi(\theta)\propto\eta(\theta)\ell_{\theta}(\mathfrak{y}).
\]
With $x=\theta,y=\theta'$ in \eqref{eq:genericMHacceptratio}, the
resulting acceptance ratio of the MH algorithm associated to a proposal
density $q(\theta,\theta')$ is 
\begin{align}
r(\theta,\theta')=\frac{q(\theta',\theta)}{q(\theta,\theta')}\frac{\eta(\theta')}{\eta(\theta)}\frac{g_{\theta'}(\mathfrak{y})}{g_{\theta}(\mathfrak{y})}\frac{C_{\theta}}{C_{\theta'}},\label{eq: MCMC acceptance probability with intractable likelihood-1}
\end{align}
which cannot be calculated because of the unknown ratio $C_{\theta}/C_{\theta'}$.
While the likelihood function may be intractable, sampling artificial
datasets $\mathfrak{z}\sim\ell_{\theta}(\cdot)$ may be possible for
any $\theta\in\Theta$, and sometimes computationally cheap. We will
describe two known approaches which exploit and expand this property
in order to design Markov kernels preserving $\pi(\theta)$ as invariant
density. 
\end{example}

\subsection{Estimating the target density\label{subsec: Estimating the target density}}

Assume for simplicity that $\pi$ has a probability density with respect
to some $\sigma$-finite measure. We will abuse notation slightly
by using $\pi$ for both the probability distribution and its density.
A powerful, yet simple, method to tackle intractability which has
recently attracted substantial interest consists of replacing the
value of $\pi(x)$ with a non-negative random estimator $\hat{\pi}(x)$
whenever it is required in the implementation of the MH algorithm
above. If $\mathbb{E}[\hat{\pi}(x)]=C\pi(x)$ for all $x\in\mathsf{X}$
and a constant $C>0$, a property we refer somewhat abusively as unbiasedness,
this strategy turns out to lead to exact algorithms, that is sampling
from $\pi$ is guaranteed at equilibrium under very mild assumptions
on $\hat{\pi}(x)$. This approach leads to so called pseudo-marginal
algorithms \citep{Andrieu_and_Roberts_2009}. However, for reasons
which will become clearer later, we refer from now on to these techniques
as Pseudo-Marginal Target (PMT) algorithms.
\begin{example}[\textbf{Example 1, ctd}]
\label{ex: pseudo-marginal for doubly intractable models} Let $h_{\mathfrak{y}}:\mathsf{Y}\rightarrow[0,\infty)$
be an integrable non-negative function of integral equal to $1$.
For a given $\theta$, an unbiased estimate of $\pi(\theta)$ can
be obtained via importance sampling whenever the support of $g_{\theta}$
includes that of $h_{\mathfrak{y}}$: 
\begin{align}
\hat{\pi}^{N}(\theta)\propto\eta(\theta)g_{\theta}(\mathfrak{y})\left\{ \frac{1}{N}\sum_{i=1}^{N}\frac{h_{\mathfrak{y}}(\mathfrak{z}^{(i)})}{g_{\theta}(\mathfrak{z}^{(i)})}\right\} ,\quad\mathfrak{z}^{(i)}\overset{{\rm iid}}{\sim}\ell_{\theta}(\cdot),\quad i=1,\ldots,N,
\end{align}
since the normalised sum is an unbiased estimator of $1/C_{\theta}$.
The auxiliary variable method of \citet{Muller_et_al_2006} corresponds
to $N=1$. An interesting feature of this approach is that $N$ is
a free parameter of the algorithm which reduces the variability of
this estimator. It is shown in \citet{Andrieu_and_Vihola_2014} that
increasing $N$ in a PMT algorithm always reduces the asymptotic variance
of averages using this chain. This is particularly interesting in
a parallel computing environment, but also serial for some models.
We illustrate this numerically on a simple Ising model (see details
in Section \ref{subsec: Numerical example: the Ising model}) for
a $20\times20$ lattice and $h_{\mathfrak{y}}(\mathfrak{z})=g_{\hat{\theta}}\big(\mathfrak{z}\big)$,
where $\hat{\theta}$ is an approximation of the maximum likelihood
estimator of $\theta$ for the data $\mathfrak{y}$. In Figure \ref{fig:IACIsingPseudoMarginal}
we report the estimated integrated auto-covariance (IAC) for the identity,
that is $\lim_{M\rightarrow\infty}M\mathsf{\mathbb{\mathsf{var}}}\big(S_{M}\big)/\mathbb{\mathsf{var}}_{\pi}(f)$
for the function $f(\theta)=\theta$, as a function of $N$ and values
of $\hat{\theta}$. The results are highly dependent on the value
of $\hat{\theta}$, but adjusting $N$ allows one to compensate for
a wrong choice of this parameter. This is important in practice since
for more complicated scenarios obtaining a good approximation of the
maximum likelihood estimator of $\theta$ may be difficult.

\begin{figure}
\includegraphics{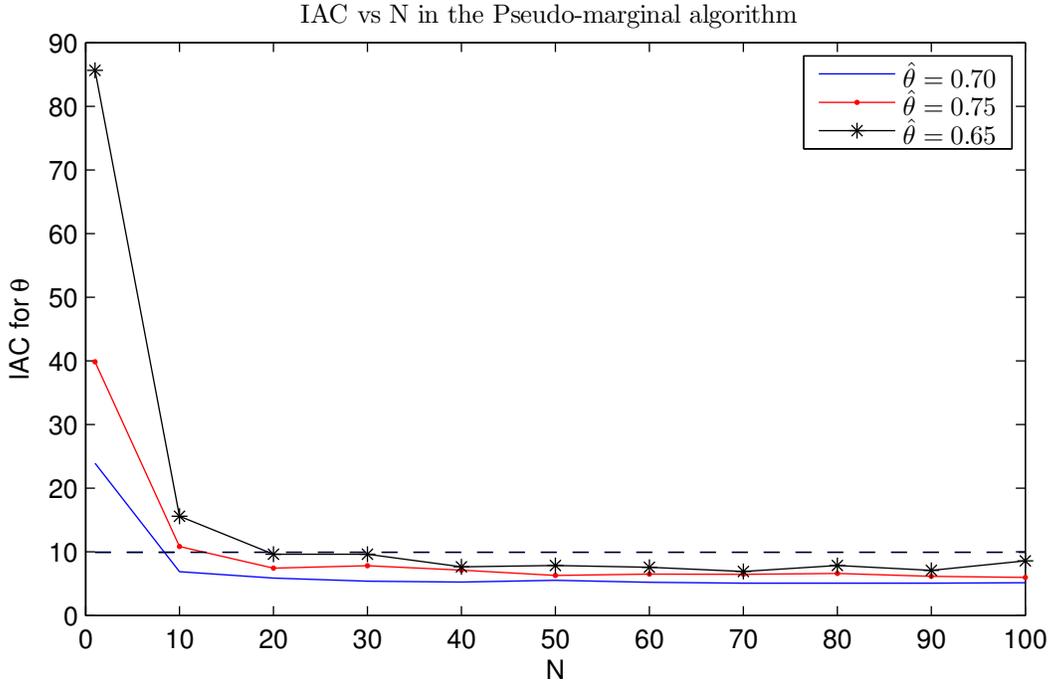}\protect\caption{IAC of the algorithm as a function of $N$.\label{fig:IACIsingPseudoMarginal}}
\end{figure}
\end{example}

\subsection{Estimating the acceptance ratio \label{subsec: Estimating the acceptance ratio}}

One can in fact push the idea of replacing algebraic expressions with
estimators further. Instead of approximating the numerator and denominator
of the acceptance ratio $r(x,y)$ independently, it is indeed possible
to use directly estimators of the acceptance ratio $r(x,y)$ and still
obtain algorithms guaranteed to sample from $\pi$ at equilibrium.
We will refer to these algorithms as Pseudo-Marginal Ratio (PMR) algorithms.
A general framework is described in \citet{Andrieu_and_Vihola_2014}
as well as in Section \ref{subsec: Pseudo-marginal ratio algorithms},
but this idea has appeared earlier in various forms in the literature,
see e.g.\ \citet{Nicholls_et_al_2012} and the references therein.
An interesting feature of PMR algorithms is that we estimate the ratio
$r(x,y)$ afresh whenever it is required. On the contrary, in a PMT
framework, if the estimate $\hat{\pi}(x)/C$ of the current state
significantly overestimates $\pi(x)$, this results in poor performance
as the algorithm will typically reject many transitions away from
$x$ as the same estimate of $\pi(x)$ is used until a proposal is
accepted. In the following continuation of Example \ref{ex:doublyintractable},
we present a particular case of this PMR idea proposed by \citet{Murray_et_al_2006}. 
\begin{example}[\textbf{Example 1, ctd}]
\label{ex: exchange algorithm} The exchange algorithm of \citet{Murray_et_al_2006}
is motivated by the realisation that while for $\mathfrak{z}\sim\ell_{\theta'}(\cdot)$
and any $\mathfrak{y}\in\mathsf{Y}$, $h_{\mathfrak{y}}(\mathfrak{z})/g_{\theta'}(\mathfrak{z})$
is an unbiased estimator of $1/C_{\theta'}$, the particular choice
$h_{\mathfrak{y}}(\mathfrak{z})=g_{\theta}(\mathfrak{z})$ leads to
an unbiased estimator $g_{\theta}(\mathfrak{z})/g_{\theta'}(\mathfrak{z})$
of $C_{\theta}/C_{\theta'}$. We can expect this estimator to have
a reasonable variance when $\theta$ and $\theta'$ are close if $\theta\mapsto g_{\theta}(\mathfrak{z})$
satisfies some form of continuity. This suggests the following algorithm.
Given $\theta\in\Theta$, sample $\theta'\sim q(\theta,\cdot)$, then
$\mathfrak{z}\sim\ell_{\theta'}(\cdot)$ and use the acceptance ratio
\begin{equation}
\frac{q(\theta',\theta)}{q(\theta,\theta')}\frac{\eta(\theta')}{\eta(\theta)}\frac{g_{\theta'}(\mathfrak{y})}{g_{\theta}(\mathfrak{y})}\frac{g_{\theta}(\mathfrak{z})}{g_{\theta'}(\mathfrak{z})},\label{eq: Murray's acceptance ratio-1}
\end{equation}
which is an unbiased estimator of the acceptance ratio in \eqref{eq: MCMC acceptance probability with intractable likelihood-1}.
The remarkable property of this algorithm is that it admits $\pi(\theta)$
as an invariant distribution and hence, under mild exploration related
assumptions, it is guaranteed to produce samples asymptotically distributed
according to $\pi$.
\end{example}

\subsection{Contribution\label{subsec: Contribution}}

As for PMT algorithms, it is natural to ask whether it is possible
to further improve the performance of PMR algorithms by reducing the
variability of the acceptance ratio estimator by averaging a number
of such estimators, while preserving the target distribution $\pi$
invariant. We shall see that, unfortunately, such naïve averaging
approach does not work for the PMR methods currently available as
it breaks the reversibility of the kernels with respect to $\pi$. 

A contribution of the present paper is the introduction of a novel
class of PMR algorithms which can exploit acceptance ratio estimators
obtained through averaging and sample from $\pi$ at equilibrium.
These algorithms described in Section \ref{sec: Pseudo-marginal ratio algorithms using averaged acceptance ratio estimators}
naturally lend themselves to parallel computations as independent
ratio estimators can be computed in parallel at each iteration. In
this respect, our methodology contributes to the emerging literature
on the use of parallel processing units, such as Graphics Processing
Units (GPUs) or multicore chips for scientific computation \citet{lee2010utility,suchard2010understanding}.
We show that this generic procedure is guaranteed to decrease the
asymptotic variance of ergodic averages as the number of independent
ratios $N$ one averages increases and that the burn-in period will
be reduced in most scenarios. The latter is particularly relevant
since exact and generic methods to achieve this are scarce \citep{sohn1995parallel},
in contrast with variance reduction techniques for which better embarrassingly
parallel solutions \citep{doi:10.1093/biomet/asx031,bornn2017use}
and/or post-processing methods are available \citep{delmas2009does,dellaportas2012control}.
We demonstrate experimentally its performance gain for the exchange
algorithm.

This new class of PMR algorithms can be understood as being a particular
instance of a more general principle which we exploit further in this
paper, beyond the above example. Let $Q_{1},Q_{2}\colon\mathsf{X}\times\mathcal{X}\rightarrow[0,1]$
be a pair of kernels such that the following Radon-Nikodym derivative
\[
r_{1}(x,y):=\frac{\pi({\rm d}y)Q_{2}(y,{\rm d}x)}{\pi({\rm d}x)Q_{1}(x,{\rm d}y)}
\]
is well defined for $(x,y)$ on some symmetric set $\mathsf{S}$ and
set $r_{1}(x,y)=0$ otherwise. This can be thought of as an asymmetric
version of the standard MH acceptance ratio \eqref{eq:genericMHacceptratio}
and naturally leads to two questions. 
\begin{enumerate}
\item Assuming that sampling from $Q_{1}(x,\cdot)$ and $Q_{2}(x,\cdot)$
for any $x\in\mathsf{X}$ is feasible and that $r_{1}(x,y)$ is tractable,
can one design a correct MCMC algorithm for $\pi$ that involves simulating
from $Q_{1}(x,\cdot)$ and $Q_{2}(x,\cdot)$ and evaluating $r_{1}(x,y)$
?
\item Assuming the answer to the above is positive, can this additional
degree of freedom be beneficial in order to design correct MCMC algorithms
with practically appealing features e.g.\ accelerated convergence?
\end{enumerate}
The answer to the first question is unsurprisingly yes, and we will
refer to the corresponding class of algorithms as MH with Asymmetric
Acceptance Ratio (MHAAR). MHAAR has already been exploited in some
specific contexts \citep{tjelmeland-eidsvik-2004,andrieu2008tutorial},
but its best known application certainly remains the reversible jump
MCMC methodology of \citet{Green_1995}. However the way we take advantage
of this additional flexibility seems completely novel. We also note,
as detailed in our discussion in Section \ref{sec: Discussion}, that
such asymmetric acceptance ratios are also at the heart of non-reversible
MCMC algorithms which have recently attracted renewed interest in
the Physics and Statistical communities \citep{gustafson1998guided,turitsyn2011irreversible}.
In Appendix \ref{sec: A general framework for MPR and MHAAR algorithms},
we describe and justify a slightly more general framework to the above
which ensures reversibility with respect to $\pi$. The answer to
the second question is the object of this paper, and averaging acceptance
ratios as suggested earlier is one such application.

In Section \ref{sec: Improving pseudo-marginal ratio algorithms for doubly intractable models}
we further investigate the doubly intractable scenario by incorporating
the Annealed Importance Sampling (AIS) mechanism \citep{Neal_2001,Murray_et_al_2006}
in MHAAR, and explore numerically the performance of MHAAR with AIS
on an Ising model.

In Section \ref{sec: Pseudo-marginal ratio algorithms for latent variable models}
we expand the class of problems our methodology can address by considering
latent variable models. This leads to important extensions of the
original AIS within MH algorithm proposed in \citet{Neal_2004}. We
demonstrate the efficiency of our MHAAR-based approach by recasting
the popular reversible jump MCMC (RJ-MCMC) methodology as a particular
case of our framework and illustrate the computational benefits of
our novel algorithm in this context on the Poisson change-point model
in \citet{Green_1995}.

In Section \ref{sec: State-space models: SMC and conditional SMC within MHAAR},
we show how MHAAR can be advantageous in the context of inference
in state-space models when it is utilised with sequential Monte Carlo
(SMC) algorithms. In particular, we expand the scope of particle MCMC
algorithms \citep{Andrieu_et_al_2010} and show novel ways of using
multiple or all possible paths from backward sampling of conditional
SMC (cSMC) to estimate the marginal acceptance ratio.

In Section \ref{sec: Discussion}, we provide some discussion and
two interesting extensions of MHAAR. Specifically, in Section \ref{sec: Using SMC based estimators for the acceptance ratio}
we discuss an SMC-based generalisation of our algorithms involving
AIS. Furthermore, in Section \ref{sec: Links to non-reversible algorithms}
we provide a new insight to non-reversible versions of MH algorithms
that is relevant to our setting. We briefly demonstrate how non-reversible
versions of our algorithms can be obtained with a small modification
so that one can benefit both from non-reversibility and the ability
to average acceptance ratio estimators.

Some of the proofs of the validity of our algorithms as well as additional
discussion on the generalisation of the methods can found in the Appendices. 

\section{PMR algorithms using averaged acceptance ratio estimators\label{sec: Pseudo-marginal ratio algorithms using averaged acceptance ratio estimators}}

\subsection{PMR algorithms\label{subsec: Pseudo-marginal ratio algorithms}}

We introduce here generic PMR algorithms, that is MH algorithms relying
on an estimator of the acceptance ratio. We then show that in their
standard form these algorithms cannot use an estimator of this ratio
obtained through averaging independent estimators. A slightly more
general framework is provided in \citet{Nicholls_et_al_2012}, while
a more abstract description is provided in \citet{Andrieu_and_Vihola_2014}.
To that purpose we introduce a $(\mathsf{U},\mathcal{U})$-valued
auxiliary variable $u$ (we use small letters for random variables
and realisations throughout) and let $\varphi:\mathsf{U}\rightarrow\mathsf{U}$
be a measurable involution, that is $\varphi=\varphi^{-1}$. Then
we introduce a pair of families of proposal distributions $\{Q_{1}(x,\cdot),x\in\mathsf{X}\}$
, $\{Q_{2}(x,\cdot),x\in\mathsf{X}\}$ on $(\mathsf{X}\times\mathsf{U},\mathcal{X}\times\mathcal{U})$,
where 
\begin{equation}
Q_{1}\big(x,{\rm d}(y,u)\big):=q(x,{\rm d}y)Q_{x,y}({\rm d}u)\label{eq: PMR Q1}
\end{equation}
with $Q_{x,y}(\cdot)$ denoting the conditional distribution of $u$
given $x,y\in\mathsf{X}$, and 
\begin{equation}
Q_{2}\big(x,{\rm d}(y,u)\big):=q(x,{\rm d}y)\bar{Q}_{x,y}({\rm d}u),\label{eq: PMR Q2}
\end{equation}
where, for any $A\in\mathcal{U}$ we have 
\begin{equation}
\bar{Q}_{x,y}\big(A\big):=Q_{x,y}\big(\varphi(A)\big).\label{eq: Q_bar}
\end{equation}
which means that in order to sample $u\sim\bar{Q}_{x.y}(\cdot)$,
one can sample $\bar{u}\sim Q_{x,y}(\cdot)$ and set $u=\varphi(\bar{u})$.
PMR algorithms are defined by the following transition kernel 
\begin{equation}
\mathring{P}(x,{\rm d}y)=\int_{\mathsf{U}}Q_{1}\big(x,{\rm d}(y,u)\big)\min\{1,\mathring{r}_{u}(x,y)\}+\mathring{\rho}(x)\delta_{x}({\rm d}y),\label{eq:defPring}
\end{equation}
where the acceptance ratio is equal, for $(x,y,u)\in\mathring{\mathsf{S}}$
and $\mathring{\mathsf{S}}$ defined similarly to \eqref{eq:def-ring-S},
to 

\begin{align}
\mathring{r}_{u}(x,y) & :=\frac{\pi({\rm d}y)Q_{2}\big(y,{\rm d}(x,u)\big)}{\pi(dx)Q_{1}\big(x,{\rm d}(y,u)\big)}\label{eq:accept-ratio-circle}\\
 & =r(x,y)\frac{\bar{Q}_{y,x}({\rm d}u)}{Q_{x,y}({\rm d}u)},\nonumber 
\end{align}
and to $1$ otherwise. It is clear from \eqref{eq:accept-ratio-circle}
that the acceptance ratio $\mathring{r}_{u}(x,y)$ is an unbiased
estimator of the standard MH acceptance $r(x,y)$, i.e. 
\begin{equation}
\int_{\mathsf{U}}\mathring{r}_{u}(x,y)Q_{x,y}({\rm d}u)=r(x,y).\label{eq:unbiasednessratioMH}
\end{equation}
Due to the particular form of symmetry between $Q_{1}$ and $Q_{2}$
imposed by \eqref{eq: Q_bar}, $\mathring{P}$ is reversible with
respect to $\pi$ by considering detailed balance for fixed $u\in\mathsf{U}$;
see Theorem \ref{thm: pseudo-marginal ratio algorithms} in Appendix
\ref{sec: A general framework for MPR and MHAAR algorithms}. 

As long as PMR algorithms are concerned, we call $Q_{1}$ the proposal
kernel of PMR and $Q_{2}$ its complementary kernel, owing to the
way \eqref{eq:defPring} is constructed. Motivation for this enumeration
will be clear in Section \ref{subsec: Pseudo-marginal algorithm with averaged acceptance ratio estimator},
in particular by Remark \ref{rem: MHAAR with N =00003D 1}.
\begin{rem}
\label{remark:ratiounbiasednotsufficient}A cautionary remark is in
order. When we substitute a non-negative unbiased estimator of $\pi$
for $\pi$ in the MH algorithm, the resulting PMT algorithm is $\pi-$invariant.
However, if we substitute a positive unbiased estimator of $r(x,y)$
for $r(x,y)$ in the MH algorithm then the resulting transition kernel
is not necessarily $\pi-$invariant. To establish that $\mathring{P}$
is $\pi-$invariant, we require our estimator to have the specific
structure given in \eqref{eq:accept-ratio-circle}. 
\end{rem}
A particular instance of this algorithm was given earlier in the set-up
of Example \ref{ex:doublyintractable}, where $x=\theta$, the random
variable $u=\mathfrak{z}$ corresponds to a fictitious dataset used
to estimate the ratio of normalising constants, and $\varphi(u)=u$.
The need to consider more general transformations $\varphi$ will
become apparent in Section \ref{sec: Improving pseudo-marginal ratio algorithms for doubly intractable models}. 

This type of algorithms is motivated by the fact that while in some
situations $r(x,y)$ cannot be computed, the introduction of the auxiliary
variable $u$ makes the computation of $\mathring{r}_{u}(x,y)$ possible.
However, this computational tractability comes at a price. Applying
Jensen's inequality to \eqref{eq:unbiasednessratioMH} shows that
\[
\int_{\mathsf{U}}Q_{x,y}({\rm d}u)\min\{1,\mathring{r}_{u}(x,y)\}\leq\min\{1,r(x,y)\}.
\]
Peskun's result \citep{Tierney_1998} thus implies that the MCMC algorithm
relying on $\mathring{P}$ is always inferior to that using $P$ for
various performance measures (see Theorem \ref{thm:theoreticaljustification}
for details). As pointed out in \citet{Andrieu_and_Vihola_2014} reducing
the variability of $\mathring{r}_{u}(x,y)$, for example in the sense
of the convex order, for all $x,y\in\mathsf{X}^{2}$ will reduce the
gap in the inequality above, resulting in improved performance. From
the rightmost expression in \eqref{eq:accept-ratio-circle} a possibility
to reduce variability might be to change $Q_{x,y}\big(\cdot\big)$
(and possibly $u$) in such a way that $Q_{x,y}\simeq\bar{Q}_{y,x}$
for all $x,y\in\mathsf{X}$, but this is impossible in most practical
scenarios. In contrast a natural idea consists of averaging ratios
$\mathring{r}_{u^{(i)}}(x,y)$'s for, say, realisations $u^{(1)},\ldots,u^{(N)}\overset{{\rm iid}}{\sim}Q_{x,y}(\cdot)$
and use the acceptance ratio 
\begin{equation}
\mathring{r}_{\mathfrak{u}}^{N}(x,y):=\frac{1}{N}\sum_{i=1}^{N}\mathring{r}_{u^{(i)}}(x,y),\label{eq: average acceptance ratio}
\end{equation}
where $\mathfrak{u}:=u^{(1:N)}=\big(u^{(1)},\ldots,u^{(N)}\big)\in\mathfrak{U}:=\mathsf{U}^{N}$\textendash we
drop the dependence on $N$ in order to alleviate notation whenever
no ambiguity is possible. While this reduces the variance of the estimator
of $r\big(x,y\big)$, this naïve modification of the acceptance rule
of $\mathring{P}$ breaks detailed balance with respect to $\pi$.
Indeed one can check that with $Q_{1}^{N}(x,{\rm d}(y,\mathfrak{u})):=q(x,{\rm d}y)\prod_{i=1}^{N}Q_{x,y}({\rm d}u^{(i)})$,
$h:\mathsf{X}^{2}\rightarrow\mathbb{R}$ a bounded measurable function
and using Fubini's result, 
\begin{align*}
\int_{\mathsf{X}\times\mathfrak{U}\times\mathsf{X}}\pi\big({\rm d}x\big)Q_{1}^{N}\big(x,{\rm d}(y,{\rm \mathfrak{u})}\big)\min\{1,\mathring{r}_{\mathfrak{u}}^{N}(x,y)\}h\big(x,y\big)\\
\neq\int_{\mathsf{X}\times\mathfrak{U}\times\mathsf{X}}\pi\big({\rm d}y\big)Q_{1}^{N}\big(y,{\rm d}(x,{\rm \mathfrak{u}})\big) & \min\{1,\mathring{r}_{\mathfrak{u}}^{N}(y,x)\}h\big(x,y\big)
\end{align*}
in general. This is best seen in a scenario where $\mathsf{X}$ and
$\mathsf{U}$ are finite and $h(x,y)=\mathbb{I}\{x=a\}\mathbb{I}\{y=b\}$
for some $a,b\in\mathsf{X}$, and it can be shown that $\pi$ is not
left invariant by the corresponding Markov transition probability.

\subsection{MHAAR for averaging PMR estimators \label{subsec: Pseudo-marginal algorithm with averaged acceptance ratio estimator}}

We show here how MHAAR updates can be used in order to use the acceptance
ratio in \eqref{eq: average acceptance ratio}, while preserving $\pi-$reversibility.
Our novel scheme is described in Algorithm \ref{alg: MHAAR for Pseudo-marginal ratio}.
For $m\in\mathbb{N}$ and $w_{1},\ldots,w_{m}\in\mathbb{R}_{+}$ we
let $\mathcal{P}\big(w_{1},\ldots,w_{m}\big)$ denote the probability
distribution of the random variable $\omega$ on $[m]:=\{1,\ldots,m\}$
such that $\mathbb{P}(\omega=k)\propto w_{k}$.

\begin{algorithm}
\caption{MHAAR for averaging PMR estimators}
\label{alg: MHAAR for Pseudo-marginal ratio}

\KwIn{Current sample $X_{n}=x$}

\KwOut{New sample $X_{n+1}$}

Sample $y\sim q(x,\cdot)$ and $v\sim\mathcal{U}(0,1)$. \\
\If{$v\leq1/2$}{

Sample $u^{(1)},\ldots,u^{(N)}\overset{{\rm iid}}{\sim}Q_{x,y}(\cdot)$
and $k\sim\mathcal{P}\big(\mathring{r}_{u^{(1)}}\big(x,y\big),\ldots,\mathring{r}_{u^{(N)}}\big(x,y\big)\big)$.
\\
Set $X_{n+1}=y$ with probability $\min\{1,\mathring{r}_{\mathfrak{u}}^{N}(x,y)\}$,
otherwise set $X_{n+1}=x$.

}\Else{

Sample $k\sim\mathcal{U}\big\{1,\ldots,N\big\}$, $u^{(k)}\sim\bar{Q}_{x,y}(\cdot)$
and $u^{(1)},\ldots,u^{(k-1)},u^{(k+1)},\ldots,u^{(N)}\overset{{\rm iid}}{\sim}Q_{y,x}(\cdot)$\\
Set $X_{n+1}=y$ with probability $\min\{1,1/\mathring{r}_{\mathfrak{u}}^{N}(y,x)\}$,
otherwise set $X_{n+1}=x$.

}
\end{algorithm}

The unusual step in this update is the random choice between two sampling
mechanisms for the auxiliary variables $u^{(1)},\ldots,u^{(N)}$ and
the fact that depending on this choice either $\mathring{r}_{\mathfrak{u}}^{N}(x,y)$
or $1/\mathring{r}_{\mathfrak{u}}^{N}(y,x)$ is used. Apart from the
reversible jump MCMC context \citep{Green_1995} and specific uses
\citet{tjelmeland-eidsvik-2004,andrieu2008tutorial}, this type of
asymmetric updates has rarely been used \textendash see Appendix \ref{subsec: Generalisation and suboptimality}
for an extensive discussion and from Section \ref{sec: Pseudo-marginal ratio algorithms for latent variable models}
on for other applications. The probability distributions corresponding
to the two proposal mechanisms in Algorithm \ref{alg: MHAAR for Pseudo-marginal ratio}
are given by 
\begin{align*}
Q_{1}^{N}\big(x,{\rm d}(y,\mathfrak{u},k)\big) & :=q(x,{\rm d}y)\prod_{i=1}^{N}Q_{x,y}({\rm d}u^{(i)})\frac{\mathring{r}_{u^{(k)}}(x,y)}{\sum_{i=1}^{N}\mathring{r}_{u^{(i)}}(x,y)},\\
Q_{2}^{N}\big(x,{\rm d}(y,\mathfrak{u},k)\big) & :=q(x,{\rm d}y)\frac{1}{N}\bar{Q}_{x,y}({\rm d}u^{(k)})\prod_{i=1,i\neq k}^{N}Q_{y,x}(\mathrm{d}u^{(i)}),
\end{align*}
and the corresponding Markov transition kernel by 
\begin{align}
 & \mathring{P}^{N}(x,{\rm d}y):=\frac{1}{2}\left[\int_{\mathfrak{U}\times[N]}Q_{1}^{N}\big(x,{\rm d}(y,\mathfrak{u},k)\big)\min\left\{ 1,\mathring{r}_{\mathfrak{u}}^{N}(x,y)\right\} +\mathring{\rho}_{1}(x)\delta_{x}({\rm d}y)\right]\nonumber \\
 & \quad\quad\quad\quad\quad\quad\quad+\frac{1}{2}\left[\int_{\mathfrak{U}\times[N]}Q_{2}^{N}\big(x,{\rm d}(y,\mathfrak{u},k)\big)\min\left\{ 1,1/\mathring{r}_{\mathfrak{u}}^{N}(y,x)\right\} +\mathring{\rho}_{2}(x)\delta_{x}({\rm d}y)\right],\label{eq:PcircleN}
\end{align}
where $\mathring{\rho}_{1}(x)$ and $\mathring{\rho}_{2}(x)$ are
the rejection probabilities for each sampling mechanism. We establish
the $\pi-$reversibility of $\mathring{P}^{N}$ in Theorem \ref{thm:theoreticaljustification}. 
\begin{rem}
\label{rem:samplingkcanbeomitted}It is necessary to include the variable
$k$ in $Q_{1}^{N}$ and $Q_{2}^{N}$ to obtain tractable acceptance
ratios validating the algorithm but, practically, its value is clearly
redundant in Algorithm \ref{alg: MHAAR for Pseudo-marginal ratio}
and sampling $k$ is therefore not required. 
\end{rem}
\begin{rem}
\label{rem: MHAAR with N =00003D 1}$Q_{1}^{N}$ and $Q_{2}^{N}$
reduce to $Q_{1}$ and $Q_{2}$ in \eqref{eq: PMR Q1} and \eqref{eq: PMR Q2}
when $N=1$ in which case $k$ becomes redundant. This implies generality
over PMR algorithms even for $N=1$ (although probably not a useful
one), in the sense that in MHAAR one can also propose from $Q_{2}$. 
\end{rem}
\begin{example}[\textbf{Example \ref{ex: pseudo-marginal for doubly intractable models},
ctd}]
\label{ex:doublyintractaveraging} As noticed in \citet{Nicholls_et_al_2012},
the exchange algorithm \citep{Murray_et_al_2006} can be recast as
a PMR algorithm of the form $\mathring{P}$ given in \eqref{eq:defPring}
where $x=\theta,y=\theta',u=\mathfrak{z}$, $\varphi(u)=u$ and $Q_{x,y}$
corresponds to $\ell_{\theta'}$. Hence an extension of this algorithm
using an averaged acceptance ratio estimator is given by Algorithm
\ref{alg: MHAAR for Pseudo-marginal ratio}. Taking into account Remark
\ref{rem:samplingkcanbeomitted}, this takes the following form. Sample
$\theta'\sim q(\theta,\cdot)$, then with probability $1/2$ sample
$u^{(1)},\ldots,u^{(N)}\overset{{\rm iid}}{\sim}\ell_{\theta'}(\cdot)$
and compute 
\[
\mathring{r}_{\mathfrak{u}}^{N}(\theta,\theta')=\frac{q(\theta',\theta)}{q(\theta,\theta')}\frac{\eta(\theta')}{\eta(\theta)}\frac{g_{\theta'}(\mathfrak{y})}{g_{\theta}(\mathfrak{y})}\frac{1}{N}\sum_{i=1}^{N}\frac{g_{\theta}(u^{(i)})}{g_{\theta'}(u^{(i)})},
\]
or (i.e.\ with probability 1/2) sample $u^{(1)}\sim\ell_{\theta'}(\cdot)$
and $u^{(2)},\ldots,u^{(N)}\overset{{\rm iid}}{\sim}\ell_{\theta}(\cdot)$,
and compute $\mathring{r}_{\mathfrak{u}}^{N}(\theta',\theta)$. This
algorithm was implemented for an Ising model (see details in Section
\ref{subsec: Numerical example: the Ising model}) and numerical simulations
are presented in Figure \ref{fig: Potts exchange with bridging vs asymmetric MCMC}
where the IAC of $f(\theta)=\theta$ is reported as a function of
$N$ (red/grey colour). As anticipated, increasing $N$ improves performance. 
\end{example}

\subsubsection{Theoretical results on validity and performance of MHAAR}

The following theorem justifies the theoretical usefulness of Algorithm
\ref{alg: MHAAR for Pseudo-marginal ratio}. The result follows from
Hilbert space techniques and the recent realisation that the convex
order plays an important rôle in the characterisation of MH updates
based on estimated acceptance ratios \citep{Andrieu_and_Vihola_2014}.
We consider standard performance measures associated to a Markov transition
probability $\Pi$ of invariant distribution $\mu$ defined on some
measurable space $\big(\mathsf{E},\mathcal{E}\big)$. Let $L^{2}(\mathsf{E},\mu):=\big\{ f\colon\mathsf{E}\rightarrow\mathbb{R},\mathsf{var}_{\mu}(f)<\infty\big\}$
and $L_{0}^{2}(\mathsf{E},\mu):=L^{2}(\mathsf{E},\mu)\cap\{f\colon\mathsf{E}\rightarrow\mathbb{R},\mathbb{E}_{\mu}(f)=0\}$.
For any $f\in L^{2}(\mathsf{E},\mu)$ the asymptotic variance is defined
as 
\[
\mathsf{var}(f,\Pi):=\lim_{M\rightarrow\infty}\mathsf{var}_{\mu}\left(M^{-1/2}{\textstyle \sum}_{i=1}^{M}f(X_{i})\right),
\]
which is guaranteed to exist for reversible Markov chains (although
it may be infinite) and for a $\mu-$reversible kernel $\Pi$ its
right spectral gap 
\[
{\rm Gap}_{R}\left(\Pi\right):=\inf\{\mathcal{E}_{\Pi}(f)\,:\,f\in L_{0}^{2}(\mathsf{E},\mu),\,{\rm var}_{\mu}(f)=1\},
\]
where for any $f\in L^{2}\big(\mathsf{E},\mu\big)$ $\mathcal{E}_{\Pi}(f):=\frac{1}{2}\int_{\mathsf{E}}\mu\big({\rm d}x\big)\Pi\big(x,{\rm d}y\big)\big[f(x)-f(y)\big]^{2}$
is the so-called Dirichlet form. The right spectral gap is particularly
useful in the situation where $\Pi$ is a positive operator, in which
case ${\rm Gap}_{R}\left(\Pi\right)$ is related to the geometric
rate of convergence of the Markov chain. 
\begin{thm}
\label{thm:theoreticaljustification}With $P$ and $\mathring{P}^{N}$
as defined in \eqref{eq: MH transition kernel} and\eqref{eq:PcircleN},
respectively, 

\begin{enumerate}
\item For any $N\geq1$ $\mathring{P}^{N}$ is $\pi-$reversible, 
\item For all $N$, ${\rm Gap}_{R}(\mathring{P}^{N})\leq{\rm Gap}_{R}(P)$
and $N\mapsto{\rm Gap}_{R}(\mathring{P}^{N})$ is non decreasing, 
\item For any $f\in L^{2}(\mathsf{X},\pi)$,

\begin{enumerate}
\item $N\mapsto\mathcal{E}_{\mathring{P}^{N}}(f)$ (or equivalently first
order auto-covariance coefficient) is non decreasing (non increasing), 
\item $N\mapsto\mathsf{var}(f,\mathring{P}^{N})$ is non increasing, 
\item for all $N$, $\mathsf{var}(f,\mathring{P}^{N})\geq\mathsf{var}(f,P)$. 
\end{enumerate}
\end{enumerate}
\end{thm}
\begin{proof}
The reversibility follows from the fact that this Markov transition
kernel fits in the framework of asymmetric MH updates described in
Theorem \ref{thm:asymmetricMH-1} in Appendix \ref{sec: A general framework for MPR and MHAAR algorithms}
after checking that for any $x,y,\mathfrak{u}\in\mathring{\mathsf{S}}^{N}$,

\begin{align}
\frac{\pi({\rm d}y)Q_{2}^{N}\big(y;{\rm d}(x,\mathfrak{u},k)\big)}{\pi({\rm d}x)Q_{1}^{N}\big(x,{\rm d}(y,\mathfrak{u},k)\big)} & =\frac{\pi({\rm d}y)q(y,{\rm d}x)\frac{1}{N}\bar{Q}_{y,x}({\rm d}u^{(k)})\prod_{i\neq k}Q_{x,y}({\rm d}u^{(i)})}{\pi({\rm d}x)q(x,{\rm d}y)Q_{x,y}({\rm d}u^{(k)})\prod_{i\neq k}Q_{x,y}({\rm d}u^{(i)})\times\frac{\mathring{r}_{u^{(k)}}(x,y)}{\sum_{i=1}^{N}\mathring{r}_{u^{(i)}}(x,y)}}\nonumber \\
 & =\mathring{r}_{u^{(k)}}(x,y)\frac{1/N}{\frac{\mathring{r}_{u^{(k)}}(x,y)}{N\times\mathring{r}_{\mathfrak{u}}^{N}(x,y)}}=\mathring{r}_{\mathfrak{u}}^{N}(x,y).\label{eq:eq:simplifiedacceptPcircN}
\end{align}
For the other statements we first start by noticing that the expression
for the Dirichlet form associated with $\mathring{P}^{N}$ can be
rewritten in either of the following simplified forms 
\begin{align*}
\mathcal{E}_{\mathring{P}^{N}} & =\frac{1}{2}\int\pi({\rm d}x)\int_{\mathsf{\mathfrak{U}\times[N]}}Q_{1}^{N}\big(x,{\rm d}(y,\mathfrak{u},k)\big)\min\{1,\mathring{r}_{\mathfrak{u}}^{N}(x,y)\}\left(f(x)-f(y)\right)^{2}\\
 & =\frac{1}{2}\int\pi({\rm d}x)\int_{\mathfrak{U}\times[N]}Q_{2}^{N}\big(x,{\rm d}(y,\mathfrak{u},k)\big)\min\{1,1/\mathring{r}_{\mathfrak{u}}^{N}(y,x)\}\left(f(x)-f(y)\right)^{2}.
\end{align*}
This follows from the identities established in \eqref{eq:PcircleN}
and \eqref{eq:eq:simplifiedacceptPcircN}. The expression on the first
line turns out to be particularly convenient. A well known result
from the convex order literature states that for any $n\geq2$ exchangeable
random variables $Z_{1},\ldots,Z_{n}$ and any convex function $\phi$
we have $\mathbb{E}\left[\phi\left(n^{-1}\sum_{i=1}^{n}Z_{i}\right)\right]\leq\mathbb{E}\left[\phi\left((n-1)^{-1}\sum_{i=1}^{n-1}Z_{i}\right)\right]$
whenever the expectations exist \citep[Corollary 1.5.24]{mullercomparison}.
The two sums are said to be convex ordered. Now since $a\mapsto-\min\{1,a\}$
is convex we deduce that for any $N\geq1$, $x,y\in\mathsf{X}$, 
\begin{equation}
\int_{\mathsf{U}^{N}}Q_{x,y}^{N}({\rm d}\mathfrak{u})\min\{1,\mathring{r}_{\mathfrak{u}}^{N}(x,y)\}\leq\int_{\mathsf{U}^{N+1}}Q_{x,y}^{N+1}({\rm d}\mathfrak{u})\min\{1,\mathring{r}_{\mathfrak{\mathfrak{u}}}^{N+1}(x,y)\}\label{eq:convexorderingratio}
\end{equation}
where $Q_{x,y}^{N}(\mathrm{d}\mathfrak{u}):=\prod_{i=1}^{N}Q_{x,y}(\mathrm{d}u^{(i)})$,
and consequently for any $f\in L^{2}(\mathsf{X},\pi)$ and $N\geq1$
\[
\mathcal{E}_{\mathring{P}^{N+1}}(f)\leq\mathcal{E}_{\mathring{P}^{N}}(f).
\]
All the monotonicity properties follow from \citet{Tierney_1998}
since $\mathring{P}^{N}$ and $\mathring{P}^{N+1}$ are $\pi-$reversible.
The comparisons to $P$ follow from the application of Jensen's inequality
to $a\mapsto\min\{1,a\}$, which leads for any $x,y\in\mathsf{X}$
to 
\[
\int_{\mathfrak{U}}Q_{x,y}^{N}({\rm d}\mathfrak{u})\min\{1,\mathring{r}_{\mathfrak{u}}^{N}(x,y)\}\leq\min\{1,r\big(x,y\big)\},
\]
and again using the results of \citet{Tierney_1998}. 
\end{proof}
This result motivates the practical usefulness of the algorithm, in
particular in a parallel computing environment. Indeed, one crucial
property of Algorithm \ref{alg: MHAAR for Pseudo-marginal ratio}
is that in both moves $Q_{1}^{N}(\cdot)$ and $Q_{2}^{N}(\cdot)$,
sampling of $u^{(1)},\ldots,u^{(N)}$ and computation of $\mathring{r}_{u^{(1)}}(x,y),\ldots,\mathring{r}_{u^{(N)}}(x,y)$
can be performed in a parallel fashion and offers the possibility
to reduce the variance $\mathsf{var}(f,\mathring{P})$ of estimators,
but more importantly the burn-in period of algorithms. Indeed one
could object that running $M\in\mathbb{N}^{+}$ independent chains
in parallel with $N=1$ and combining their averages, instead of using
the output from a single chain with $N=M$ would achieve variance
reduction. However our point is that the former does not speed up
convergence to equilibrium, while the latter will, in general. Unfortunately,
while estimating the asymptotic variance $\mathsf{var}(f,\mathring{P}^{N})$
from simulations is achievable, estimating time to convergence to
equilibrium is far from standard in general. The following toy example
is an exception and illustrates our point.
\begin{example}
Here we let $\pi$ be the uniform distribution on $\mathsf{X}=\{-1,1\}$,
$\mathsf{U}=\{a,a^{-1}\}$ for $a>0$, $Q_{x,-x}(u=a)=1/(1+a)$, $Q_{x,-x}(u=1/a)=a/(1+a)$
and $\varphi(u)=1/u$. In other words $\mathring{P}$ can be reparametrized
in terms of $a$ and with the choice $q(x,-x)=1-\theta$ for $(\theta,x)\in[0,1)\times\mathsf{X}$
we obtain
\begin{align*}
\mathring{P}(x,-x) & =(1-\theta)\left[\frac{1}{1+a}\min\big\{1,a\big\}+\frac{a}{1+a}\min\big\{1,a^{-1}\big\}\right].
\end{align*}
Note that there is no need to be more specific than say $Q_{x,x}(u)>0$
for $x,u\in\mathsf{X}\times\mathsf{U}$ as then a proposed ``stay''
is always accepted. This suggests that we are in fact drawing the
acceptance ratio, and corresponds to \citep[Example 8 in][]{Andrieu_and_Vihola_2014}
of their abstract parametrisation of PMR algorithms. Now for $N\geq2$
and $x\in\mathsf{X}$ we have
\begin{align*}
\mathring{P}^{N}(x,-x) & =\frac{1-\theta}{2}\left[\sum_{k=0}^{N}\beta^{N}(k)\min\big\{1,w_{k}(N)\big\}\right.\\
 & \hspace{1.5cm}+\left.\sum_{k=0}^{N}\left(\frac{a}{1+a}\beta^{N-1}(k-1)+\frac{1}{1+a}\beta^{N-1}(k)\right)\min\big\{1,w_{k}^{-1}(N)\big\}\right],
\end{align*}
where $\beta^{N}(k)$ is the probability mass function of the binomial
distribution of parameters $N$ and $1/(1+a)$ and $w_{k}(N):=ka/N+\big(1-k/N\big)a^{-1}.$The
second largest eigenvalue of the corresponding Markov transition matrix
is $\lambda_{2}(N)=1-2\mathring{P}^{N}(x,-x)$ from which we find
the relaxation time $T_{{\rm relax}}(N):=1/\big(2\mathring{P}^{N}(x,-x)\big)$,
and bounds on the mixing time $T_{{\rm mix}}(\epsilon,N)$, that is
the number of iterations required for the Markov chain to have marginal
distribution within $\epsilon$ of $\pi$, in the total variation
distance, \citet[Theorem 12.3 and Theorem 12.4]{levin2017markov}
\[
-(T_{{\rm relax}}(N)-1)\log(2\epsilon)\leq T_{{\rm mix}}(\epsilon,N)\leq-T_{{\rm relax}}(N)\log(\epsilon/2).
\]
We define the time reduction, $\gamma(N):=T_{{\rm relax}}(N)/T_{{\rm relax}}(1)$,
which is independent of $\theta$ and captures the benefit of MHAAR
in terms of convergence to equilibrium. In Fig. \ref{fig:toy-example-relaxation}
we present the evolution of $N\mapsto\gamma(N)$ for $a=2,5,10$ and
$\gamma(1000)$ as a function of $a$. As expected the worse the algorithm
corresponding to $\mathring{P}$ is, the more beneficial averaging
is: for $a=2,5,10$ we observe running time reductions of approximately
$35\%$, $65\%$ and $80\%$ respectively. This suggests that computationally
cheap, but possibly highly variable, estimators of the acceptance
ratio may be preferable to reduce burn-in when a parallel machine
is available and communication costs are negli-geable.

\begin{figure}
\includegraphics[width=0.3\textheight]{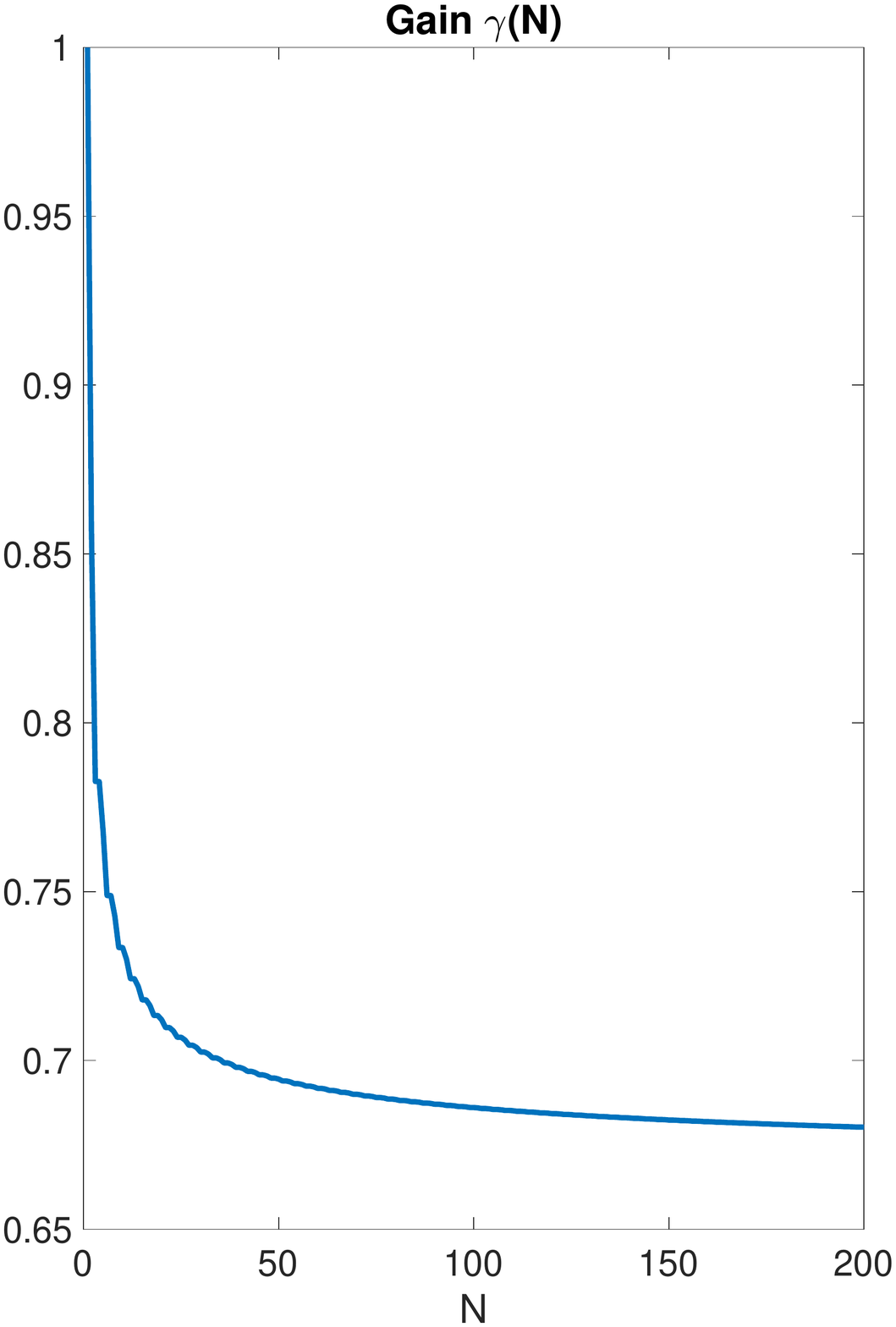}\hfill{}\includegraphics[width=0.3\textheight]{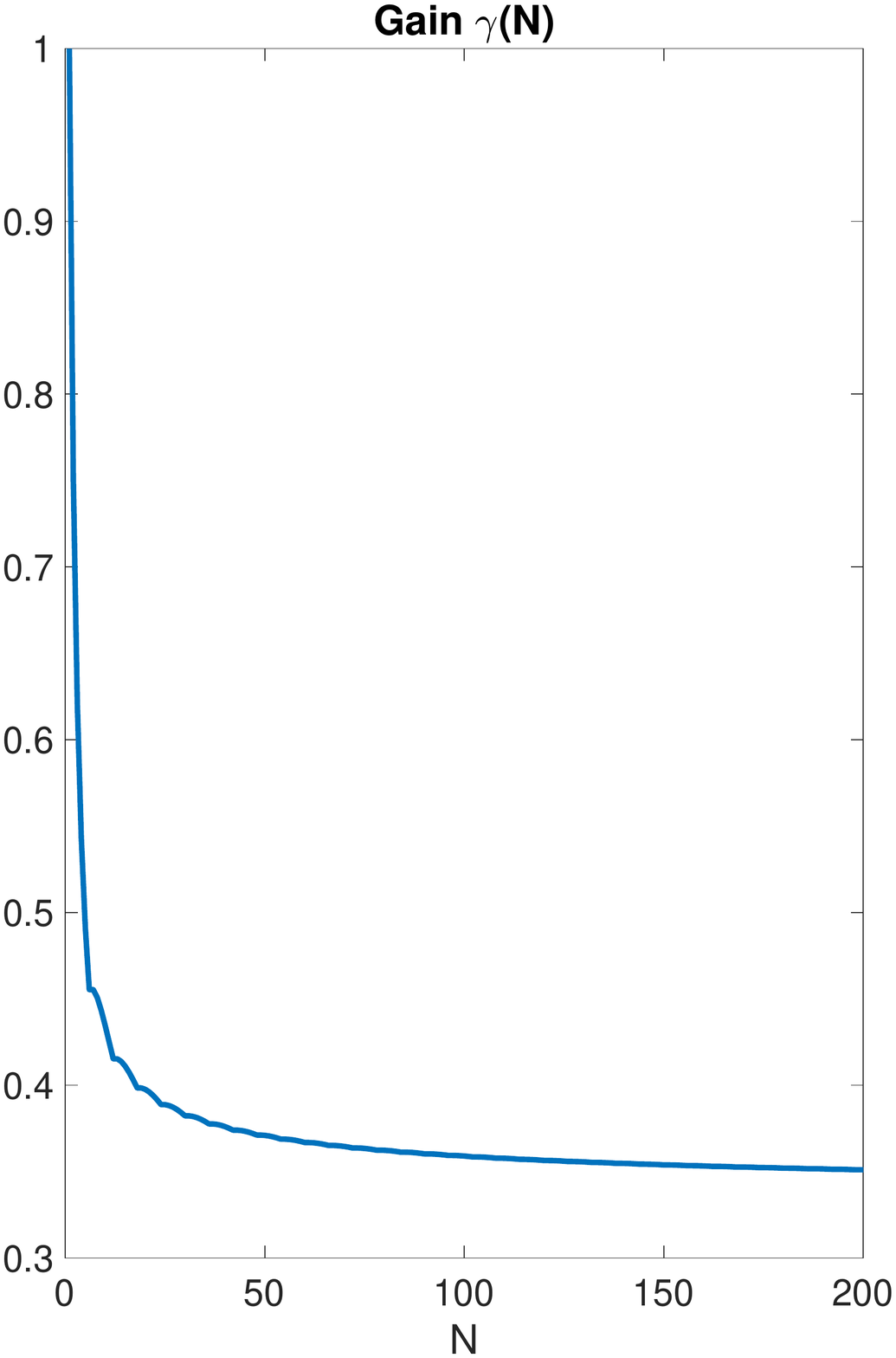}

\vspace{0cm}

\includegraphics[width=0.3\textheight]{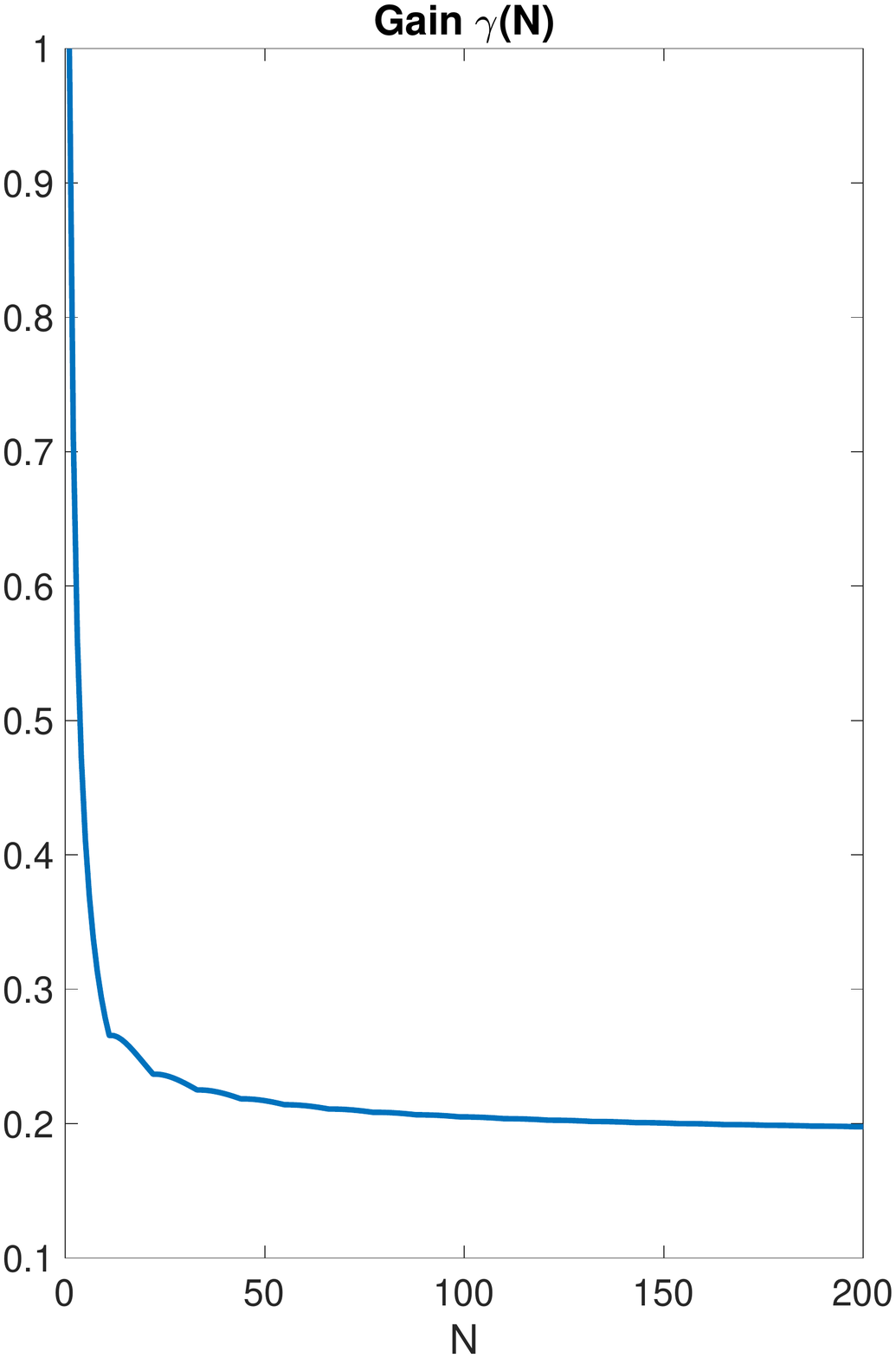}\hfill{}\includegraphics[width=0.3\textheight]{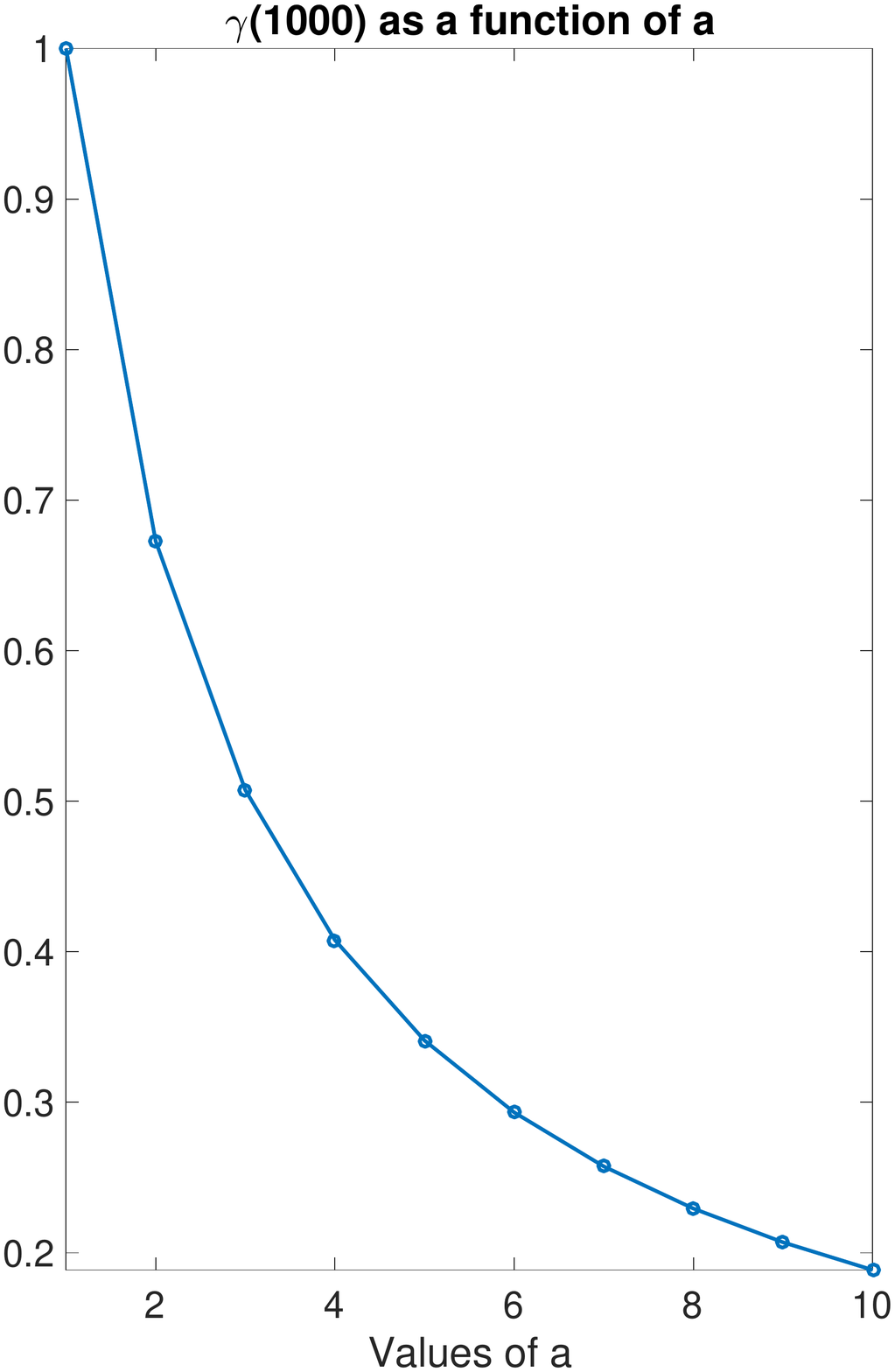}

\caption{\label{fig:toy-example-relaxation}Top left: $a=2$, Top right: $a=5$,
Bottom left: $a=10$, Bottom right: evolution of $\gamma(1000)$ as
a function of $a$}

\end{figure}
\end{example}

\subsubsection{Introducing dependence \label{subsec: Introducing dependence}}

The following discussion on possible extensions can be omitted on
a first reading. There are numerous possible variations around the
basic algorithm presented above. A practically important extension
related to the order in which variables are drawn is discussed in
Section \ref{sec: Pseudo-marginal ratio algorithms for latent variable models}
in the general context of latent variable models. There is another
possible extension worth mentioning here. Close inspection of the
proof of $\pi-$reversibility of $\mathring{P}^{N}$ in Theorem \ref{thm:theoreticaljustification}
suggests that conditional independence of $u^{(1)},\ldots,u^{(N)}$
is not a requirement. Define $\mathfrak{u}^{(-k)}:=\left(u^{(1)},\ldots,u^{(k-1)},u^{(k+1)},\ldots,u^{(N)}\right)$.
\begin{thm}
\label{thm:generalisationexchangeable}Let $N\geq1$ and for any $x,y\in\mathsf{X}$
let $Q_{x,y}^{N}({\rm d}\mathfrak{u})$ be a probability distribution
on $\big(\mathfrak{U},\mathcal{U}^{\otimes N}\big)$ such that all
its marginals are identical and equal to \textup{$Q_{x,y}(\cdot)$.
Assume that} 
\begin{align*}
Q_{1}^{N}\big(x,{\rm d}(y,\mathfrak{u},k)\big) & :=q(x,{\rm d}y)Q_{x,y}^{N}({\rm d}\mathfrak{u})\frac{\mathring{r}_{u^{(k)}}(x,y)}{\sum_{i=1}^{N}\mathring{r}_{u^{(i)}}(x,y)},\\
Q_{2}^{N}\big(x,{\rm d}(y,\mathfrak{u},k)\big) & :=q(x,{\rm d}y)\frac{1}{N}\bar{Q}_{x,y}({\rm d}u^{(k)})Q_{y,x}^{N}({\rm d}\mathfrak{u}^{(-k)}\mid u^{(k)}).
\end{align*}
Then $\mathring{P}^{N}$ with acceptance ratio $\mathring{r}_{\mathfrak{u}}^{N}(x,y)$
as in \eqref{eq: average acceptance ratio} is $\pi-$reversible.
Further, if $u^{(1)},\ldots,u^{(N)}$ are exchangeable with respect
to $Q_{x,y}^{N}({\rm d}\mathfrak{u})$ then all the comparison results
in Theorem \ref{thm:theoreticaljustification} still hold.
\end{thm}
\begin{proof}
One can check that 
\begin{align*}
\frac{\pi({\rm d}y)Q_{2}^{N}\big(y,{\rm d}(x,\mathfrak{u},k)\big)}{\pi({\rm d}x)Q_{1}^{N}\big(x,{\rm d}(y,\mathfrak{u},k)\big)} & =\frac{\pi({\rm d}y)q(y,{\rm d}x)\frac{1}{N}\bar{Q}_{y,x}({\rm d}u^{(k)})Q_{x,y}^{N}({\rm d}\mathfrak{u}^{(-k)}\mid u^{(k)})}{\pi({\rm d}x)q(x,{\rm d}y)Q_{x,y}({\rm d}u^{(k)})Q_{x,y}^{N}({\rm d}\mathfrak{u}^{(-k)}\mid u^{(k)})\frac{\mathring{r}_{u^{(k)}}(x,y)}{\sum_{i=1}^{N}\mathring{r}_{u^{(i)}}(x,y)}}\\
 & =\mathring{r}_{u^{(k)}}(x,y)\frac{1/N}{\frac{\mathring{r}_{u^{(k)}}(x,y)}{N\times\mathring{r}_{\mathfrak{u}}^{N}(x,y)}}=\mathring{r}_{\mathfrak{u}}^{N}(x,y),
\end{align*}
which remains the same as in \eqref{eq: average acceptance ratio}.
The exchangeability assumption ensures that \eqref{eq:convexorderingratio}
holds. 
\end{proof}
\begin{example}
The following is a short discussion of a scenario which may be relevant
in practice. Assume that it is possible to sample $u^{\left(1\right)}$
from $Q_{x,y}(\cdot)$ but that this is computationally expensive,
as is the case for sampling exactly from Markov random fields such
as the Ising model. One could suggest sampling the remaining samples
$\mathfrak{u}^{(-1)}$ as defined in $Q_{1}^{N}(\cdot,\cdot)$ using
a $Q_{x,y}-$reversible Markov transition probability $K_{x,y}$ (and
similarly for $Q_{y,x}(\cdot)$ in $Q_{2}^{N}(\cdot,\cdot)$ using
$K_{y,x}$), which will in general be far less expensive. Here $Q_{1}^{N}(\cdot,\cdot)$
corresponds to sampling 
\[
u^{(1:N)}\sim Q_{x,y}({\rm d}u^{(1)})K_{x,y}(u^{(1)},{\rm d}u^{(2)})\ldots K_{x,y}(u^{(N-1)},{\rm d}u^{(N)}).
\]
In order to describe sampling in $Q_{2}^{N}(\cdot,\cdot)$, we first
establish a convenient expression for $Q_{x,y}^{N}({\rm d}\mathfrak{u}^{(-k)}\mid u^{(k)})$
for $x,y\in\mathsf{X}^{2}$ and $k=1,\ldots,N$. By reversibility
of $K_{x,y}$, we have for $k=1,\ldots,N$ (with straightforward conventions
for $k\in[N]$)
\begin{align*}
Q_{x,y}({\rm d}u^{(1)})\prod_{i=2}^{N}K_{x,y}(u^{(i-1)},{\rm d}u^{(i)})=Q_{x,y}({\rm d}u^{(k)})\prod_{i=2}^{k}K_{x,y}(u^{(i)},{\rm d}u^{(i-1)})\prod_{i=k+1}^{N}K_{x,y}(u^{(i-1)},{\rm d}u^{(i)})
\end{align*}
from which one obtains the desired conditional, and deduces that sampling
the auxiliary variables in $Q_{2}^{N}(\cdot,\cdot)$ consists of sampling
$k\sim\mathcal{U}\{1,2,\ldots,N\}$, $u^{(k)}\sim\bar{Q}_{x,y}(\cdot)$,
and then simulate the rest of the chain ``forward'' and ``backward''
as follows
\[
(u^{(k-1)},\ldots,u^{(1)})\sim\prod_{i=2}^{k}K_{y,x}(u^{(i)},{\rm d}u^{(i-1)}),\quad(u^{(k+1)},\ldots,u^{(N)})\sim\prod_{i=k+1}^{N}K_{y,x}(u^{(i-1)},{\rm d}u^{(i)}).
\]
Note that in this case, Remark \ref{rem:samplingkcanbeomitted} does
not hold. While sampling $k$ is still not necessary in $Q_{1}^{N}(\cdot,\cdot)$,
sampling $k$ in $Q_{2}^{N}(\cdot,\cdot)$ is required. The last part
of the theorem is applicable by averaging over the set of permutations
of $[N]$ 
\[
Q_{x,y}^{N}\big({\rm d}\mathfrak{u}\big)=\frac{1}{N!}\sum_{\sigma\in\mathfrak{S}}Q_{x,y}({\rm d}u^{(\sigma(1))})K_{x,y}(u^{(\sigma(1))},{\rm d}u^{(\sigma(2))})\ldots K_{x,y}(u^{(\sigma(N-1))},{\rm d}u^{(\sigma(N))}),
\]
and noting that for $k\in[N]$ and by using the reversibility as above
for each $\sigma\in\mathfrak{S}$ leads to
\[
Q_{x,y}^{N}\big({\rm d}\mathfrak{u}\big)=Q_{x,y}({\rm d}u^{(k)})\frac{1}{N!}\sum_{\sigma\in\mathfrak{S}}\prod_{i=2}^{\sigma^{-1}(k)}K_{x,y}(u^{(\sigma(i))},{\rm d}u^{(\sigma(i-1))})\prod_{i=\sigma^{-1}(k)+1}^{N}K_{x,y}(u^{(\sigma(i-1))},{\rm d}u^{(\sigma(i))}).
\]
We do not investigate this algorithm further here.
\end{example}

\section{Improving PMR algorithms with AIS \label{sec: Improving pseudo-marginal ratio algorithms for doubly intractable models}}

Before moving on to more complex scenarios in Section \ref{sec: Pseudo-marginal ratio algorithms for latent variable models},
we focus in this section on the averaging of acceptance ratios in
the specific context of our running Example \ref{ex:doublyintractaveraging}.
The exchange algorithm \citep{Murray_et_al_2006}, described in Example
\ref{ex: exchange algorithm}, exploits the fact that for $\theta,\theta'\in\Theta$
and $u\sim\ell_{\theta'}(\cdot)$, the ratio $g_{\theta}\big(u\big)/g_{\theta'}\big(u\big)$
is an estimator of $C_{\theta}/C_{\theta'}$. Another possible estimator
of $C_{\theta}/C_{\theta'}$, based on AIS \citep{Crooks1998,Neal_2001},
was also used in \citet{Murray_et_al_2006}. It has the advantage
that it involves a tuning parameter which can be used to reduce the
variability of the estimator, and hence improve the theoretical performance
of exchange type algorithms. It has recently been established theoretically
that this approach can beat the curse of dimensionality by reducing
complexity from exponential to polynomial in the problem dimension
\citet{2016arXiv161207583A,beskos2014stability}. This is however
at the expense of an additional computational cost. In this section,
we show that the AIS based exchange algorithm can be reinterpreted
as a PMR algorithm of the form \eqref{eq:defPring}. It is thus straightforward
to extend this methodology through Algorithm \ref{alg: MHAAR for Pseudo-marginal ratio}
so as to use acceptance ratio estimators obtained through averaging. 

\subsection{AIS based exchange algorithm and its average acceptance ratio form
\label{subsec: AIS based exchange algorithm and its average acceptance ratio form}}

The estimator $g_{\theta}\big(u\big)/g_{\theta'}\big(u\big)$ for
$u\sim\ell_{\theta'}(\cdot)$ of the ratio of $C_{\theta}/C_{\theta'}$
may be very variable when the functions $g_{\theta}(\cdot)$ and $g_{\theta'}(\cdot)$
differ too much. The basic idea behind AIS consists of rewriting the
ratio of interest as a telescoping product of ratios of normalising
constants corresponding to a sequence of artificial probability densities
\[
\mathscr{P}_{\theta,\theta',T}:=\big\{\pi_{\theta,\theta',t}(\cdot),t=0,\ldots,T+1\big\}
\]
 for some $T\geq1$ evolving from $\pi_{\theta,\theta',0}(u)=\ell_{\theta'}(u)$
to $\pi_{\theta,\theta',T+1}(u)=\ell_{\theta}(u)$; i.e.\ $\pi_{\theta,\theta',t}(u)=f_{\theta,\theta',t}(u)/C_{\theta,\theta',t}$
where $f_{\theta,\theta',t}(u)$ can be computed pointwise but $C_{\theta,\theta',t}$
is intractable. More precisely one rewrites $C_{\theta}/C_{\theta'}=\prod_{t=0}^{T}C_{\theta,\theta',t+1}/C_{\theta,\theta',t}$
(with $C_{\theta,\theta',0}=C_{\theta'}$ and $C_{\theta,\theta',T+1}=C_{\theta}$)
where the densities $\big\{ f_{\theta,\theta',t}(\cdot),t=1,\ldots,T\big\}$
are such that estimating each term $C_{\theta,\theta',t+1}/C_{\theta,\theta',t}$
can be performed efficiently using the technique above for example.
Good performance therefore necessitates that successive unnormalised
densities are close (and become ever closer as $T$ increases). A
naive implementation would require exact sampling from each of the
intermediate probability distributions but the remarkable fact noticed
independently in \citet{Crooks1998} and \citet{Neal_2001} is that
the estimators involved in the product may arise from an inhomogeneous
Markov chain, therefore rendering the algorithm highly practical.
The following proposition establishes that this algorithm is of the
same form as $\mathring{P}$ given in \eqref{eq:defPring}.
\begin{prop}
\label{prop: AIS MCMC for doubly intractable models}Assume the set-up
of Example \ref{ex:doublyintractable} and for all $\theta,\theta'\in\Theta$,
let 

\begin{enumerate}
\item $\mathscr{F}_{\theta,\theta',T}=\big\{ f_{\theta,\theta',t}(\cdot),t=0,\ldots,T+1\big\}$
be a family of tractable unnormalised densities of $\mathscr{P}_{\theta,\theta',T}$
such that for $t=0,\ldots,T$

\begin{enumerate}
\item $f_{\theta,\theta',t}(\cdot)$ and $f_{\theta,\theta',t+1}(\cdot)$
have the same support, 
\item for any $u\in\mathsf{Y}$ 
\[
f_{\theta,\theta',0}(u)=g_{\theta'}(u),\quad f_{\theta,\theta',T+1}(u)=g_{\theta}(u),\quad f_{\theta,\theta',t}(u)=f_{\theta',\theta,T+1-t}(u),
\]
\end{enumerate}
\item $\mathscr{R}_{\theta,\theta',T}=\big\{ R_{\theta,\theta',t}(\cdot,\cdot)\colon\mathsf{Y}\times\mathcal{Y}\to[0,1],t=1,\ldots,T\big\}$
be a family of Markov transition kernels such that for any $t=1,\ldots,T$

\begin{enumerate}
\item $R_{\theta,\theta',t}(\cdot,\cdot)$ is $\pi_{\theta,\theta',t}-$reversible, 
\item $R_{\theta,\theta',t}(\cdot,\cdot)=R_{\theta',\theta,T+1-t}(\cdot,\cdot)$, 
\end{enumerate}
\item $Q_{\theta,\theta'}(\cdot)$ be the probability distributions $\big(\mathsf{U},\mathcal{U}\big)$,
where $\mathsf{U}:=\mathcal{\mathsf{Y}}^{T+1}$, defined for $u:=(u_{0},\ldots,u_{T})\in\mathsf{U}$
as 
\begin{align}
\;Q_{\theta,\theta'}({\rm d}u):= & \ell_{\theta'}({\rm d}u_{0})\prod_{t=1}^{T}R_{\theta,\theta',t}(u_{t-1},{\rm d}u_{t}),\label{eq: Q for doubly intractable with annealing}
\end{align}
and $\varphi$ the involution reversing the order of the components
of $u$; i.e.\ $\varphi(u_{0},u_{1},\ldots,u_{T}):=(u_{T},u_{T-1},\ldots,u_{0})$
for all $u\in\mathsf{U}$.
\end{enumerate}
\noindent Then for any $\theta,\theta'\in\Theta$ and any $u\in\mathsf{U}$
\[
\bar{Q}_{\theta,\theta'}({\rm d}u)=\ell_{\theta'}({\rm d}u_{T})\prod_{t=1}^{T}R_{\theta',\theta,T-t+1}(u_{T-t+1},{\rm d}u_{T-t}),
\]

\noindent and
\[
\frac{\bar{Q}_{\theta',\theta}({\rm d}u)}{Q_{\theta,\theta'}({\rm d}u)}=\frac{C_{\theta'}}{C_{\theta}}\prod_{t=0}^{T}\frac{f_{\theta,\theta',t+1}(u_{t})}{f_{\theta,\theta',t}(u_{t})}.
\]
The AIS based exchange algorithm of \citet{Murray_et_al_2006} corresponds
to $\mathring{P}$ in \eqref{eq:defPring} with proposal distribution
\[
Q_{1}(\theta,{\rm d}(\theta',u))=q(\theta,{\rm d}\theta')Q_{\theta,\theta'}({\rm d}u)
\]
 and its complementary kernel 
\[
Q_{2}(\theta,{\rm d}(\theta',u))=q(\theta,{\rm d}\theta')\bar{Q}_{\theta,\theta'}({\rm d}u).
\]
Its acceptance ratio on $\mathring{\mathsf{S}}$ is 
\begin{equation}
\mathring{r}_{u}(\theta,\theta')=\frac{q(\theta',\theta)}{q(\theta,\theta')}\frac{\eta(\theta')}{\eta(\theta)}\frac{g_{\theta'}(\mathfrak{y})}{g_{\theta}(\mathfrak{y})}\prod_{t=0}^{T}\frac{f_{\theta,\theta',t+1}(u_{t})}{f_{\theta,\theta',t}(u_{t})}.\label{eq: ratio for doubly intractable with annealing}
\end{equation}

\end{prop}
\begin{proof}
Since $Q_{\theta,\theta'}(\varphi(A))=\bar{Q}_{\theta,\theta'}(A)$,
we can check that the pair $Q_{1}(x,\cdot)$, $Q_{2}(x,\cdot)$ satisfy
the assumption of Theorem \ref{thm: pseudo-marginal ratio algorithms}
in Appendix \ref{sec: A general framework for MPR and MHAAR algorithms}.
Moreover, by the symmetry assumption on $\mathscr{R}_{\theta,\theta',T}$,
we obtain
\begin{align*}
\bar{Q}_{\theta,\theta'}({\rm d}u) & =\ell_{\theta'}({\rm d}u_{T})\prod_{t=1}^{T}R_{\theta,\theta',t}(u_{T-t+1},{\rm d}u_{T-t})\\
 & =\ell_{\theta'}({\rm d}u_{T})\prod_{t=1}^{T}R_{\theta',\theta,T-t+1}(u_{T-t+1},{\rm d}u_{T-t}),
\end{align*}
so we can apply Theorem \ref{thm:AIS} in Appendix \ref{sec: A short justification of AIS and an extension}
with $\mu_{0}=\ell_{\theta'}$, $\mu_{\tau+1}=\ell_{\theta}$, $\tau=T$
and $\mu_{t}=\pi_{\theta,\theta',t}$ and $\Pi_{t}=R_{\theta,\theta',t}$
for $t=1,\ldots,T$ to show that $\bar{Q}_{\theta',\theta}(\cdot)$
is absolutely continuous with respect to $Q_{\theta,\theta'}(\cdot)$
and the expression for the corresponding Radon-Nikodym derivative
ensures that \eqref{eq:accept-ratio-circle} is indeed equal to \eqref{eq: ratio for doubly intractable with annealing}. 
\end{proof}
By selecting an appropriate sequence of intermediate distributions
$\mathscr{P}_{\theta,\theta',T}$ as detailed in Section \ref{subsec: Numerical example: the Ising model},
the variability of this noisy acceptance ratio can be reduced by increasing
$T$. Another approach to reduce variability is given in Algorithm
\ref{alg: MHAAR-AIS exchange algorithm} which consists of averaging
acceptance ratios as described in Algorithm \ref{alg: MHAAR for Pseudo-marginal ratio}.
For $T=0$ and $N>1$ Algorithm \ref{alg: MHAAR-AIS exchange algorithm}
reduces to that in Example \ref{ex:doublyintractaveraging}, for $N=1$
and $T>0$, we recover the exchange algorithm with bridging of \citet{Murray_et_al_2006}
and for $T=0$ and $N=1$, this reduces to the exchange algorithm.
Our generalisation presents a clear computational interest: while
sampling a realisation of the Markov chain defined by $Q_{\theta,\theta'}(\cdot)$
is fundamentally a serial operation, sampling $N$ independent such
realisations is trivially parallelisable. On an ideal parallel computer,
running the algorithm for any $N>1$ or $N=1$ would take the same
amount of the user's time. We explore numerically combinations of
the parameters $T$ and $N$ in Section \ref{subsec: Numerical example: the Ising model}.

\begin{algorithm}
\caption{MHAR for averaged AIS PMR estimators}

\label{alg: MHAAR-AIS exchange algorithm} 

\KwIn{Current state $\theta_{n}=\theta$.} 

\KwOut{Next sample $\theta_{n+1}$} 

Sample $\theta'\sim q(\theta,\cdot)$ and $v\sim\mathcal{U}(0,1)$.
\\
\If{$v\leq1/2$}{ 

\For{$i=1,\ldots,N$}{

Sample $u_{0}^{(i)}\sim\ell_{\theta'}(\cdot)$ and $u_{t}^{(i)}\sim R_{\theta,\theta',t}(u_{t-1}^{(i)},\cdot)$,
$t=1,\ldots,T$.

} 

Set $\theta_{n+1}=\theta'$ with probability $\min\{1,\mathring{r}_{\mathfrak{u}}^{N}(\theta,\theta')\}$
(see \eqref{eq: ratio for doubly intractable with annealing}), otherwise
set $\theta_{n+1}=\theta$.

}\Else{ 

Sample $u_{T}^{(1)}\sim\ell_{\theta'}(\cdot)$ and $u_{t-1}^{(1)}\sim R_{\theta,\theta',t}(u_{t}^{(1)},\cdot)$,
$t=T,\ldots,1$,\\
\For{$i=2,\ldots,N$,}{

Sample $u_{0}^{(i)}\sim\ell_{\theta}(\cdot)$ and $u_{t}^{(i)}\sim R_{\theta',\theta,t}(u_{t-1}^{(i)},\cdot)$,
$t=1,\ldots,T$.

} 

Set $\theta_{n+1}=\theta'$ with probability $\min\{1,1/\mathring{r}_{\mathfrak{u}}^{N}(\theta',\theta)\}$,
otherwise set $\theta_{n+1}=\theta$.

}
\end{algorithm}

\subsection{Using a single sample from $\ell_{\theta'}(\cdot)$ per iteration\label{subsec: Using a single sample from ell_theta_prime per iteration} }

This section can be omitted on a first reading. In Algorithm \ref{alg: MHAAR-AIS exchange algorithm},
each of the $N$ chains has a different initial point, which is a
sample from an intractable distribution. Obtaining such a sample can
be computationally expensive. Algorithm \ref{alg: MHAAR-AIS exchange algorithm - reduced computation}
is an alternative that only requires one such sample at each iteration.
The proof that the associated Markov kernel is $\pi$-reversible can
be derived from Theorem \ref{thm:generalisationexchangeable} in Section
\ref{subsec: Introducing dependence}, hence we omit it.

Although computationally more expensive on a serial machine, we expect
Algorithm \ref{alg: MHAAR-AIS exchange algorithm} to have better
statistical properties than Algorithm \ref{alg: MHAAR-AIS exchange algorithm - reduced computation}
as it uses independent chains to estimate the acceptance ratio. This
is demonstrated experimentally in Section \ref{subsec: Numerical example: the Ising model}.
Moreover, the computational advantage of Algorithm \ref{alg: MHAAR-AIS exchange algorithm - reduced computation}
is questionable on a parallel architecture, where one can in principle
run all the chains in $Q_{1}^{N}(\cdot,\cdot)$ and $Q_{2}^{N}(\cdot,\cdot)$
of Algorithm \ref{alg: MHAAR-AIS exchange algorithm} at the same
time. In fact, Algorithm \ref{alg: MHAAR-AIS exchange algorithm}
may be even faster since all the chains in the backward move can be
produced in parallel whereas this can not be done in Algorithm \ref{alg: MHAAR-AIS exchange algorithm - reduced computation}.

\begin{algorithm}
\caption{MHAR for averaged AIS PMR estimators reduced computation}

\label{alg: MHAAR-AIS exchange algorithm - reduced computation} 

\KwIn{Current sample $\theta_{n}=\theta$.}

\KwOut{New sample $\theta_{n+1}$} 

Sample $\theta'\sim q(\theta,\cdot)$ and $v\sim\mathcal{U}(0,1)$.
\\
\If{$v\leq1/2$}{

Sample $u_{0}\sim\ell_{\theta'}(\cdot)$.\\
\For{$i=1,\ldots,N$}{ 

Set $u_{0}^{(i)}=u_{0}$ and sample $u_{t}^{(i)}\sim R_{\theta,\theta',t}(u_{t-1}^{(i)},\cdot)$,
$t=1,\ldots,T$.

} 

Set $\theta_{n+1}=\theta'$ with probability $\min\{1,\mathring{r}_{\mathfrak{u}}^{N}(\theta,\theta')\}$,
otherwise set $\theta_{n+1}=\theta$.

}\Else{ 

Sample $u_{T}^{(1)}\sim\ell_{\theta'}(\cdot)$ and $u_{t-1}^{(1)}\sim R_{\theta',\theta,t}(u_{t}^{(1)},\cdot)$,
$t=T,\ldots,1$,\\
\For{$i=2,\ldots,N$,}{

Set $u_{0}^{(i)}=u_{0}^{(1)}$ and sample $u_{t}^{(i)}\sim R_{\theta',\theta,t}(u_{t-1}^{(i)},\cdot)$,
$t=1,\ldots,T$.

} 

Set $\theta_{n+1}=\theta'$ with probability $\min\{1,1/\mathring{r}_{\mathfrak{u}}^{N}(\theta',\theta)\}$,
otherwise set $\theta_{n+1}=\theta$.

} 
\end{algorithm}

\subsection{Numerical example: the Ising model\label{subsec: Numerical example: the Ising model}}

We illustrate the performance of Algorithms \ref{alg: MHAAR-AIS exchange algorithm}
and \ref{alg: MHAAR-AIS exchange algorithm - reduced computation}
on the Ising model used in statistical mechanics to model ferromagnetism.
For $m,n\in\mathbb{N}$ we consider an $m\times n$ lattice $\Lambda$.
Associated to each site $k\in\Lambda$ is a binary variable $\mathfrak{z}[k]\in\{-1,1\}$
representing the spin configuration of the site. The probability of
a given configuration $u=\{\mathfrak{z}[k],k\in\Lambda\}$ depends
on an energy function, or Hamiltonian, which may depend on some parameter
$\theta$. A standard choice used in the absence of an external magnetic
field is 
\[
H_{\theta}(\mathfrak{z})=-\theta\sum_{i\sim j}\mathfrak{z}[i]\mathfrak{z}[j],
\]
where $i\sim j$ denotes a pair or adjacent sites and $\theta\in\Theta=\mathbb{R}_{+}$
is referred to as the inverse temperature parameter. The probability
of configuration $\mathfrak{z}$ for temperature $\theta^{-1}$ is
given by $\ell_{\theta}(\mathfrak{z})=g_{\theta}(\mathfrak{z})/C_{\theta}$
where $g_{\theta}(\mathfrak{z})=\exp(-H_{\theta}(\mathfrak{z}))$
and $C_{\theta}=\sum_{\mathfrak{z}\in\{0,1\}^{|\Lambda|}}g_{\theta}(\mathfrak{z})$
is the intractable and $\theta$-dependent normalising constant. In
the following experiment, we perform Bayesian estimation of $\theta$
given a $20\times30$ configuration $\mathfrak{y}$ drawn from $\ell_{\theta^{*}}(\cdot)$
for $\theta^{\ast}=0.35$, which is slightly above the critical (inverse)
temperature $\log(1+\sqrt{2})/2$, resulting in strongly correlated
neighbouring sites. The prior distribution for $\theta$ is taken
to be the uniform distribution on $(0,10)$. The difficulty here is
that computing $C_{\theta}$ requires the summation of $2^{600}$
terms, which is computationally infeasible.

The sequence of intermediate distributions used within AIS relies
on a geometric annealing schedule for the unnormalised densities of
the annealing distributions that is 
\[
f_{\theta,\theta',t}(\mathfrak{z})=g_{\theta}(\mathfrak{z})^{1-\beta_{t}}g_{\theta'}(\mathfrak{z})^{\beta_{t}}=g_{\theta(1-\beta_{t})+\theta'\beta_{t}}(\mathfrak{z}),\quad\beta_{t}=1-\frac{t}{T+1},\quad t=0,1,\ldots,T+1.
\]
Sampling from the intractable distribution is performed approximately
by running Wolff's algorithm, essentially an MCMC kernel iterated
for $100$ iterations. For $\theta,\theta'\in\Theta$ and $t=1,\ldots,T$
we chose $R_{\theta,\theta',t}$ to be a single iteration of the MCMC
kernel of the Wolff's algorithm targeting $\ell_{\theta(1-\beta_{t})+\theta'\beta_{t}}(\cdot)$.
We ran both Algorithms \ref{alg: MHAAR-AIS exchange algorithm} and
\ref{alg: MHAAR-AIS exchange algorithm - reduced computation} for
all of the combinations of $N=1,10,20,\ldots,100$ and $T=1,2,\ldots,10,20,\ldots,100$.
For each run, $K=10^{6}$ samples were generated and the last $3K/4$
of them were used to compute the IAC of the sequence $\{\theta_{i},i=1\geq1\}$.
Figure \ref{fig: Potts exchange with bridging vs asymmetric MCMC}
concentrates on the two extreme scenarios where $N=1$ and when $T=0$,
that correspond to the exchange algorithm with bridging as in \citet{Murray_et_al_2006}
and our novel averaging algorithm applied to Example \ref{ex:doublyintractaveraging},
respectively. The figure suggests that our algorithm is computationally
superior on an ideal parallel machine, at least for the present example.

The rest of the results are shown in Figure \ref{fig: IAC vs T vs N for the Ising model}.
The results are organised in order to contrast Algorithms \ref{alg: MHAAR-AIS exchange algorithm}
and \ref{alg: MHAAR-AIS exchange algorithm - reduced computation}.
The figure suggests that Algorithm \ref{alg: MHAAR-AIS exchange algorithm},
which uses multiple samples from the intractable distribution per
iteration, is uniformly better, as expected. Finally, although for
large $T$ the performances of the two algorithms get closer, for
small $T$ the advantage of using more samples from the intractable
distribution, i.e.\ using Algorithm \ref{alg: MHAAR-AIS exchange algorithm}
is more significant.

\begin{figure}[H]
\centerline{\includegraphics[scale=0.5]{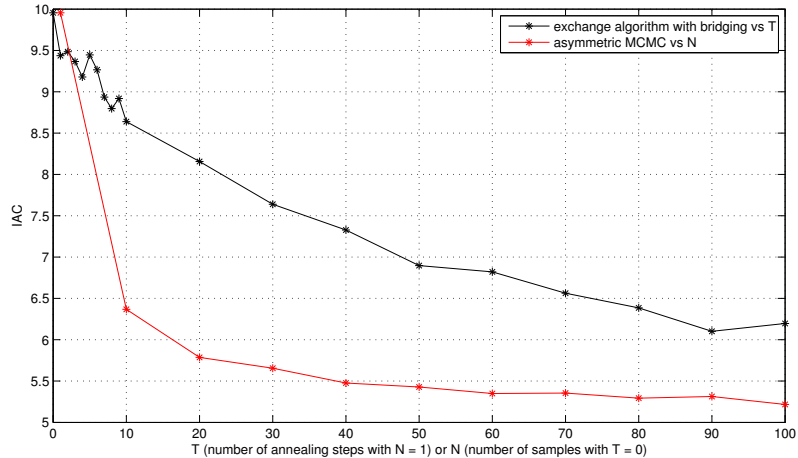}}
\protect\protect\caption{IAC for $\theta$ in the Ising model vs (a) the number of averaged
ratios $N=1,10,20,\ldots,100$ for $T=0$ (red/grey) and (b) the number
of annealing steps $T=0,1,2,\ldots,10,20,\ldots,100$ for $N=1$ (black).}

\label{fig: Potts exchange with bridging vs asymmetric MCMC} 
\end{figure}

\begin{figure}[!h]
\centerline{\includegraphics[scale=0.7]{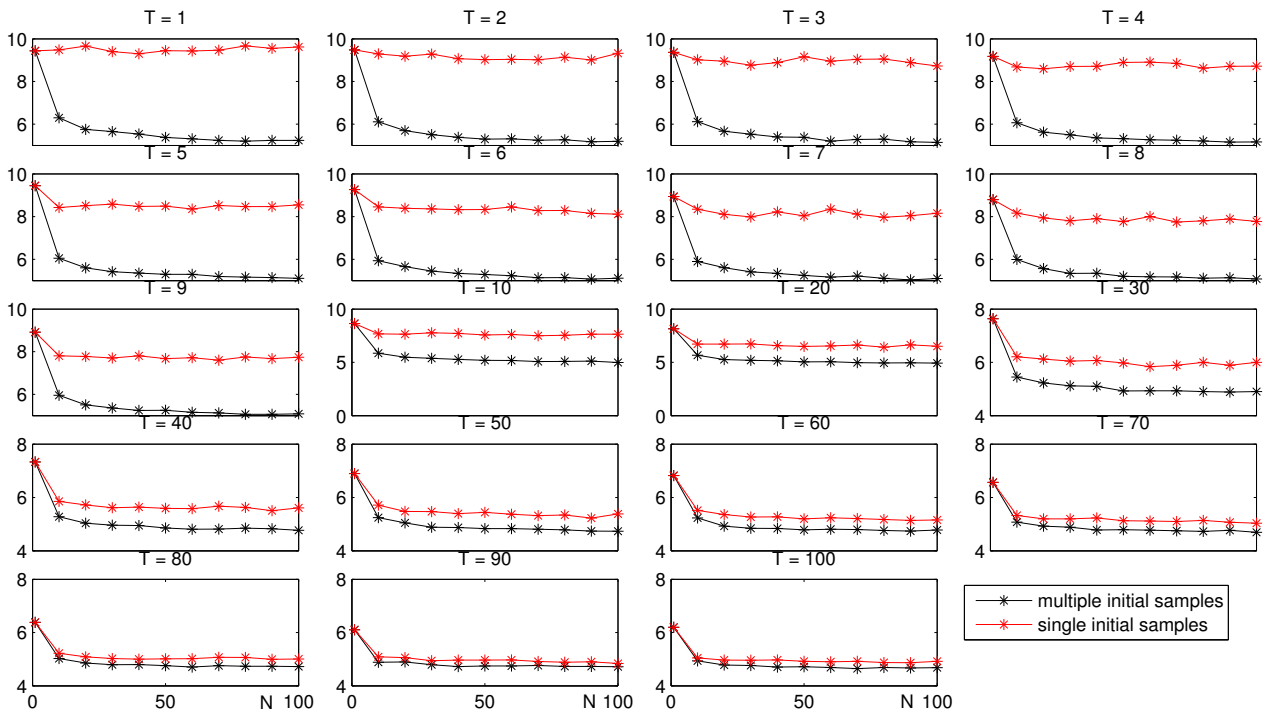}}
\protect\protect\caption{IAC for $\theta$ for the combinations of $N=1,10,20,\ldots,100$
and $T=1,2,\ldots,10,20,\ldots,100$. Each plot shows IAC vs $N$
for a fixed $T$.}

\label{fig: IAC vs T vs N for the Ising model} 
\end{figure}

\section{PMR algorithms for latent variable models \label{sec: Pseudo-marginal ratio algorithms for latent variable models}}

\subsection{Latent variable models \label{subsec: Latent variable models}}

We consider here sampling from a distribution that is the marginal
of a given joint distribution. More precisely, let $(\Theta,\mathcal{E})$
and $(\mathsf{Z},\mathcal{Z})$ be two measurable spaces, and define
the product spaces $\mathsf{X}=\Theta\times\mathsf{Z}$ and $\mathcal{X}=\mathcal{E}\otimes\mathcal{Z}$
the corresponding product $\sigma$-algebra. Let $\pi({\rm d}x):=\pi({\rm d}(\theta,z))$
be a probability distribution on $(\mathsf{X},\mathcal{X})$ which
is assumed known up to a normalising constant. Our primary interest
is to sample from the marginal distribution of $\theta$, 
\[
\pi({\rm d}\theta)=\int_{\mathsf{Z}}\pi\big({\rm d}(\theta,z)\big),
\]
assumed to be intractable, i.e.\ no useful density is available,
even up to a normalising constant. The doubly intractable scenario
covered so far falls into this category. It exploits the fact that
\[
\pi\big(\theta,z\big)\propto\eta(\theta)g_{\theta}(\mathfrak{y})\frac{h_{\mathfrak{y}}(z)}{g_{\theta}(z)}\ell_{\theta}\big(z\big),
\]
has $\pi\big(\theta\big)\propto\eta(\theta)\ell_{\theta}\big(\mathfrak{y}\big)$
as marginals, but also the additional property that sampling from
the intractable distribution $\ell_{\theta}\big(z\big)$ is possible.
This latter property is fundamental to by-pass the intractability
of the normalising constant, but also allows one to refresh $z$ at
each iteration of the MCMC algorithm, in contrast with the pseudo-marginal
approach. As a result the exchange algorithm defines an algorithm
which directly targets $\pi(\theta)$ with a Markov chain defined
on $(\Theta,\mathcal{E})$. This however turns out to be too specific
and restrictive for numerous applications, such as state-space models.
\begin{example}
\label{ex: state-space example for latent variable section}We consider
the well-known non-linear state-space, often used to assess the performance
of inference methods for non-linear state-space models, 
\begin{align*}
Z_{t} & =Z_{t-1}/2+25Z_{t-1}/(1+Z_{t-1}^{2})+8\cos(1.2t)+V_{t},\quad t\geq2\\
Y_{t} & =Z_{t}^{2}/20+W_{t},\quad t\geq1,
\end{align*}
where $Z_{1}\sim\mathcal{N}(0,10)$, $V_{t}\overset{\mathrm{iid}}{\sim}\mathcal{N}(0,\sigma_{v}^{2})$,
$W_{t}\overset{\mathrm{iid}}{\sim}\mathcal{N}(0,\sigma_{w}^{2})$.
The parameter of primary interest is $\theta=(\sigma_{v}^{2},\sigma_{w}^{2})$
and is ascribed the prior $(\sigma_{v}^{2},\sigma_{w}^{2})\overset{\mathrm{iid}}{\sim}\mathcal{IG}(0.01,0.01)$
where $\mathcal{IG}(a,b)$ is the inverse gamma distribution with
shape and scale parameters $a$ and $b$. The aim is to infer $x=(\theta,z)$,
where the latent variable is $z=z_{1:P}$ for some $P>1$, from a
particular data set $Y_{1:P}=y_{1:P}$. 
\end{example}
Ideally we would like to use the following ``marginal'' algorithm.
Let $q(\theta,\cdot)$ be a Markov kernel on $(\Theta,\mathcal{E})$
such that for each $\theta\in\Theta$, $q(\theta,\cdot)$ admits a
density $q(\theta,\cdot)$ with respect to ${\rm d}\theta'$. The
acceptance rate of the MH algorithm with proposal kernel $q(\cdot,\cdot)$
targeting $\pi(\theta)$ is 
\begin{equation}
r(\theta,\theta')=\frac{q(\theta',\theta)\pi(\theta')}{q(\theta,\theta')\pi(\theta)}.\label{eq: marginal MCMC acceptance probability}
\end{equation}

The latter cannot be evaluated in numerous scenarios of interest and
the aim of this section is to extend the framework developed for the
doubly intractable scenario to the more general situation where sampling
of the latent variable must be included in the MCMC scheme itself
and cannot be performed exactly. This results in an algorithm tar-getting
the distribution $\pi(\mathrm{d}(\theta,z))$. It turns out that the
framework developed in Section \ref{sec: Pseudo-marginal ratio algorithms using averaged acceptance ratio estimators}
can also be easily adapted to this scenario. More precisely, here
we have $x=(\theta,z)$ and $y=(\theta',z')$ and the only difference
with the developments of Algorithm \ref{alg: MHAAR for Pseudo-marginal ratio}
is concerned with the order in which the variables are sampled. In
Algorithm \ref{alg: MHAAR for Pseudo-marginal ratio} we have assumed
a specific sampling order for the variables involved, that is the
auxiliary variable copies are sampled after the proposed value $y$.
Here we are going to consider the scenario where $\theta'$ is sampled
first, then the auxiliary variables $u^{(1)},\ldots,u^{(N)}$ are
sampled from a kernel $Q_{\theta,\theta',z}({\rm d}u)$ and $z'$
is proposed last, conditional upon the auxiliary variables $\theta,\theta',z$
and $u^{(1)},\ldots,u^{(N)}$. The resulting expression for the acceptance
ratio remains the same as that used in Algorithm \ref{alg: MHAAR for Pseudo-marginal ratio}
since it is not affected by the order in which the variables are sampled.

\subsection{AIS within MH for latent variable models \label{subsec: AIS within MH algorithms}}

\citet{Neal_2004} suggested to use AIS, as described in Section \ref{sec: Pseudo-marginal ratio algorithms using averaged acceptance ratio estimators}
and Theorem \ref{thm:AIS} in Appendix \ref{sec: A short justification of AIS and an extension},
in order to achieve sampling from $\pi$. The idea should be clear
upon noticing that for $\theta\in\Theta$ fixed, $\pi(\theta)$ is
the normalising constant of the conditional distribution for $z$
that is proportional to $\pi(\theta,z)$, that is $\pi_{\theta}(z)\propto\pi(\theta,z)$.
To estimate the ratio $\pi(\theta')/\pi(\theta)$ one therefore defines
a sequence of artificial probability densities 
\[
\mathscr{P}_{\theta,\theta',T}:=\big\{\pi_{\theta,\theta',t},t=0,\ldots,T+1\big\}
\]
for some $T\geq1$ evolving from $\pi_{\theta,\theta',0}(z)=\pi_{\theta}(z)$
to $\pi_{\theta,\theta',T+1}(z)=\pi_{\theta'}(z)$, through a sequence
of unnormalised intermediate probability densities $\mathscr{F}_{\theta,\theta',T}=\{f_{\theta,\theta',t},t=0,\ldots,T+1\}$.
The following proposition establishes that this algorithm is conceptually
of the same form as $\mathring{P}$ given in \eqref{eq:defPring}
and this allows us to extend this methodology through Algorithm \ref{alg: MHAAR for Pseudo-marginal ratio}.
\begin{prop}
\label{prop:MHwithAISinside}Consider the latent variable model given
in the introduction of Section \ref{sec: Pseudo-marginal ratio algorithms for latent variable models}
and for any $\theta,\theta'\in\Theta$ let 

\begin{enumerate}
\item \label{enu:firstassumptionMHwithAISinside}$\mathscr{F}_{\theta,\theta',T}=\big\{ f_{\theta,\theta',t},t=0,\ldots,T+1\big\}$
be a family of tractable unnormalised densities of $\mathscr{P}_{\theta,\theta',T}$
defined on $\big(\mathsf{Z},\mathcal{Z}\big)$ such that

\begin{enumerate}
\item \label{enu:a}for $t=0,\ldots,T$, $f_{\theta,\theta',t}$ and $f_{\theta,\theta',t+1}$
have the same support, 
\item \label{enu:b}for any $z\in\mathsf{Z}$ and $t=1,\ldots,T$ $f_{\theta,\theta',t}(z)=f_{\theta',\theta,T+1-t}(z)$,
\item \label{enu:c}$f_{\theta,\theta',0}(z)=\pi(\theta,z)$ and $f_{\theta,\theta',T+1}(z)=\pi(\theta',z)$,
\end{enumerate}
\item \label{enu:secondassumptionMHwithAISinside}$\mathscr{R}_{\theta,\theta',T}=\big\{ R_{\theta,\theta',t}(\cdot,\cdot)\colon\mathsf{Z}\times\mathcal{Z}\to[0,1],t=1,\ldots,T\big\}$
be a family of Markov transition kernels such that for any $t=1,\ldots,T$

\begin{enumerate}
\item $R_{\theta,\theta',t}(\cdot,\cdot)$ is $\pi_{\theta,\theta',t}-$reversible, 
\item $R_{\theta,\theta',t}(\cdot,\cdot)=R_{\theta',\theta,T-t+1}(\cdot,\cdot)$, 
\end{enumerate}
\item \label{enu:Ruzispithetareversibleforlatents}$R_{\theta}(\cdot,\cdot)\colon\mathsf{Z}\times\mathcal{Z}\to[0,1]$
be a $\pi_{\theta}-$reversible Markov transition kernel,
\item \label{enu:fourhassumption MHwithAISinside}$Q_{\theta,\theta',z}(\cdot)$
be probability distributions on $\big(\mathsf{U},\mathcal{U}\big)$
where $\mathsf{U}:=\mathcal{\mathsf{Z}}^{T+1}$ defined for 
\begin{align}
Q_{\theta,\theta',z}({\rm d}u) & =R_{\theta}(z,{\rm d}u_{0})\prod_{t=1}^{T}R_{\theta,\theta',t}(u_{t-1},{\rm d}u_{t}),\label{eq: M}
\end{align}
and let $\varphi$ be the involution which reverses the order of the
components of $u$; i.e.\ $\varphi(u_{0},u_{1},\ldots,u_{T}):=(u_{T},u_{T-1},\ldots,u_{0})$
for all $u\in\mathsf{U}$. 
\end{enumerate}
Then for any $(\theta,z),(\theta',z'),u\in\big(\Theta\times\mathsf{\mathsf{Z}}\big)^{2}\times\mathsf{U}$
\begin{equation}
\bar{Q}_{\theta,\theta',z}({\rm d}u)=R_{\theta}(z,{\rm d}u_{T})\prod_{t=1}^{T}R_{\theta',\theta,t}(u_{t},{\rm d}u_{t-\text{1}}),\label{eq:L}
\end{equation}
and
\begin{equation}
\frac{\pi_{\theta'}({\rm d}z')\bar{Q}_{\theta',\theta,z'}({\rm d}u)R_{\theta}(u_{0},{\rm d}z)}{\pi_{\theta}({\rm d}z)Q_{\theta,\theta',z}({\rm d}u)R_{\theta'}(u_{T},{\rm d}z')}=\frac{\pi(\theta)}{\pi(\theta')}\prod_{t=0}^{T}\frac{f_{\theta,\theta',t+1}(u_{t})}{f_{\theta,\theta',t}(u_{t})}.\label{eq:AISbasedAcceptRatio}
\end{equation}
The AIS MCMC algorithm of \citet{Neal_2004} for latent variable models
corresponds to $\mathring{P}$ in Theorem \ref{thm: pseudo-marginal ratio algorithms}
with $x=(\theta,z)$ and $y=(\theta',z')$, the proposal kernel 
\[
Q_{1}(x,{\rm d}(y,u)):=q(\theta,{\rm d}\theta')Q_{\theta,\theta',z}({\rm d}u)R_{\theta'}(u_{T},{\rm d}z')
\]
and its complementary kernel 
\[
Q_{2}(x,{\rm d}(y,u)):=q(\theta,{\rm d}\theta')\bar{Q}_{\theta,\theta',z}({\rm d}u)R_{\theta'}(u_{0},{\rm d}z').
\]
Its acceptance ratio on $\mathring{\mathsf{S}}$ is 
\[
\mathring{r}_{u}(\theta,z;\theta',z')=\frac{\pi({\rm d}x')Q_{2}(y,{\rm d}(x,u))}{\pi({\rm d}x)Q_{1}(x,{\rm d}(y,u))}=\frac{q(\theta',\theta)}{q(\theta,\theta')}\prod_{t=0}^{T}\frac{f_{\theta,\theta',t+1}(u_{t})}{f_{\theta,\theta',t}(u_{t})}.
\]

\end{prop}
\begin{proof}
Since $\bar{Q}_{\theta,\theta',z}(A)=Q_{\theta,\theta',z}(\varphi(A))$,
we can check that the pair $Q_{1}(x,\cdot)$, $Q_{2}(x,\cdot)$ satisfy
the assumption of Theorem \ref{thm: pseudo-marginal ratio algorithms}
in Appendix \ref{sec: A general framework for MPR and MHAAR algorithms}.
Next, using the symmetry assumption on $\mathscr{R}_{\theta,\theta',T}$,
we obtain
\begin{align*}
\bar{Q}_{\theta,\theta',z}({\rm d}u) & =R_{\theta}(z,{\rm d}u_{T})\prod_{t=1}^{T}R_{\theta',\theta,T-t+1}(u_{T-t+1},{\rm d}u_{T-t})
\end{align*}
and we can thus apply Theorem \ref{thm:AIS} in Appendix \ref{sec: A short justification of AIS and an extension}
(with $\tau=T+2$ intermediate distributions, two repeats $\mu_{0}=\mu_{1}$
and $\mu_{\tau}=$ $\mu_{\tau+1}$, $\mu_{t}=\pi_{\theta,\theta',t-1}$
for $t=2,\ldots,\tau-1$ and kernels $\Pi_{1}=R_{\theta}$, $\Pi_{\tau}=R_{\theta'}$
and $\Pi_{t}=R_{\theta,\theta',t-1}$ for $t=2,\ldots,\tau-1$) to
show that $\pi_{\theta'}\times\bar{Q}_{\theta',\theta,\cdot}\times R_{\theta}$
is absolutely continuous with respect to $\pi_{\theta}\times Q_{\theta,\theta',\cdot}\times R_{\theta'}$
and that the expression for the corresponding Radon-Nikodym derivative
ensures that the acceptance ratio defined in \eqref{eq:accept-ratio-circle}
is indeed equal to \eqref{eq: ratio for doubly intractable with annealing}. 
\end{proof}
The standard choice made in \citet{Neal_2004} corresponds to $R_{\theta}(z,{\rm d}u_{0})=\delta_{z}\big({\rm d}u_{0}\big)$,
but more general choices are possible. As we shall see in the next
section, a choice different from $\delta_{z}\big({\rm d}u_{0}\big)$
can improve performance significantly when averaging acceptance ratios. 

The variance of this unbiased estimator $\mathring{r}_{u}(\theta,z;\theta',z')$
of $r\big(\theta,\theta'\big)$ can usually be tuned by increasing
$T$, under natural smoothness conditions on the sequences $\mathscr{F}_{\theta,\theta',T}$
for $T\geq1$. An important point here is that although the approximated
acceptance ratio is reminiscent of that of a MH algorithm targeting
$\pi({\rm d}\theta)$, the present algorithm targets the joint distribution
$\pi\big({\rm d}(\theta,z)\big)$: the simplification occurs only
because the random variable corresponding to $u_{T}$ in \eqref{eq: M}
will be approximately distributed according to $\pi_{\theta'}(\cdot)$
when $T$ is large enough, under proper mixing conditions. 

We note that the expression for $\mathring{r}_{u}(\theta,z;\theta',z')$
does not depend on either $z$ or $z'$, and can in particular be
calculated before sampling $z'$. This is of importance in what follows
and justifies the use of the simplified piece of notation $\mathring{r}_{u}(\theta,\theta')$
below. 

\subsection{Averaging AIS based pseudo-marginal ratios \label{subsec: Averaging AIS based acceptance ratios}}

We show here how the algorithm of the previous section (Proposition
\ref{prop:MHwithAISinside}) can be modified in order to average multiple
($N>1$) estimators $\mathring{r}_{u}(\theta,\theta')$ of $r(\theta,\theta')$
while preserving reversibility of the algorithm of interest. Let $u=(u_{0},\ldots,u_{T})\in\mathsf{U}=\mathsf{Z}^{T+1}$
and $k\in\{1,\ldots,N\}$. 
\begin{prop}
\label{prop: asymmetric MCMC with latent variables}Assume that the
conditions of Proposition \ref{prop:MHwithAISinside} hold. For $N\geq1$
define the proposal kernels $Q_{1}^{N}(\cdot,\cdot)$ and $Q_{2}^{N}(\cdot,\cdot)$
on $\big(\mathsf{X}\times\mathsf{\mathfrak{U}}\times[k],\mathcal{X}\otimes\mathscr{U}\otimes\mathscr{P}[k]\big)$
\begin{align}
Q_{1}^{N}\big(x;{\rm d}(y,\mathfrak{u},k)\big) & =q(\theta,{\rm d}\theta')\prod_{i=1}^{N}Q_{\theta,\theta',z}({\rm d}u^{(i)})\frac{\mathring{r}_{u^{(k)}}(\theta,\theta')}{\sum_{i=1}^{N}\mathring{r}_{u^{(i)}}(\theta,\theta')}R_{\theta'}(u_{T}^{(k)},{\rm d}z'),\label{eq: asymmetric MCMC combining AIS and pseudo MCMC Q1}\\
Q_{2}^{N}\big(x;{\rm d}(y,\mathfrak{u},k)\big) & =q(\theta,{\rm d}\theta')\frac{1}{N}\bar{Q}_{\theta,\theta',z}({\rm d}u^{(k)})R_{\theta'}(u_{0}^{(k)},{\rm d}z')\prod_{i=1,i\neq k}^{N}Q_{\theta',\theta,z'}({\rm d}u^{(i)}).\label{eq: asymmetric MCMC combining AIS and pseudo MCMC Q2}
\end{align}
Then one can implement $\mathring{P}^{N}$ corresponding to $\mathring{P}$
defined in Proposition \ref{prop:MHwithAISinside}, with $Q_{1}^{N}(\cdot,\cdot)$
and $Q_{2}^{N}(\cdot,\cdot)$ above and
\[
\mathring{r}_{\mathfrak{u}}^{N}(\theta,\theta')=\frac{1}{N}\sum_{i=1}^{N}\mathring{r}_{u^{(i)}}(\theta,\theta').
\]
\end{prop}
\begin{proof}
One can check directly that $\mathring{r}^{N}(\theta,\theta')$ is
of the expected form despite the sampling order change 
\begin{align*}
\frac{\pi({\rm d}y)Q_{2}^{N}\big(y,{\rm d}(x,\mathfrak{u},k)\big)}{\pi({\rm d}x)Q_{1}^{N}\big(x,{\rm d}(y,\mathfrak{u},k)\big)} & =\frac{\pi({\rm d}y)q(\theta',{\rm d}\theta)\frac{1}{N}\bar{Q}_{\theta,\theta',z'}({\rm d}u^{(k)})R_{\theta}(u_{0}^{(k)},{\rm d}z)\prod_{i=1,i\neq k}^{N}Q_{\theta,\theta',z}({\rm d}u^{(i)})}{\pi({\rm d}x)q(\theta,{\rm d}\theta')\prod_{i=1}^{N}Q_{\theta,\theta',z}({\rm d}u^{(i)})\frac{\mathring{r}_{u^{(k)}}(\theta,\theta')}{\sum_{i=1}^{N}\mathring{r}_{u^{(i)}}(\theta,\theta')}R_{\theta'}(u_{T}^{(k)},{\rm d}z')}\\
 & =\frac{q(\theta',\theta)\pi(\theta')}{q(\theta,\theta')\pi(\theta)}\frac{\pi_{\theta'}(\mathrm{d}z')\bar{Q}_{\theta',\theta,z'}({\rm d}u^{(k)})R_{\theta}(u_{0}^{(k)},{\rm d}z)}{\pi_{\theta}(\mathrm{d}z)Q_{\theta,\theta',z}({\rm d}u^{(k)})R_{\theta'}(u_{T}^{(k)},{\rm d}z')}\mathring{r}_{u^{(k)}}^{-1}(\theta,\theta')\frac{1}{N}\sum_{i=1}^{N}\mathring{r}_{u^{(i)}}(\theta,\theta').
\end{align*}
\end{proof}
The implementation of the resulting asymmetric MCMC algorithm is described
in Algorithm \ref{alg: MHAAR for pseudo-marginal ratio in latent variable models}.
The interest of introducing a general form for $R_{\theta}$ should
now be clear: the standard choice $R_{\theta}\big(z,\cdot\big)=\delta_{z}(\cdot)$
introduces dependence among $u^{(1)},u^{(2)},\ldots,u^{(N)}$ which
can be alleviated by the introduction of a more general ergodic transition,
which may consist of an iterated reversible Markov transition of invariant
distribution $\pi_{\theta}$. We also notice that some computational
savings are possible. For example when $Q_{1}^{N}(\cdot,\cdot)$ is
the distribution we sample from, the acceptance ratio does not depend
on $k$, whose sampling can therefore be postponed until after a decision
to accept has been made. The complementary update for which we sample
from $Q_{2}^{N}(\cdot,\cdot)$ effectively does not require sampling
$k$ which is set to $1$ in our implementation in Algorithm \ref{alg: MHAAR for pseudo-marginal ratio in latent variable models}.

\begin{algorithm}[h]
\caption{MHAR for averaged AIS PMR estimators for general latent variable models}

\label{alg: MHAAR for pseudo-marginal ratio in latent variable models} 

\KwIn{Current sample $X_{n}=x=(\theta,z)$} 

\KwOut{New sample $X_{n+1}$} 

Sample $\theta'\sim q(\theta,\cdot)$ and $v\sim\mathcal{U}(0,1)$.
\\
\If{$v\leq1/2$}{

\For{$i=1,\ldots,N$}{ 

Sample $u_{0}^{(i)}\sim R_{\theta}(z,\cdot)$ and $u_{t}^{(i)}\sim R_{\theta,\theta',t}(u_{t-1}^{(i)},\cdot)$
for $t=1,\ldots,T$.

} 

Sample $k\sim\mathcal{P}\big(\mathring{r}_{u^{(1)}}(\theta,\theta'),\ldots,\mathring{r}_{u^{(N)}}(\theta,\theta')\big)$
and $z'\sim R_{\theta'}(u_{T}^{(k)},\cdot)$.\\
Set $X_{n+1}=(\theta',z')$ with probability $\min\{1,\mathring{r}_{\mathfrak{u}}^{N}(\theta,\theta')\},$
otherwise set $X_{n+1}=x$. 

}\Else{ 

Sample $u_{T}^{(1)}\sim R_{\theta}\big(z,\cdot\big)$, $u_{t-1}^{(1)}\sim R_{\theta',\theta,t}(u_{t},\cdot)$
for $t=T,\ldots,1$ and $z'\sim R_{\theta'}(u_{0}^{(1)},\cdot)$.\\
\label{line:elsepseudomarginal} \For{ $i=2,\ldots,N$,}{

Sample $u_{0}^{(i)}\sim R_{\theta'}(z',\cdot)$ and $u_{t}^{(i)}\sim R_{\theta',\theta,t}(u_{t-1}^{(i)},\cdot)$
for $t=1,\ldots,T$.

} 

Set $X_{n+1}=(\theta',z')$ with probability $\min\{1,1/\mathring{r}_{\mathfrak{u}}^{N}(\theta',\theta)\}$,
otherwise set $X_{n+1}=x$. 

} 
\end{algorithm}
\begin{example}[Example \ref{ex: state-space example for latent variable section},
ctd.]
\label{ex: ctd state-space example for latent variable section}In
order to illustrate the interest of our approach, we generated data
from the model for $P=500$, $\sigma_{v}^{2}=10$ and $\sigma_{w}^{2}=0.1$.
The set-up for Algorithm \ref{alg: MHAAR for pseudo-marginal ratio in latent variable models}
was as follows. We let $T=1$ and for $\theta,\theta'\in\Theta$ the
unnormalised density of the intermediate distribution was chosen to
be $f_{\theta,\theta',1}(z)=\pi((\theta+\theta')/2,z)$. The MCMC
kernel $R_{\theta,\theta',1}(\cdot,\cdot)$ was a conditional SMC
(cSMC) \citet{Andrieu_et_al_2010} tar getting the intermediate distribution,
with $M=100$ particles and the model transitions as proposal distributions;
for convenience the cSMC kernel is described in Section \ref{sec: State-space models: SMC and conditional SMC within MHAAR}.
We used a normal random walk proposal with diagonal covariance matrix
as a parameter proposal, where the standard deviations for $\sigma_{v}$
and $\sigma_{w}$ were $0.15$ and $0.08$ respectively. Performance,
measured in terms of convergence to equilibrium and asymptotic variance
for $N=1$, $N=10$ and $N=100$, is presented in Figure \ref{fig: IAC vs N for the state-space model-1}
and \ref{fig: IAC vs N for the state-space model}. For each set-up,
$200$0 independent Monte Carlo runs of length 1000 each were used
to assess convergence to the posterior mean, posterior second moment
and median, via ensemble averages over the runs. We observe in Figure
\ref{fig: IAC vs N for the state-space model-1} that this simple
approach improves performance and reduces time to convergence by approximately
50\%. In addition to faster convergence, of the order of $30\%$,
in terms of IAC in Figure \ref{fig: IAC vs N for the state-space model-1}.
The estimated IAC values were obtained after discarding the first
300 iterations and by averaging over 2000 Monte Carlo runs. We present
further new developments for this application in Section \ref{subsec: An application: trans-dimensional distributions}.

\begin{figure}[h]
\centerline{\includegraphics[scale=0.7]{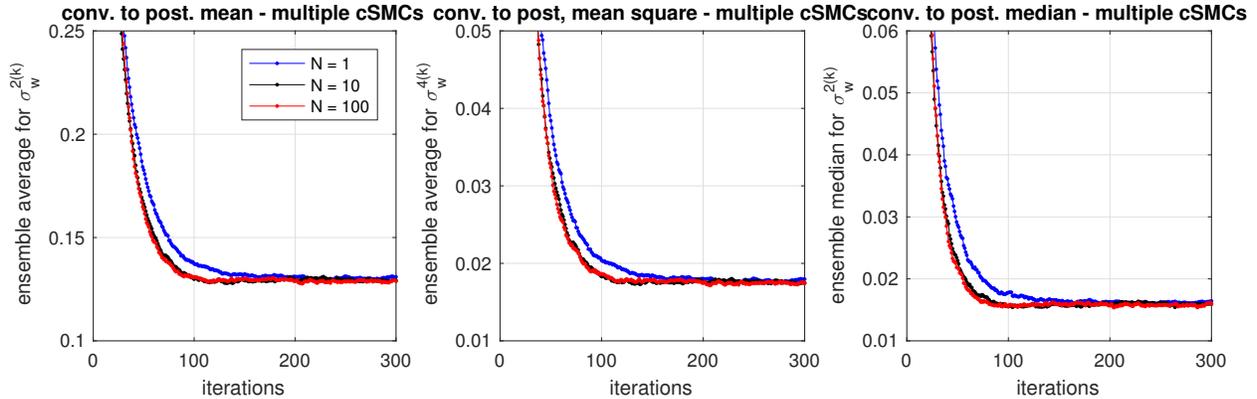}}
\protect\protect\caption{Convergence results for $\theta=(\sigma_{v}^{2},\sigma_{w}^{2})$
vs $N$ in Algorithm \ref{alg: MHAAR for pseudo-marginal ratio in latent variable models}.}

\label{fig: convergence results for the state-space model} 
\end{figure}

\begin{figure}[h]
\centerline{\includegraphics[scale=0.7]{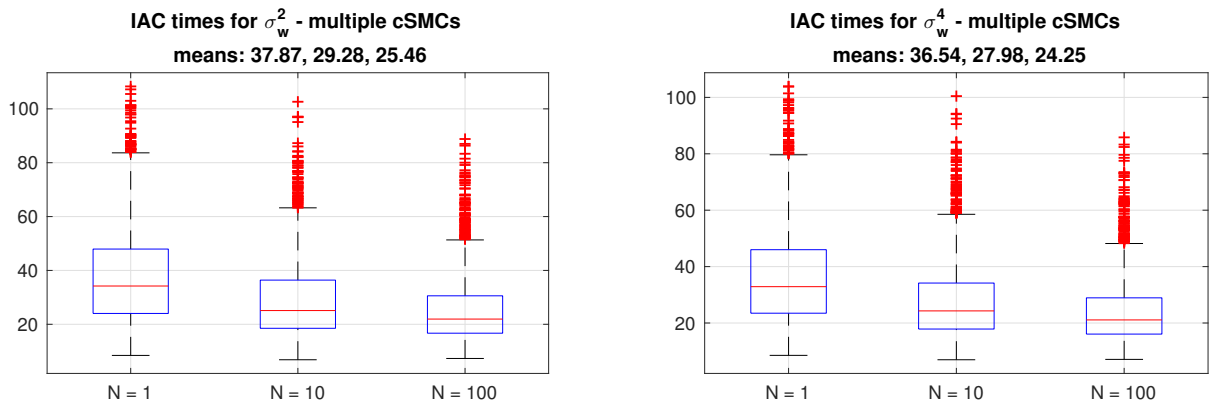}}
\protect\protect\caption{IAC for $\sigma_{w}^{2}$ and $\sigma_{w}^{4}$ vs $N$ in Algorithm
\ref{alg: MHAAR for pseudo-marginal ratio in latent variable models}.}

\label{fig: IAC vs N for the state-space model-1} 
\end{figure}
\end{example}

\subsection{Generalisations of MHAAR algorithms for latent variable models \label{subsec: Generalisations of pseudo-marginal asymmetric MCMC}}

We now discuss two generalisations of Algorithm \ref{alg: MHAAR for pseudo-marginal ratio in latent variable models}
above which will prove crucial in Section \ref{subsec: An application: trans-dimensional distributions},
where we present our trans-dimensional example as an application of
the methodology presented here, albeit in a scenario involving additional
complications.

\subsubsection{Annealing in a different space \label{subsec: Annealing in a different space}}

The first generalisation is based on the main idea that condition
\ref{enu:Ruzispithetareversibleforlatents} in Proposition \ref{prop:MHwithAISinside}
can be relaxed in the light of Theorem \ref{thm:AIS}, and in particular
allows the latent variable $z$ and auxiliary variables $u_{t}$ to
live on different spaces.
\begin{prop}
\label{prop:extensionMHwithAISinside}Suppose that assumptions \ref{enu:a}-\ref{enu:b}
and \ref{enu:secondassumptionMHwithAISinside} of Proposition \ref{prop:MHwithAISinside}
are satisfied with $\mathscr{F}_{\theta,\theta',T},\mathscr{P}_{\theta,\theta',T}$
and $\mathscr{R}_{\theta,\theta',T}$ now defined on some space $\big(\mathsf{V},\mathcal{V}\big)$
(and $\mathsf{U:=\mathsf{V}}^{T+1}$), (therefore $\pi_{\theta,\theta',0}\neq$$\pi_{\theta}$
, and $\pi_{\theta,\theta',T+1}\neq\pi_{\theta'}$ in general), and
assumptions \ref{enu:c} and \ref{enu:Ruzispithetareversibleforlatents}
replaced, for $\theta,\theta'\in\Theta$ and $z,z'\in\mathsf{Z}$,
with

\begin{enumerate}
\item the endpoint conditions for the unnormalised densities are of the
form
\begin{align*}
f_{\theta,\theta',0}(v) & =\pi_{\theta,\theta',0}(v)\pi(\theta),\\
f_{\theta,\theta',T+1}(v) & =\pi_{\theta,\theta',T+1}(v)\pi(\theta'),
\end{align*}
\item the existence of Markov transition kernels $\overrightarrow{R}_{\theta,\theta',0}$,$\overleftarrow{R}_{\theta,\theta',T+1}:\mathsf{Z}\times\mathcal{V}\rightarrow[0,1]$
and $\overrightarrow{R}{}_{\theta,\theta',T+1}$,$\overleftarrow{R}_{\theta,\theta',0}:\mathsf{V}\times\mathcal{Z}\rightarrow[0,1]$
such that 
\begin{align*}
\pi_{\theta}({\rm d}z)\overrightarrow{R}_{\theta,\theta',0}(z,{\rm d}v) & =\pi_{\theta,\theta',0}({\rm d}v)\overleftarrow{R}_{\theta,\theta',0}(v,{\rm d}z),\\
\pi_{\theta,\theta',T+1}({\rm d}v)\overrightarrow{R}_{\theta,\theta',T+1}(v,{\rm d}z) & =\pi_{\theta'}({\rm d}z)\overleftarrow{R}{}_{\theta,\theta',T+1}(z,{\rm d}v),
\end{align*}
\item Define the proposal probability distributions on $\big(\mathsf{U},\mathcal{U}\big)$
such that for any $u\in\mathsf{U}=\mathsf{V}^{T+1}$,
\begin{align*}
Q_{\theta,\theta',z}\big({\rm d}u\big) & =\overrightarrow{R}_{\theta,\theta',0}(z,{\rm d}u_{0})\prod_{t=1}^{T}R_{\theta,\theta',t}(u_{t-1},{\rm d}u_{t}),
\end{align*}
and the involution $\varphi$ reversing the order of the components
of $u$; i.e.\ $\varphi(u_{0},u_{1},\ldots,u_{T}):=(u_{T},u_{T-1},\ldots,u_{0})$
for all $u\in\mathsf{U}$.
\end{enumerate}
Then for any $\big((\theta,z),(\theta',z'),u\big)\in\big(\Theta\times\mathsf{\mathsf{Z}}\big)^{2}\times\mathsf{U}$
\[
\bar{Q}_{\theta,\theta',z}\big({\rm d}u\big)=\overleftarrow{R}_{\theta',\theta,T+1}(z,{\rm d}u_{T})\prod_{t=1}^{T}R_{\theta',\theta,T-t+1}(u_{T-t+1},{\rm d}u_{T-t}),
\]
and
\begin{equation}
\frac{\pi_{\theta'}({\rm d}z')\bar{Q}_{\theta',\theta,z'}({\rm d}u)\overleftarrow{R}_{\theta,\theta',0}(u_{0},{\rm d}z)}{\pi_{\theta}({\rm d}z)Q_{\theta,\theta',z}({\rm d}u)\overrightarrow{R}_{\theta,\theta',T+1}(u_{T},{\rm d}z')}=\frac{\pi(\theta)}{\pi(\theta')}\prod_{t=0}^{T}\frac{f_{\theta,\theta',t+1}(u_{t})}{f_{\theta,\theta',t}(u_{t})}.\label{eq:AISbasedAcceptRatio-1}
\end{equation}
Furthermore, suppose the additional symmetry conditions 
\begin{equation}
\overrightarrow{R}_{\theta,\theta',T+1}(v,{\rm d}z)=\overleftarrow{R}_{\theta',\theta,0}(v,{\rm d}z),\quad\overrightarrow{R}_{\theta,\theta',0}(z,{\rm d}v)=\overleftarrow{R}_{\theta',\theta,T+1}(z,{\rm d}v).\label{eq: general AIS MCMC further symmetry conditions}
\end{equation}
Then, a generalisation of the AIS MCMC algorithm in \citet{Neal_2004}
corresponds to $\mathring{P}$ in Theorem \ref{thm: pseudo-marginal ratio algorithms}
with $x=(\theta,z)$ and $y=(\theta',z')$, the proposal kernel
\[
Q_{1}\big(x,{\rm d}(y,u)\big):=q(\theta,{\rm d}\theta')Q_{\theta,\theta',z}({\rm d}u)\overrightarrow{R}_{\theta,\theta',T+1}(u_{T},{\rm d}z')
\]
and its complementary kernel 
\[
Q_{2}\big(x,{\rm d}(y,u)\big):=q(\theta,{\rm d}\theta')\bar{Q}_{\theta,\theta',z}({\rm d}u)\overleftarrow{R}_{\theta',\theta,0}(u_{0},{\rm d}z').
\]
Its acceptance ratio on set $\mathring{\mathsf{S}}$ is
\begin{equation}
\mathring{r}_{u}(\theta;\theta')=\frac{\pi\big({\rm d}y\big)Q_{2}\big(y,{\rm d}(x,u)\big)}{\pi\big({\rm d}x\big)Q_{1}\big(x,{\rm d}(y,u)\big)}=\frac{q(\theta',\theta)}{q(\theta,\theta')}\prod_{t=0}^{T}\frac{f_{\theta,\theta',t+1}(u_{t})}{f_{\theta,\theta',t}(u_{t})}.\label{eq: AIS-MCMC-acceptance-ratio-general}
\end{equation}
\end{prop}
\begin{proof}
The first claim follows from Theorem \ref{thm:extensionAIS}, which
can be exploited with similar steps to those in the proof of Proposition
\ref{prop:MHwithAISinside}. The second claim on the generalisation
of AIS MCMC follows from the fact that the symmetry conditions in
\eqref{eq: general AIS MCMC further symmetry conditions} ensure that
$Q_{1}$ and $Q_{2}$ defined in the proposition satisfy the assumption
of Theorem \ref{thm: pseudo-marginal ratio algorithms}.
\end{proof}
\begin{rem}
It may appear that the additional coupling conditions on the initial
and terminal distributions is only satisfied for reversible kernels.
However it should be clear that in the formulation above $z,z'$ and
$u_{0},\ldots,u_{T}$ can be of a different nature i.e.\ defined
on different spaces, which turns out to be relevant in some scenarios,
including that considered in Section \ref{subsec: Numerical example: Poisson multiple changepoint model}.
In fact, the generalisation of AIS MCMC mentioned in Proposition \ref{prop:extensionMHwithAISinside}
corresponds to the AIS RJ-MCMC algorithm of \citet{Karagiannis_and_Andrieu_2013}
for trans-dimensional distributions. It also covers the standard version
of the hybrid Monte Carlo algorithm, for example. 
\end{rem}
One can build upon this generalisation and use the framework of asymmetric
acceptance ratio MH algorithms corresponding to $\mathring{P}^{N}$
of Section \ref{subsec: Averaging AIS based acceptance ratios} in
order to define a $\pi-$reversible Markov transition probability.
\begin{prop}
\label{prop: asymmetric MCMC with latent variables - general spaces}Assume
that the conditions of Proposition \ref{prop:extensionMHwithAISinside}
hold. For $N\geq1$ define the proposal kernels $Q_{1}^{N}(\cdot)$
and $Q_{2}^{N}(\cdot)$ on $\big(\mathsf{X}\times\mathsf{\mathfrak{U}}\times[k],\mathcal{X}\otimes\mathscr{U}\otimes\mathscr{P}[k]\big)$
\begin{align}
Q_{1}^{N}\big(x,{\rm d}(y,\mathfrak{u},k)\big) & =q(\theta,{\rm d}\theta')\prod_{i=1}^{N}Q_{\theta,\theta',z}({\rm d}u^{(i)})\frac{\mathring{r}_{u^{(k)}}(\theta,\theta')}{\sum_{i=1}^{N}\mathring{r}_{u^{(i)}}(\theta,\theta')}\overrightarrow{R}_{\theta,\theta',T+1}(u_{T}^{(k)},{\rm d}z'),\label{eq: asymmetric MCMC combining AIS and pseudo MCMC Q1-1}\\
Q_{2}^{N}\big(x,{\rm d}(y,\mathfrak{u},k)\big) & =q(\theta,{\rm d}\theta')\frac{1}{N}\bar{Q}_{\theta,\theta',z}({\rm d}u^{(k)})\overleftarrow{R}_{\theta',\theta,0}(u_{0}^{(k)},{\rm d}z')\prod_{i=1,i\neq k}^{N}Q_{\theta',\theta,z'}({\rm d}u^{(i)}).\label{eq: asymmetric MCMC combining AIS and pseudo MCMC Q2-1}
\end{align}
Then one can implement $\mathring{P}^{N}$ corresponding to $\mathring{P}$
defined in Proposition \ref{prop:MHwithAISinside}, with $Q_{1}^{N}(\cdot,\cdot)$
and $Q_{2}^{N}(\cdot,\cdot)$ as above and
\[
\mathring{r}_{\mathfrak{u}}^{N}(\theta,\theta')=\frac{1}{N}\sum_{i=1}^{N}\mathring{r}_{u^{(i)}}(\theta,\theta')
\]
with $\mathring{r}_{u}(\theta;\theta')$ defined in \eqref{eq: AIS-MCMC-acceptance-ratio-general}.
\end{prop}

\subsubsection{Choosing $Q_{1}^{N}(\cdot,\cdot)$ and $Q_{2}^{N}(\cdot,\cdot)$
with different probabilities \label{subsec: Choosing Q1 and Q2 with different probabilities}}

Notice from \eqref{eq: asymmetric MCMC combining AIS and pseudo MCMC Q1}
and \eqref{eq: asymmetric MCMC combining AIS and pseudo MCMC Q2}
that $Q_{1}^{N}(\cdot)$ and $Q_{2}^{N}(\cdot)$ share the same proposal
distribution for $\theta'$ and start differing from each other when
generating the auxiliary variables and proposing $z'$ thereafter.
In some cases, depending on the values of $\theta$ and $\theta'$,
$Q_{1}^{N}(\cdot)$ (or $Q_{2}^{N}(\cdot)$) may be preferable over
$Q_{2}^{N}(\cdot)$ (or $Q_{1}^{N}(\cdot)$) for proposing $z'$.
This is indeed the case in our trans-dimensional example in Section
\ref{subsec: An application: trans-dimensional distributions}, where
the $\theta$ component stands for the model number. One can enjoy
this degree of freedom by a function $\beta:\Theta^{2}\rightarrow[0,1]$
which satisfies 
\begin{equation}
\int\beta(\theta,\theta')Q_{1}^{N}(x,\mathrm{\mathrm{d}}(y,\mathfrak{u},k))+\left(1-\beta(\theta,\theta')\right)Q_{2}^{N}(x,\mathrm{d}(y,\mathfrak{u},k))=1.\label{eq: condition for alpha and Q1 and Q2}
\end{equation}
Then, we can modify the overall transition kernel of the asymmetric
MCMC as follows:
\begin{align}
\begin{aligned}\mathring{\bar{P}}^{N}(x,dy)= & \left[\int\beta(\theta,\theta')Q_{1}^{N}(x,\mathrm{\mathrm{d}}(y,\mathfrak{u},k)\min\left\{ 1,\mathring{\overline{r}}_{\mathfrak{u}}^{N}(\theta,\theta')\right\} +\delta_{x}(\mathrm{d}y)\mathring{\bar{\rho}}_{1}(x)\right]\end{aligned}
\nonumber \\
+\left[\int\left(1-\beta(\theta,\theta')\right)Q_{2}^{N}(x,\mathrm{d}(y,\mathfrak{u},k))\min\left\{ 1,1/\mathring{\overline{r}}_{\mathfrak{u}}^{N}(\theta',\theta)\right\} +\delta_{x}(\mathrm{d}y)\mathring{\bar{\rho}}_{2}(x)\right]\label{eq: generalised asymmetric kernel with alpha}
\end{align}
where the modified acceptance ratio is defined as 
\begin{equation}
\mathring{\overline{r}}_{\mathfrak{u}}^{N}(\theta,\theta'):=\mathring{r}_{\mathfrak{u}}^{N}(\theta,\theta')\frac{1-\beta(\theta',\theta)}{\beta(\theta,\theta')}.\label{eq: acceptance ratio modified by alpha}
\end{equation}
Implementing the modification with respect to Algorithm \ref{alg: MHAAR for pseudo-marginal ratio in latent variable models}
is straightforward: One needs to replace $v\leq1/2$ with $v\leq\beta(\theta,\theta')$
and use $\mathring{\overline{r}}_{\mathfrak{u}}^{N}(x,y)$ instead
of $\mathring{r}_{\mathfrak{u}}^{N}(x,y)$. Proof of reversibility
is very similar to the proof of Proposition \ref{prop: asymmetric MCMC with latent variables}
and we skip it.

Note that the condition in \eqref{eq: condition for alpha and Q1 and Q2}
ensures that \eqref{eq: generalised asymmetric kernel with alpha}
is a valid transition kernel and it is satisfied whenever $\theta'$
is proposed by $Q_{1}^{N}$ and $Q_{2}^{N}$ in the same way, as in
\eqref{eq: asymmetric MCMC combining AIS and pseudo MCMC Q1} and
\eqref{eq: asymmetric MCMC combining AIS and pseudo MCMC Q2} where
the same $q(\theta,\mathrm{d}\theta')$ is used. One can in principle
write an even more general kernels than the one in \eqref{eq: generalised asymmetric kernel with alpha}
by making $\beta$ a function of $x$ and $y$ and imposing a condition
similar to \eqref{eq: condition for alpha and Q1 and Q2}, however
we find this generalisation not as interesting from a practical point
of view.

\subsection{An application: trans-dimensional distributions\label{subsec: An application: trans-dimensional distributions}}

Consider a trans-dimensional distribution $\bar{\pi}(m,{\rm d}z_{m})$
on $\mathsf{X}=\cup_{m\in\Theta}\{m\}\times\mathsf{Z}_{m}$ where
$\Theta\subseteq\mathbb{N}$ and the dimension $d_{m}$ of $\mathsf{Z}_{m}$
depends on $m$. For each $m$, we assume that the distribution $\pi(m,{\rm d}z_{m})$
admits a density $\pi(m,z_{m})$ known up to a normalising constant
not depending on $m$ or $z_{m}$. We let $\mathscr{Z}_{m}$ be the
sigma-algebra of the conditional distribution $\pi_{m}({\rm d}z_{m})$.
We are interested in efficient sampling from the marginal distribution
$\pi(m)$.

An approach for sampling from trans-dimensional distributions is the
reversible jump MCMC (RJ-MCMC) algorithm of \citet{Green_1995}. Designing
efficient RJ-MCMC algorithms is notoriously difficult and can lead
to unreliable samplers. \citet{Karagiannis_and_Andrieu_2013} develop
what they call the AIS RJ-MCMC algorithm to improve on the performance
of the standard RJ-MCMC algorithm. The AIS RJ-MCMC algorithm is a
variant of the AIS MCMC algorithm of \citet{Neal_2004} devised for
trans-dimensional distributions. Full details of the method are available
in \citet{Karagiannis_and_Andrieu_2013}; however, we will need to
go into some details here as well, in order to state our contribution,
the reversible multiple jump MCMC (RmJ-MCMC), of which we present
an instance in Algorithm \ref{alg: Reversible multiple jump MCMC}.
In what follows, for notational simplicity, we consider only algorithms
consisting of a single ``move'' in Green's terminology, between
any pair of models $m,m'\in\Theta$\textendash the generalisation
to multiple pairs is straightforward but requires additional indexing.
A RJ-MCMC update can be understood as being precisely the procedure
proposed in Section \ref{subsec: Annealing in a different space},
but adapted to the present trans-dimensional set-up. In this scenario
the nature of the target distributions comes with the additional complication
that statistically interpretable parameters ($z_{m},z_{m'}$ for models
$m,m'$ respectively) must be, following \citet{Green_1995}'s idea,
embedded in a potentially larger common space and that this expanded
parametrisation is only unique up to an invertible transformation.
We mainly deal with this issue in this section, as the details of
the algorithm are then very similar to those of Section \ref{subsec: Annealing in a different space}. 

\subsubsection{Dimension matching and ``forward'' parametrisation}

Following \citet{Green_1995} we couple models pairwise. More precisely,
for any couple $m,m'\in\Theta$, consider the $d_{m,m'}$ and $d_{m',m}$
dimensional variables such that $d_{m}+d_{m,m'}=d_{m'}+d_{m',m}$,
\[
\mathfrak{z}_{m,m'}\in\mathfrak{Z}_{m,m'},\quad\mathfrak{z}_{m,m'}\sim\omega_{m,m'},\quad\mathfrak{z}_{m',m}\in\mathsf{\mathfrak{Z}}_{m',m},\quad\mathfrak{z}_{m',m}\sim\omega_{m',m},
\]
which are called dimension matching variables, with the convention
that these variables and associated quantities should be ignored when
either $d_{m,m'}=0$ or $d_{m',m}=0$. Letting the extended space
$\mathsf{Z}_{m,m'}:=\mathsf{Z}_{m}\times\mathfrak{Z}_{m,m'}$, consider
a one-to-one measurable mapping $\phi_{m,m'}:\mathsf{Z}_{m,m'}\rightarrow\mathsf{Z}_{m',m}$
with its inverse $\phi_{m,m'}^{-1}=\phi_{m',m}$. Note that the nature
of $z_{m}$ and $z_{m'}$ may differ as may that of $\mathfrak{z}{}_{m,m'}$
and $\mathfrak{z}_{m',m}$, which explains the need for the (cumbersome)
indexing. In order to ease the notation in the following presentation,
for $z_{m,m'}:=(z_{m},\mathfrak{z}_{m,m'})\in\mathsf{Z}_{m,m'}$ and
$A_{m}\in\mathscr{Z}_{m,m'}$, we will use the following transformations
with implicit reference to $z_{m,m'}$ and $A_{m,m'}$:
\begin{equation}
\begin{split} & z_{m',m}:=(z_{m'},\mathfrak{z}_{m',m})=\phi_{m,m'}(z_{m,m'}),\quad A_{m',m}:=\phi_{m,m'}(A_{m,m'})\\
 & z_{m,m'}^{[1]}:=z_{m},\quad z_{m,m'}^{[2]}:=\mathfrak{z}_{m,m'},\quad\phi_{m,m'}^{[1]}(z_{m,m'}):=z_{m'}\quad\phi_{m,m'}^{[2]}(z_{m,m'}):=\mathfrak{z}_{m',m}.
\end{split}
\label{eq: transdimensional model shorthand transformations}
\end{equation}
This change of variables plays a crucial r\^ole in describing and
establishing the correctness of the algorithms. 

In the following, we define the ingredients required for the AIS RJ-MCMC
algorithm and its MHAAR extension, paralleling the conditions of Propositions
\ref{prop:MHwithAISinside} and \ref{prop:extensionMHwithAISinside}.
\begin{itemize}
\item For any $m,m'\in\Theta$, we first define below the sequence of bridging
distributions $\mathscr{P}_{m,m',T}=\{\pi_{m,m',t},t=0,\ldots,T+1\}$
on the extended probability space $\big(\mathsf{Z}_{m,m'},\mathscr{Z}_{m,m'}\big)$.
First, we impose the end-point condition 
\begin{align}
\pi_{m,m',0}(z_{m,m'})\propto\pi(m,z_{m})\omega_{m,m'}(\mathfrak{z}_{m,m'})=:f_{m,m',0}\big(z_{m,m'}\big),
\end{align}
from which for any $m,m'\in\Theta$ we define $\pi_{m,m',T+1}(\cdot)$
and its unnormalised density $f_{m,m',T+1}(\cdot)$ via a change of
variable, that is for any $A_{m,m'}\in\mathscr{Z}_{m,m'}$,
\[
\pi_{m,m',T+1}(A_{m,m'}):=\pi_{m',m,0}\big(A_{m',m}\big),
\]
where we recall that $\pi_{m',m,0}(\cdot)$ has marginal $\pi_{m'}(\cdot)$.
From the associated densities one can define $f_{m,m',t}(\cdot)$
for $t=1,\ldots,T$, as discussed in earlier sections for non-trans-dimensional
setups (see also \citet{Karagiannis_and_Andrieu_2013} for a detailed
discussion). In order to satisfy \ref{enu:b} of Proposition \ref{prop:MHwithAISinside}
we further impose, noting the bijective nature of $z_{m',m}=\phi_{m,m'}\big(z_{m,m'}\big)$,
that for any $A_{m,m'}\in\mathscr{Z}_{m,m'}$
\begin{align*}
\pi_{m,m',t}(A_{m,m'})=\pi_{m',m,T+1-t}\big(A_{m',m}\big),\quad t=1,\ldots,T.
\end{align*}
It is this set of constraints which requires care, and an arbitrary
choice of parametrisation in the calculation of the Radon-Nikodym
derivative of our algorithm. The normalising constants for $f_{m,m',0}(z_{m,m'})$
and $f_{m,m',T+1}(z_{m,m'})$ are $\pi(m)$ and $\pi(m')$ respectively,
so the AIS stage of the algorithm will produce an estimate of the
ratio $\pi(m')/\pi(m)$.
\item Next, we define the AIS kernels used in the proposal mechanism. For
any $m,m'\in\Theta$ and $t=1,\ldots,T$ we let $R_{m,m',t}(\cdot,\cdot)$
be a $\pi_{m,m',t}-$reversible Markov kernels and impose the symmetry
conditions for any $(z_{m,m'},A_{m,m'})\in\mathsf{Z}_{m,m'}\times\mathscr{Z}_{m,m'}$,
\begin{equation}
R_{m,m',t}(z_{m,m'},A_{m,m'})=R_{m',m,T-t+1}(z_{m',m},A_{m',m}).\label{eq:RJ-symmetry-cond}
\end{equation}
The space bridging transition kernels $\overrightarrow{R}_{m,m',0}\colon\mathsf{Z}_{m}\times\mathscr{Z}_{m,m'}\rightarrow[0,1]$
and $\overrightarrow{R}_{m,m',T+1}\colon\mathsf{Z}_{m,m'}\times\mathscr{Z}_{m'}\rightarrow[0,1]$
are defined as
\begin{equation}
\begin{split}\overrightarrow{R}_{m,m',0}\big(z_{m},{\rm d}z_{m,m'}\big) & :=\delta_{z_{m}}\big({\rm d}z_{m,m'}^{[1]}\big)\omega_{m,m'}({\rm d}z_{m,m'}^{[2]})\\
\overrightarrow{R}_{m,m',T+1}\big(z_{m,m'},{\rm d}z_{m'}\big) & :=\delta_{\phi_{m,m'}^{[1]}(z_{m,m'})}\big({\rm d}z_{m'}\big)
\end{split}
\label{eq: trandimensional first two space bridging kernels}
\end{equation}
The other space bridging kernels $\overleftarrow{R}_{m,m',0}\colon\mathsf{Z}_{m,m'}\times\mathscr{Z}_{m}\rightarrow[0,1]$
and $\overleftarrow{R}_{m,m',T+1}\colon\mathsf{Z}_{m'}\times\mathscr{Z}_{m,m'}\rightarrow[0,1]$
will be defined from the first two above. Specifically, for any $m,m'\in\Theta$
and $\big(z_{m},z_{m'},z_{m,m'},A_{m,m'}\big)\in\mathsf{Z}_{m}\times\mathsf{Z}_{m'}\times\mathsf{Z}_{m,m'}\times\mathscr{Z}_{m,m'}$
\begin{equation}
\begin{split}\overleftarrow{R}_{m,m',0}\big(z_{m,m'},{\rm d}z{}_{m}\big) & :=\overrightarrow{R}_{m',m,T+1}\big(z_{m',m},{\rm d}z{}_{m}\big)\\
\overleftarrow{R}_{m,m',T+1}\big(z_{m'},A_{m,m'}\big) & :=\overrightarrow{R}_{m',m,0}\big(z_{m'},A_{m',m}\big)
\end{split}
.\label{eq: trandimensional last two dim-matching kernels-1}
\end{equation}
We notice the important properties, central to Green's methodology,
\begin{equation}
\begin{split}\pi_{m}({\rm d}z_{m})\overrightarrow{R}_{m,m',0}\big(z_{m},{\rm d}z_{m,m'}\big) & =\pi_{m,m',0}({\rm d}z_{m,m'})\overleftarrow{R}_{m,m',0}\big(z_{m,m'},{\rm d}z_{m}\big)\\
\pi_{m'}({\rm d}z_{m'})\overleftarrow{R}_{m,m',T+1}\big(z_{m'},{\rm d}z_{m,m'}\big) & =\pi_{m,m',T+1}({\rm d}z_{m,m'})\overrightarrow{R}_{m,m',T+1}\big(z_{m,m'},{\rm d}z_{m'}\big)
\end{split}
\label{eq: trans-dimensional model endpoint relations}
\end{equation}
so that we are in the framework described in Theorem \ref{thm:extensionAIS}
and satisfy the corresponding conditions in Proposition \ref{prop:extensionMHwithAISinside}. 
\item Finally, we define the distribution for the auxiliary variables of
AIS and the involution function. Given $m,m'\in\Theta$, define the
auxiliary variables 
\[
u_{m,m'}:=(u_{m,m',0},\ldots,u_{m,m',T})\in\mathsf{Z}_{m,m'}^{T+1}
\]
and the mapping $\varphi_{m,m'}:\mathsf{Z}_{m,m'}^{T+1}\rightarrow\mathfrak{\mathsf{Z}}_{m',m}^{T+1}$
\begin{equation}
u_{m',m}=\varphi_{m,m'}(u_{m,m'}):=\big(\phi_{m,m'}(u_{m,m',T}),\phi_{m,m'}(u_{m,m',T-1}),\ldots,\phi_{m,m'}(u_{m,m',0})\big),\label{eq: trans-dimensional bijection+involution}
\end{equation}
so that $\varphi_{m,m'}^{-1}=\varphi_{m',m}$. For $m,m'\in\Theta$
and $T\geq0$, we define the distribution for the auxiliary variables
\[
Q_{m,m',z_{m}}({\rm d}u_{m,m'}):=\overrightarrow{R}_{m,m',0}(z_{m},{\rm d}u_{m,m',0})\prod_{t=1}^{T}R_{m,m',t}(u_{m,m',t-1},{\rm d}u_{m,m',t}).
\]
\end{itemize}
Now we are ready to define the AIS RJ-MCMC algorithm. From the symmetry
conditions \eqref{eq:RJ-symmetry-cond} and \eqref{eq: trandimensional last two dim-matching kernels-1},
and our choice of $\varphi_{m,m'}$, one can establish that
\[
\bar{Q}_{m,m',z_{m}}({\rm d}u_{m,m'})=\overleftarrow{R}_{m',m,T+1}(z_{m},{\rm d}u_{m',m,T})\prod_{t=1}^{T}R_{m',m,T-t+1}(u_{m',m,T-t+1},{\rm d}u_{m',m,T-t}),
\]
where $u_{m',m,t}=\phi_{m,m'}(u_{m,m',T-t+1})$ for $t=0,1,\ldots,T$
by \eqref{eq: trans-dimensional bijection+involution}. This implies
in particular that
\begin{align}
\bar{Q}_{m',m,z_{m'}}({\rm d}u_{m',m}) & =\overleftarrow{R}_{m,m',T+1}(z_{m'},{\rm d}u_{m,m',T})\prod_{t=1}^{T}R_{m,m',T-t+1}(u_{m,m',T-t+1},{\rm d}u_{m,m',T-t}).\label{eq: AIS RJ-MCMC Q_bar}
\end{align}
The AIS RJ-MCMC algorithm of \citet{Karagiannis_and_Andrieu_2013}
uses the proposal kernel
\[
Q_{1}((m,z_{m}),{\rm d}(m',z{}_{m'},u_{m,m'})=q(m,m')Q_{m,m',z_{m}}({\rm d}u_{m,m'})\overrightarrow{R}_{m,m',T+1}(u_{m,m',T},{\rm d}z{}_{m'})
\]
and its complementary
\begin{align}
Q_{2}((m,z_{m}),{\rm d}(m',z{}_{m'},u_{m,m'}) & =q(m,m')\bar{Q}_{m,m',z_{m}}({\rm d}u_{m,m'})\overleftarrow{R}_{m',m,0}(u_{m',m,0},{\rm d}z{}_{m'}),\label{eq: AIS RJ-MCMC Q2 line1}\\
 & =q(m,m')\bar{Q}_{m,m',z_{m}}({\rm d}u_{m,m'})\overrightarrow{R}_{m,m',T+1}(u_{m,m',0},{\rm d}z{}_{m'}),\label{eq: AIS RJ-MCMC Q2 line2}
\end{align}
(We have kept \eqref{eq: AIS RJ-MCMC Q2 line2} to emphasise that
one can write (and implement) both kernels using the same auxiliary
variables $u_{m,m'}$. This will be more relevant in the MHAAR extension
in Section \ref{subsec:Extension to MHAAR} where one also samples
from $Q_{2}^{N}$, which is based on $Q_{2}$.) Equation \eqref{eq: AIS RJ-MCMC Q2 line1},
combined with \eqref{eq: AIS RJ-MCMC Q_bar} and \eqref{eq: trans-dimensional model endpoint relations},
show that we are in the framework of Theorem \ref{thm: pseudo-marginal ratio algorithms}
and Theorem \ref{thm:extensionAIS}. The acceptance rate of the AIS
RJ-MCMC can be written in terms of $m,m',z_{m}$ and $u_{m,m'}$,
leading to
\[
\mathring{r}_{u_{m,m'}}(m,m')=\frac{q(m',m)}{q(m,m')}\prod_{t=1}^{T}\frac{f_{m,m',t+1}(u_{m,m',t})}{f_{m,m',t}(u_{m,m',t})}.
\]
When $T=0$, the AIS RJ-MCMC algorithm reduces to the original RJ-MCMC
algorithm of \citet{Green_1995}.

\subsubsection{MHAAR extension of AIS RJ-MCMC \label{subsec:Extension to MHAAR}}

The MHAAR extension of AIS RJ-MCMC for averaging AIS based pseudo-marginal
ratios, that is Algorithm \ref{alg: MHAAR for pseudo-marginal ratio in latent variable models}
crafted for the trans-dimensional model, should be clear now. By analogy
to the case in general latent variable models, the proposal mechanisms
of the MHAAR extension of AIS RJ-MCMC follows immediately from the
kernels defined above as
\begin{align*}
Q_{1}^{N}((m,z_{m}),\mathrm{d}(y,\mathfrak{u}_{m,m'},k))\\
:=q(m,m') & \prod_{i=1}^{N}Q_{m,m',z_{m}}({\rm d}u_{m,m'}^{(i)})\frac{\mathring{r}_{u_{m,m'}^{(k)}}(m,m')}{\sum_{i=1}^{N}\mathring{r}_{u_{m,m'}^{(i)}}(m,m')}\overrightarrow{R}_{m,m',T+1}(u_{m,m',T}^{(k)},{\rm d}z{}_{m'}),\\
Q_{2}^{N}\big((m,z_{m});{\rm d}(y,\mathfrak{u}_{m,m'},k)\big)\\
:=q(m,m') & \frac{1}{N}\bar{Q}_{m,m',z_{m}}({\rm d}u_{m,m'}^{(k)})\overrightarrow{R}_{m,m',T+1}(u_{m,m',0}^{(k)},{\rm d}z_{m'})\prod_{i=1,i\neq k}^{N}Q_{m',m,z_{m'}}({\rm d}u_{m',m}^{(i)}),
\end{align*}
which leads to the averaged acceptance ratio $\mathring{r}_{\mathfrak{u}_{m,m'}}^{N}(m,m')=(1/N)\sum_{i=1}^{N}\mathring{r}_{u_{m,m'}^{(i)}}(m,m')$
when sampling from $Q_{1}^{N}(\cdot)$ and $\mathring{r}_{\mathfrak{u}_{m,m'}}^{N}(m',m)^{-1}$
when sampling from $Q_{2}^{N}(\cdot)$. 

As discussed in Section \ref{subsec: Choosing Q1 and Q2 with different probabilities},
it is possible to choose between the two proposal mechanisms with
a probability dependent on the current and part of the proposed states,
in contrast with the $1/2-1/2$ default choice above, leading to modified
acceptance ratios of the form given in \eqref{eq: acceptance ratio modified by alpha}.
We discuss here how this can be taken advantage of for computational
purposes. Assume for simplicity of exposition that only moves from
model $m$ to models $m-1$ and $m+1$ are allowed (for $m$ such
that these moves are valid). As illustrated below, it may be sensible
to use $Q_{1}^{N}(\cdot)$ rather than $Q_{2}^{N}(\cdot)$ to increase
the model index and vice-versa to decrease the model index. This can
be achieved for example by setting $\beta(m,m+1)=1$ and $\beta(m,m-1)=0$.
A scenario where this is a potentially good idea is for example when
$\mathfrak{z}_{m,m+1}$ can take values on a continuum while $\mathfrak{z}_{m+1,m}$
can only take a finite number of values, say $c_{m+1,m}$. Generating
$N\gg c_{m+1,m}$ copies of $\mathfrak{z}_{m+1,m}$ and averaging
may be wasteful in comparison to the generation of $N$ values of
$\mathfrak{z}_{m,m+1}$. Using the strategy above one can ensure that
$Q_{1}^{N}(\cdot)$ is used when ``going up'' while $Q_{2}^{N}(\cdot)$
is only used to ``go down''. This is the case for the Poisson change-point
model example of the next section. 

In Algorithm \ref{alg: Reversible multiple jump MCMC} we present
the implementation of a particular version of this algorithm, for
a general value $\beta(m,m')$, $T=0$ and $N\geq1$. Because of its
similarity to the RJ-MCMC of \citet{Green_1995} but with the difference
of generating multiple auxiliary variables (hence multiple jumps)
instead of one, we call this algorithm Reversible-multiple-jump MCMC
(RmJ-MCMC).

\begin{algorithm}
\caption{RmJ-MCMC: Algorithm \ref{alg: MHAAR for pseudo-marginal ratio in latent variable models}
for trans-dimensional models with $T=0$ (no annealing). }
\label{alg: Reversible multiple jump MCMC} 

\KwIn{Current sample $X_{n}=(m,z_{m})$}

\KwOut{New sample $X_{n+1}$}

Sample $m'\sim q(m,\cdot)$ and $v\sim\mathcal{U}(0,1)$.\\
\If{$v\leq\beta(m,m')$}{ 

\For{$i=1,\ldots,N$}{

Sample $\mathfrak{z}_{m,m'}^{(i)}\sim\omega_{m,m'}(\cdot)$, set $u_{m,m'}^{(i)}=(z_{m},\mathfrak{z}_{m,m'}^{(i)})$.} 

Sample $k\sim\mathcal{P}\big(\mathring{r}_{u^{(1)}}(m,m'),\ldots,\mathring{r}{}_{u^{(N)}}(m,m')\big)$,
set $z'_{m'}=\phi_{m,m'}^{[1]}(u_{m,m'}^{(k)})$.\\
Set $X_{n+1}=(m',z'{}_{m'})$ with probability $\min\{1,\mathring{\bar{r}}_{\mathfrak{u}_{m,m'}}^{N}(m,m')\},$
otherwise set $X_{n+1}=X_{n}$.

}\Else{ 

Sample $\mathfrak{z}_{m,m'}^{(1)}\sim\omega_{m,m'}(\cdot)$, set $u_{m',m}^{(1)}=\phi_{m,m'}(z_{m},\mathfrak{z}_{m,m'}^{(1)})$
and $z'_{m'}=\phi_{m,m'}^{[1]}(u_{m,m'}^{(k)})$. \\
\For{$i=2,\ldots,N$}{

Sample $\mathfrak{z}_{m',m}^{(i)}\sim\omega_{m',m}(\cdot)$, set $u_{m',m}^{(i)}=(z{}_{m'},\mathfrak{z}_{m',m}^{(i)})$.

}

Set $X_{n+1}=(m',z'{}_{m'})$ with probability $\min\{1,\mathring{\bar{r}}_{\mathfrak{u}_{m',m}}^{N}(m',m)^{-1}\}$,
otherwise set $X_{n+1}=X_{n}$. 

}
\end{algorithm}

\subsubsection{Numerical example: the Poisson multiple change-point model \label{subsec: Numerical example: Poisson multiple changepoint model}}

The Poisson multiple change-point model was proposed for the analysis
of the coal-mining disasters in \citet{Green_1995}. The model assumes
that $n$ data points $y_{1},\ldots,y_{n}$, which are the times of
occurrence of disasters with the choice $y_{0}=0$, arise from a non-homogenous
Poisson process model on a time interval $[0,L]$ with intensity modelled
as a step function with an unknown number of steps $m$ having unknown
starting points $0=s_{0,m}<s_{1,m},\ldots<s_{m,m}=L$ and unknown
heights $h_{1,m},\ldots,h_{m,m}$. We will refer to the model involving
$m$ steps as model $m$. Therefore, denoting $\omega_{m}=(\{s_{j,m}\}_{j=0}^{m},\{h_{j,m}\}_{j=1}^{m})$,
the data likelihood under model $m$ is 
\[
\log\mathcal{L}(y_{1:n};\omega_{m})=\sum_{j=1}^{m}h_{j,m}\log\left(\sum_{i=1}^{n}\mathbb{I}_{[s_{j-1,m},s_{j,m})}(y_{i})\right)-\sum_{j=1}^{m}h_{j,m}(s_{j,m}-s_{j-1,m}).
\]
The prior distribution for $\phi_{m}$ is as follows: $\{s_{j,m}\}_{j=1}^{m-1}$
are distributed as the even-numbered order statistics from $2m-1$
points uniformly distributed on $[0,L]$; the heights $h_{j,m}$,
$j=1,\ldots,m$ are independent and each follow a Gamma distribution
$\mathcal{G}(\alpha_{k},\beta_{k})$, where $\alpha_{k}$ and $\beta_{k}$
themselves are independent random variables admitting distributions
$\mathcal{G}(c,d)$ and $\mathcal{G}(e,f)$, respectively. Finally,
the prior distribution for $m$ is a truncated Poisson distribution
$\mathcal{P}_{m_{\max}}(\lambda)$ where $m\leq m_{\max}\geq1$. The
hyper parameters $(c,d,e,f,\lambda,m_{\max})$ are assumed known,
we let $\Theta=\{1,\ldots,m_{\max}\}$ and 
\[
z_{m}=(\omega_{m},\alpha_{m},\beta_{m})\in\mathsf{Z}_{m}=(0,L)^{m-1}\times(0,\infty)^{m}\times(0,\infty)\times(0,\infty)
\]
be the within-model parameters of model $m$. This defines a trans-dimensional
distribution $\pi(m,dz_{m})$ on $\mathsf{X}=\cup_{m\in\Theta}\{m\}\times\mathsf{Z}_{m}$
where the dimension $d_{m}$ of $\mathsf{Z}_{m}$ depends on $m$.
The distribution $\pi(m,\mathrm{d}z_{m})$ admits a density $\pi(m,z_{m})$
known up to a normalising constant; this unnormalised density can
easily be derived from the description of the model above.

Our experiment on the Poisson change-point model focuses on showing
that improvement over standard RJ-MCMC can be obtained solely by using
asymmetric MCMC with multiple dimension matching variables (as discussed
in the paragraph above); hence we run RmJ-MCMC in Algorithm \ref{alg: Reversible multiple jump MCMC}
for several values of $N$ and $T=0$. Each run generates $K=10^{6}$
samples of which the last $3K/4$ are used to compute the IAC for
$m$. Note that we also include an MCMC move for the within model
variables $z_{m}$ at every iteration in order to ensure irreducibility,
the details of this move can be found given in \citet{Karagiannis_and_Andrieu_2013}.
In order to illustrate the gains in terms of convergence to equilibrium
of our scheme we ran $3000$ independent realisations of the algorithm
started at the same point $x_{0}$ and estimated the expectations
of $f_{m}(X_{i}):=\mathbb{I}\{M_{t}=m\}$, that is $\mathbb{E}_{x_{0}}^{N}\big[f_{m}(X_{i})\big]$,
by an ensemble average and report $\big|\hat{\pi}(m)-k^{-1}\sum_{k=1}^{3000}f_{m}(X_{t}^{(k)})\big|$
for $m\in\{1,\ldots,8\}$ and $N=1,10,100$ in Figure \ref{fig:reversible-jump-burnin}
where $\hat{\pi}(m)$ was estimated by a realisation of length $10^{6}$
with $N=90$ and $T=50$, discarding the burn-in. We see that the
approach reduces time to convergence to equilibrium by the order of
$50\%$, while variance reduction is automatic and of the order of
$60\%$ as illustrated in Figure \ref{fig: reversiblejumpexample}.
We also provide results for the AIS scheme for illustration.

\begin{figure}
\includegraphics[width=1\textwidth]{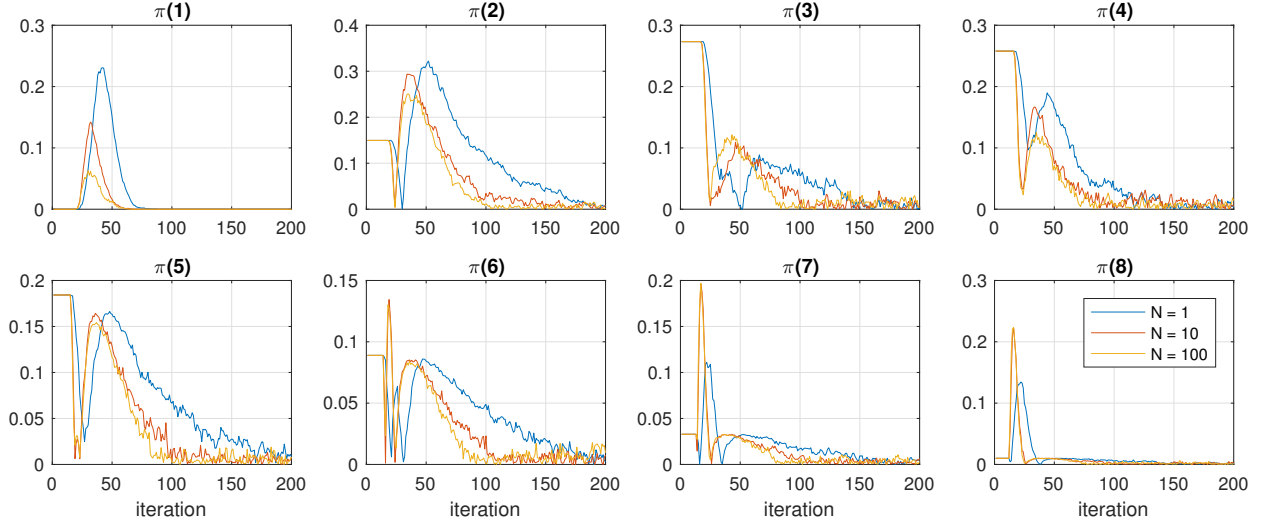}

\caption{Estimates of time to convergence of $\mathbb{E}_{x_{0}}^{N}\big[f_{m}(X_{i})\big]$
to $\pi(m)$ for $N=1,10,100$.}
\label{fig:reversible-jump-burnin}
\end{figure}

\begin{figure}
\centerline{\includegraphics[scale=0.6]{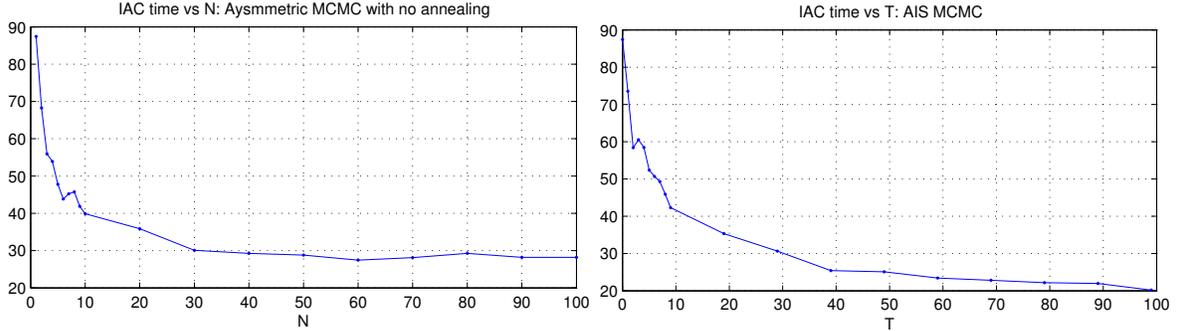}}
\protect\protect\caption{Left: IAC for $m$ vs number of particles $N=1,2,\ldots,10,20,\ldots,100$
with $T=0$. Right: IAC for $m$ vs number of particles $T=0,1,2,\ldots,10,20,\ldots,100$
with $N=1$.}

\label{fig: reversiblejumpexample} 
\end{figure}

\section{State-space models: SMC and cSMC within MHAAR \label{sec: State-space models: SMC and conditional SMC within MHAAR}}

In Section \ref{sec: Pseudo-marginal ratio algorithms for latent variable models}
we have shown how the generic MHAAR strategy which consists of averaging
independent estimates of the acceptance ratio could be helpful in
the context of inference in state-space models. Here we present an
alternative where dependent estimates arising from a single conditional
SMC algorithm can be averaged in order to improve performance.

\subsection{State-space models and cSMC \label{subsec: State-space models and conditional SMC} }

In its simplest form, a state-space model is comprised of a latent
Markov chain $\{Z_{t};t\geq1\}$ taking its values in some measurable
space $(\mathsf{Z},\mathcal{Z})$ and observations $\{Y_{t};t\geq1\}$
taking values in $(\mathsf{Y},\mathcal{Y})$. The latent process has
initial probability with density $\mu_{\theta}(z_{1})$ and transition
density $f_{\theta}(z_{t-1},z_{t})$, dependent on a parameter $\theta\in\Theta\subset\mathbb{R}^{d_{\theta}}$.
An observation at time $t$ is assumed conditionally independent of
all other random variables given $Z_{t}=z_{t}$ and its conditional
observation density is $g_{\theta}(z_{t},y_{t})$. The corresponding
joint density of the latent and observed variables up to time $T$
is 
\begin{equation}
p_{\theta}(z_{1:T},y_{1:T})=\mu_{\theta}(z_{1})\prod_{t=2}^{T}f_{\theta}(z_{t-1},z_{t})\prod_{t=1}^{T}g_{\theta}(z_{t},y_{t}),\label{eq: HMM joint density}
\end{equation}
from which the likelihood function associated to the observations
$y_{1:T}$ can be obtained 
\begin{equation}
l_{\theta}(y_{1:T}):=\int_{\mathsf{X}^{T}}p_{\theta}(z_{1:T},y_{1:T}){\rm d}z_{1:T}.\label{eq:likelihoodSSM}
\end{equation}
Note that the densities $f_{\theta}$ and $g_{\theta}$ could also
depend on $t$, at the expense of notational complications, and that
$T$ is here the time horizon of the time series and should not be
confused with the number of intermediate steps in AIS in the previous
sections. We allow this abuse of notation since there are no intermediate
steps involved in the methodology for HMMs developed in this paper.

In order to go back to our generic notation, we let $z=z_{1:T}$ and
$y=y_{1:T}$. With a prior distribution $\eta(\mathrm{d}\theta)$
on $\theta$ with density $\eta(\theta)$, the joint posterior $\pi(\mathrm{d}(\theta,z))$
has the density
\[
\pi(\theta,z)\propto\eta(\theta)p_{\theta}(z,y)
\]
so that $\pi(\theta)\propto\eta(\theta)\ell_{\theta}(y)$ and $\pi_{\theta}(z):=p_{\theta}(z\mid y)=p_{\theta}(z,y)/\ell_{\theta}(y)$.

The conditional sequential Monte Carlo (cSMC) algorithm for this state-space
model is given in Algorithm \ref{alg: Conditional SMC}, where particles
are initialised using distribution $h_{\theta}(\cdot)$ on $\mathsf{(Z},\mathcal{Z})$
at time $1$ and propagated at times $t>1$ using the transition kernel
$H_{\theta}(\cdot,\cdot)$ on $\mathsf{(Z},\mathcal{Z})$. The cSMC
algorithm is an MCMC transition probability, akin to particle filters,
particularly well suited to sampling from $\pi_{\theta}(\mathrm{d}z)$
\citet{Andrieu_et_al_2010}. It was recently shown in \citet{Lindsten_and_Schon_2012}
that cSMC with backward sampling \citet{whiteley2010discussion} can
be used efficiently as part of a more elaborate Metropolis-within-Particle
Gibbs algorithm in order to sample from the posterior distribution
$\pi(\mathrm{d}(\theta,z))$; see Algorithm \ref{alg: Metropolis within particle Gibbs}.

\begin{algorithm}
\caption{$\mathrm{cSMC}\big(M,\theta,z\big)$}
\label{alg: Conditional SMC}

\KwIn{Number of particles $M$, parameter $\theta$, current sample
$z$}

\KwOut{Particles $\zeta=\zeta_{1:T}^{(1:M)}$, new sample $z'$}

Set $\zeta_{1}^{(1)}=z_{1}$.

\For{ $i=2,\ldots,M$}{

Sample $\zeta_{1}^{(i)}\sim h_{\theta}(\cdot)$.\\
Compute $w_{1}^{(i)}=\mu_{\theta}\big(\zeta_{1}^{(i)}\big)g_{\theta}\big(\zeta_{1}^{(i)},y_{1}\big)/h_{\theta}\big(\zeta_{1}^{(i)}\big)$.

}

\For{ $t=2,\ldots,T$}{

Set $\zeta_{t}^{(1)}=z_{t}$.

\For{ $i=2,\ldots,M$}{

Sample $a_{t-1}^{(i)}\sim\mathcal{P}\big(w_{t-1}^{(1)},\ldots,w_{t-1}^{(M)}\big)$
and $\zeta_{t}^{(i)}\sim H_{\theta}\big(\zeta_{t-1}^{(a_{t-1}^{(i)})},\cdot\big)$.\\
Compute $w_{t}^{(i)}=f_{\theta}\big(\zeta_{t-1}^{(a_{t-1}^{(i)})},\zeta_{t}^{(i)}\big)g_{\theta}\big(\zeta_{t}^{(i)},y_{t}\big)/H_{\theta}\big(\zeta_{t-1}^{(a_{t-1}^{(i)})},\zeta_{t}^{(i)}\big)$.

}

}

Sample $k_{T}\sim\mathcal{P}\big(w_{T}^{(1)},\ldots,w_{T}^{(M)}\big)$
and set $z'_{T}=\zeta_{T}^{(k_{T})}$.

\label{line:beginBS}\For{ $t=T-1,\ldots,1$}{

\For{ $i=1,\ldots,M$}{

Compute $\tilde{w}_{t}^{(i)}=w_{t}^{(i)}f_{\theta}\big(\zeta_{t}^{(i)},\zeta_{t+1}^{(k_{t+1})}\big)$.

}

Sample $k_{t}\sim\mathcal{P}\big(\tilde{w}_{t}^{(1)},\ldots,\tilde{w}_{t}^{(M)}\big)$
and set $z'_{t}=\zeta_{t}^{(k_{t})}$.

}

\label{line:BSend}

\Return$\zeta=\zeta_{1:T}^{(1:N)}$ and $z'=z'_{1:T}$.
\end{algorithm}

\begin{algorithm}
\caption{Metropolis-within-particle Gibbs}
\label{alg: Metropolis within particle Gibbs}

\KwIn{Current sample $X_{n}=(\theta,z)$}

\KwOut{New sample $X_{n+1}$}

Sample $z'\sim\mathrm{cSMC}\big(M,\theta,z\big)$.\\
Sample $\theta'\sim q(\theta,\cdot)$. \\
Return $X_{n+1}=(\theta',z')$ with probability 
\begin{equation}
\min\left\{ 1,\frac{\eta(\theta')p_{\theta'}(z,y)q(\theta',\theta)}{\eta(\theta)p_{\theta}(z,y)q(\theta,\theta')}\right\} ;\label{eq: noisy acceptance ratio}
\end{equation}
otherwise return $X_{n+1}=(\theta,z')$. 
\end{algorithm}

Retaining one path from the $T\times M$ samples in the cSMC algorithm
involved in Algorithm \ref{alg: Metropolis within particle Gibbs}
may seem to be wasteful, and a natural idea is whether it is possible
to make use of multiple, or even use all possible, trajectories and
average out the corresponding acceptance ratios \eqref{eq: noisy acceptance ratio}.
We show that this is indeed possible with Algorithms \ref{alg: MHAAR for SSM with cSMC - Rao-Blackwellised backward sampling}
and \ref{alg: MHAAR for SSM with cSMC - multiple paths from backward sampling}
in the next section. We then show that these schemes improve performance
at a cost which can be negligible, in particular when a parallel computing
architecture is available. In order to avoid notational overload we
postpone the justification of the algorithms to Appendix \ref{sec: Auxiliary results and proofs for cSMC based algorithms }.
Algorithms \ref{alg: MHAAR for SSM with cSMC - Rao-Blackwellised backward sampling}
and \ref{alg: MHAAR for SSM with cSMC - multiple paths from backward sampling}
are alternative to the recently developed method

\subsection{MHAAR with cSMC for state-space models \label{subsec: MHAAR with cSMC for SSM}}

The law of the indices $\mathbf{k}:=(k_{1},\ldots,k_{T})$ drawn
in the backward sampling step in Algorithm \ref{alg: Conditional SMC}
(lines \ref{line:beginBS}-\ref{line:BSend}) conditional upon $\zeta=\zeta_{1:T}^{(1:M)}$
is given by

\[
\phi_{\theta}(\mathbf{k}\mid\zeta):=\frac{w_{T}(\zeta_{T}^{(k_{T})})}{\sum_{i=1}^{M}w_{T}(\zeta_{T}^{(i)})}\prod_{t=1}^{T-1}\frac{w_{t}(\zeta_{t}^{(k_{t})})f_{\theta}(\zeta_{t}^{(k_{t})},\zeta_{t+1}^{(k_{t+1})})}{\sum_{i=1}^{M}w_{t}(\zeta_{t}^{(i)})f_{\theta}(\zeta_{t}^{(i)},\zeta_{t+1}^{(k_{t+1})})}.
\]
We introduce the Markov kernel which corresponds to the sampling of
a trajectory $z$ with backward-sampling, conditional upon $\zeta$,
\[
\check{\Phi}_{\theta}(\zeta,{\rm d}z)=\sum_{k\in[M]^{T}}\phi_{\theta}(\mathbf{k}|\zeta)\delta_{\zeta^{(k)}}({\rm d}z),
\]
where we define $[M]=\{1,\ldots,M\}$ and $\zeta^{(\mathbf{k})}:=(\zeta_{1}^{(k_{1})},\ldots,\zeta_{T}^{(k_{T})})$.
Further, for any $\theta,\theta',\tilde{\theta}\in\Theta$, and $z,z'\in\mathsf{Z}^{T}$,
define 
\begin{equation}
\mathring{r}_{z,z'}(\theta,\theta';\tilde{\theta})=\frac{q(\theta',\theta)\eta(\theta')p_{\theta'}(z',y)p_{\tilde{\theta}}(z,y)}{q(\theta,\theta')\eta(\theta)p_{\tilde{\theta}}(z',y)p_{\theta}(z,y)}.\label{eq: AIS acceptance ratio for SSM}
\end{equation}
In the following, we show that it is possible to construct unbiased
estimators of $r(\theta,\theta')$ using cSMC, provided we have a
random sample $z\sim\pi_{\theta}(\cdot)$. Specifically, this is achieved
as the expected value of $\mathring{r}_{z,\zeta^{(\mathbf{k})}}(\theta,\theta';\tilde{\theta})$
with respect to the backward sampling distribution on $\mathbf{k}$,
\begin{equation}
\mathring{r}_{z,\zeta}(\theta,\theta';\tilde{\theta}):=\sum_{\mathbf{k}\in[M]^{T}}\phi_{\tilde{\theta}}(\mathbf{k}|\zeta)\mathring{r}_{z,\zeta^{(\mathbf{k})}}(\theta,\theta';\tilde{\theta}).\label{eq: SMC acceptance ratio estimator all paths}
\end{equation}
\begin{thm}
\label{thm: SMC unbiased estimator of acceptance ratio}For $\theta,\theta',\tilde{\theta}\in\Theta$
and any $M\geq1$, let $z\sim\pi_{\theta}(\cdot)$, $\zeta|z\sim\mathrm{cSMC}(M,\tilde{\theta},z)$
be the generated particles from the cSMC algorithm targeting $\pi_{\tilde{\theta}}(\cdot)$
with $M$ particles, conditioned on $z$. Then, \textup{$\mathring{r}_{z,\zeta}(\theta,\theta';\tilde{\theta})$}
is an unbiased estimator of $r(\theta,\theta')$.
\end{thm}
Theorem \ref{thm: SMC unbiased estimator of acceptance ratio} is
original to the best of our knowledge and we find it interesting in
several aspects. Firstly, unlike the estimator in Metropolis-within-Particle
Gibbs (Algorithm \ref{alg: Metropolis within particle Gibbs}), the
estimators in Theorem \ref{thm: SMC unbiased estimator of acceptance ratio}
use all possible paths from the particles generated by the cSMC. Also,
with a slight modification one can similarly obtain unbiased estimators
for $\pi(\theta')/\pi(\theta)$ which is in some applications of primary
interest. The theorem is derived from \citet[Theorem 5.2]{del2010backward}
and the results in \citet{Andrieu_et_al_2010} relating the laws of
cSMC and SMC. The proof of the theorem is left to Appendix \ref{sec: Auxiliary results and proofs for cSMC based algorithms }.

In particular, Theorem \ref{thm: SMC unbiased estimator of acceptance ratio}
motivates us to design an asymmetric MCMC algorithm which uses the
unbiased estimator mentioned in the theorem in its acceptance ratios.
We present Algorithm \ref{alg: MHAAR for SSM with cSMC - Rao-Blackwellised backward sampling}
that is developed with this motivation. The algorithm requires a pair
of functions $\tilde{\theta}_{1}:\Theta^{2}\rightarrow\Theta$ and
$\tilde{\theta}_{2}:\Theta^{2}\rightarrow\Theta$ that satisfy $\tilde{\theta}_{1}(\theta,\theta')=\tilde{\theta}_{2}(\theta',\theta)$,
in order to determine the intermediate parameter value at which cSMC
is run. 

\begin{algorithm}
\caption{MHAAR for state-space models with cSMC-based estimator of the likelihood
ratio}

\label{alg: MHAAR for SSM with cSMC - Rao-Blackwellised backward sampling}

\KwIn{Current sample $X_{n}=(\theta,z),$ number of particles $M\geq1$}

\KwOut{New sample $X_{n+1}$}

Sample $\theta'\sim q(\theta,\cdot)$ and $v\sim\mathcal{U}(0,1)$\\
\If { $v\leq1/2$ }{

Set $\tilde{\theta}=\tilde{\theta}_{1}(\theta,\theta')$.\\
Run a cSMC$(M,\tilde{\theta},z)$ targeting $\pi_{\tilde{\theta}}$
conditional on $z$ to obtain $\zeta$.\\
Sample $\mathbf{k}=(k_{1},\ldots,k_{T})$ with probability $\frac{\phi_{\tilde{\theta}}(\mathbf{k}|\zeta)\mathring{r}_{z,\zeta^{(\mathbf{k})}}(\theta,\theta';\tilde{\theta})}{\sum_{\mathbf{l}\in[M]^{T}}\phi_{\tilde{\theta}}(\mathbf{l}|\zeta)\mathring{r}_{z,\zeta^{(\mathbf{l})}}(\theta,\theta';\tilde{\theta})}$
and set $z'=\zeta^{(\mathbf{k})}$. \label{alg line: asymmetric MCMC for HMM selection of z_prime}\\
Set $X_{n+1}=(\theta',z')$ with probability $\min\{1,\mathring{r}_{z,\zeta}(\theta,\theta';\tilde{\theta})\}$;
otherwise reject the proposal and set $X_{n+1}=(\theta,z)$.

}\Else{

Set $\tilde{\theta}=\tilde{\theta}_{2}(\theta,\theta')$.\\
Run a ${\rm cSMC}(M,\tilde{\theta},z)$ targeting $\pi_{\tilde{\theta}}$
conditional on $u^{(1)}=z$ to obtain $\zeta.$\\
Sample $\mathbf{k}=(k_{1},\ldots,k_{T})$ with probability $\phi_{\tilde{\theta}}(\mathbf{k}|\zeta)$
and set $z'=\zeta^{(\mathbf{k})}$.\\
Set $X_{n+1}=(\theta',z')$ with probability $\min\{1,1/\mathring{r}_{z',\zeta}(\theta',\theta;\tilde{\theta})\}$;
otherwise reject the proposal and set $X_{n+1}=(\theta,z)$. 

}
\end{algorithm}

The proof that Algorithm \ref{alg: MHAAR for SSM with cSMC - Rao-Blackwellised backward sampling}
is reversible is established in Appendix \ref{subsec: Proof of reversibility for Algorithms}.
The proof has two interesting by-products: (i) An alternative proof
of Theorem \ref{thm: SMC unbiased estimator of acceptance ratio},
and (ii) another unbiased estimate of $r(\theta,\theta')$ that uses
all possible paths that can be constructed from the particles generated
by a cSMC, which is stated in the following corollary.
\begin{cor}
\label{cor: SMC unbiased estimator of acceptance ratio}For $\theta,\theta',\tilde{\theta}\in\Theta$
and any $M\geq1$, let $z\sim\pi_{\theta}(\cdot)$, $\zeta|z\sim{\rm cSMC}(M,\tilde{\theta},z)$
be the generated particles from the cSMC algorithm with $M$ particles
at $\tilde{\theta}$ conditioned on $z$, and $z'|\zeta\sim\check{\Phi}_{\tilde{\theta}}(\zeta,\cdot)$.
Then, $1/\mathring{r}_{z',\zeta}(\theta',\theta,\tilde{\theta})$
is an unbiased estimator of $r(\theta,\theta').$
\end{cor}

\subsection{Reduced computational cost via subsampling \label{subsec: Easing computational burden with subsampling: Multiple paths BS-SMC}}

The computations needed to implement Algorithm \ref{alg: MHAAR for SSM with cSMC - Rao-Blackwellised backward sampling}
can be performed with a complexity of $\mathcal{O}(M^{2}T)$ upon
observing that the unnormalised probability can be written as
\[
\phi_{\tilde{\theta}}(\mathbf{k}|\zeta)\mathring{r}_{z,\zeta^{(\mathbf{k})}}(\theta,\theta';\tilde{\theta})=:\kappa_{z,\zeta}(\mathbf{k})=\kappa_{z,\zeta,1}(k_{1})\prod_{t=2}^{T}\kappa_{z,\zeta,t}(k_{t-1},k_{t})
\]
for an appropriate choice for the functions $\kappa_{z,\zeta,t}$.
Indeed, the expression above implies that computation of $\mathring{r}_{z,\zeta}(\theta,\theta';\tilde{\theta})=\sum_{\mathbf{k}\in[M]^{T}}\kappa_{z,\zeta}(\mathbf{k})$
can be performed by a sum-product algorithm and sampling $\mathbf{k}$
with probability proportional to $\kappa_{z,\zeta}(\mathbf{k})$ can
be performed with a forward-filtering backward-sampling algorithm.
However, $\mathcal{O}(M^{2}T)$ can still be overwhelming, especially
when $M$ is large. 

In the following, we introduce a computationally less demanding version
of Algorithm \ref{alg: MHAAR for SSM with cSMC - Rao-Blackwellised backward sampling}
which uses a subsampled version of \eqref{eq: SMC acceptance ratio estimator all paths}
obtained from $N$ paths drawn using backward sampling and still preserves
reversibility. Letting $\mathfrak{u}=(u^{(1)},\ldots,u^{(N)})\in\mathsf{Z}^{TN}$
, consider 
\[
\mathring{r}_{z,\mathfrak{u}}^{N}(\theta,\theta';\tilde{\theta})=\frac{1}{N}\sum_{i=1}^{N}\mathring{r}_{z,u^{(i)}}(\theta,\theta';\tilde{\theta}),
\]
which is an unbiased estimator of \eqref{eq: SMC acceptance ratio estimator all paths}
when $u^{(1)},\ldots,u^{(N)}\overset{{\rm iid}}{\sim}\check{\Phi}(\zeta,\cdot)$.
In Algorithm \ref{alg: MHAAR for SSM with cSMC - multiple paths from backward sampling}
we present the multiple paths BS-cSMC asymmetric MCMC algorithm which
uses $\mathring{r}_{z,\mathfrak{u}}^{N}(\theta,\theta';\tilde{\theta})$,
but still targets $\pi(\mathrm{d}(\theta,z))$, as desired. The computational
complexity of this algorithm is $\mathcal{\mathcal{O}}(NMT)$ per
iteration instead of $\mathcal{O}(M^{2}T)$; moreover, sampling $N$
paths can be parallelised. Reversibility of the algorithm with respect
to $\pi(\mathrm{d}(\theta,z))$ is proved in Appendix \ref{subsec: Proof of reversibility for Algorithms}.

\begin{algorithm}[!h]
\caption{MHAAR for state-space models with cSMC-based estimator of the likelihood
ratio - with reduced computation via subsampling}
\label{alg: MHAAR for SSM with cSMC - multiple paths from backward sampling}

\KwIn{Current sample $X_{n}=(\theta,z)$, number of particles $M\geq1$,
number of backward paths $N\geq1$} 

\KwOut{New sample $X_{n+1}$} 

Sample $\theta'\sim q(\theta,\cdot)$ and $v\sim\mathcal{U}(0,1)$.
\\
\If{$v\leq1/2$}{

Set $\tilde{\theta}=\tilde{\theta}_{1}(\theta,\theta')$.\\
Run a ${\rm cSMC}(M,\tilde{\theta},z)$ to obtain the particles $\zeta$.
\\
Draw $N$ paths with backward sampling, $u^{(1)},\ldots,u^{(N)}\overset{{\rm iid}}{\sim}\check{\Phi}_{\tilde{\theta}}(\zeta,\cdot)$.\\
Sample $k\sim\mathcal{P}\big(\mathring{r}_{z,u^{(1)}}(\theta,\theta';\tilde{\theta}),\ldots,\mathring{r}_{z,u^{(N)}}(\theta,\theta';\tilde{\theta})\big)$
and set $z'=u^{(k)}$. \\
Set $X_{n+1}=(\theta',z')$ with probability $\min\{1,\mathring{r}_{z,\mathfrak{u}}^{N}(\theta,\theta';\tilde{\theta})\}$;
otherwise reject and set $X_{n+1}=(\theta,z)$.

}\Else{ 

Set $\tilde{\theta}=\tilde{\theta}_{2}(\theta,\theta')$.\\
Sample $k$ uniformly from $\{1,\ldots,N\}$ and set $u^{(k)}=z$.\\
Run a ${\rm cSMC}(M,\tilde{\theta},z)$ to obtain particles $\zeta$.
\\
Draw $N$ paths with backward sampling $u^{(1)},\ldots,u^{(k-1)},z',u^{(k+1)},\ldots,u^{(N)}\overset{{\rm iid}}{\sim}\check{\Phi}_{\tilde{\theta}}(\zeta,\cdot)$.\\
Set $X_{n+1}=(\theta',z')$ with probability $\min\{1,1/\mathring{r}_{z',\mathfrak{u}}^{N}(\theta',\theta;\tilde{\theta})\}$;
otherwise reject and set $X_{n+1}=(\theta,z)$. 

}
\end{algorithm}
\begin{example}
We consider the non-linear state-space of Example \ref{ex: state-space example for latent variable section}
for the same set-up. We conducted experiments similar to those of
Example \ref{ex: ctd state-space example for latent variable section},
but using this time Algorithm \ref{alg: MHAAR for SSM with cSMC - multiple paths from backward sampling}
instead, for $N=1$, $N=10$, $N=100$ and $M=150$ particles. The
intermediate distribution used was similar, as were the various proposal
distributions. The results for convergence and IAC times are shown
in Figures \ref{fig: convergence vs N for the state-space model-2}
and \ref{fig: IAC time vs N for the state-space model-2} where the
results from Example \ref{ex: state-space example for latent variable section}
are repeated in order to ease comparison. (Note that, assuming perfect
parallelisation and that the computation time of cSMC is proportional
to the number of particles, Algorithm \ref{alg: MHAAR for SSM with cSMC - multiple paths from backward sampling}
with $M=150$ particles and Algorithm \ref{alg: MHAAR for pseudo-marginal ratio in latent variable models}
with $M=100$ particles are equally costly. This is because of the
non-parallelisable part of $Q_{2}$ of Algorithm \ref{alg: MHAAR for pseudo-marginal ratio in latent variable models}.)

\begin{figure}[h]
\centerline{\includegraphics[scale=0.7]{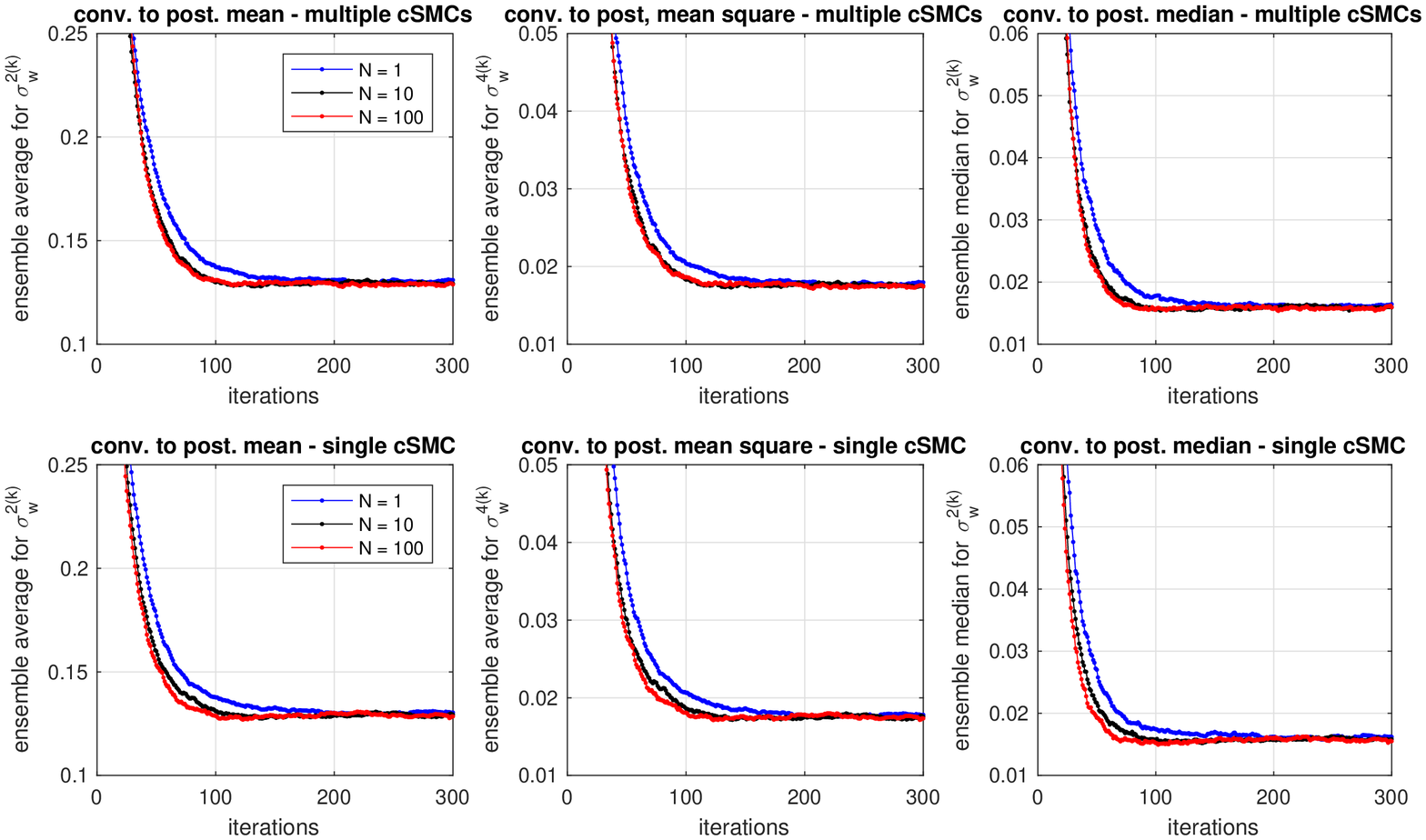}}
\protect\protect\caption{Convergence results for $\theta=(\sigma_{v}^{2},\sigma_{w}^{2})$
vs $N$ in Algorithm \ref{alg: MHAAR for SSM with cSMC - multiple paths from backward sampling}
in comparison with Algorithm \ref{alg: MHAAR for pseudo-marginal ratio in latent variable models}.}

\label{fig: convergence vs N for the state-space model-2} 
\end{figure}

\begin{figure}[h]
\centerline{\includegraphics[scale=0.7]{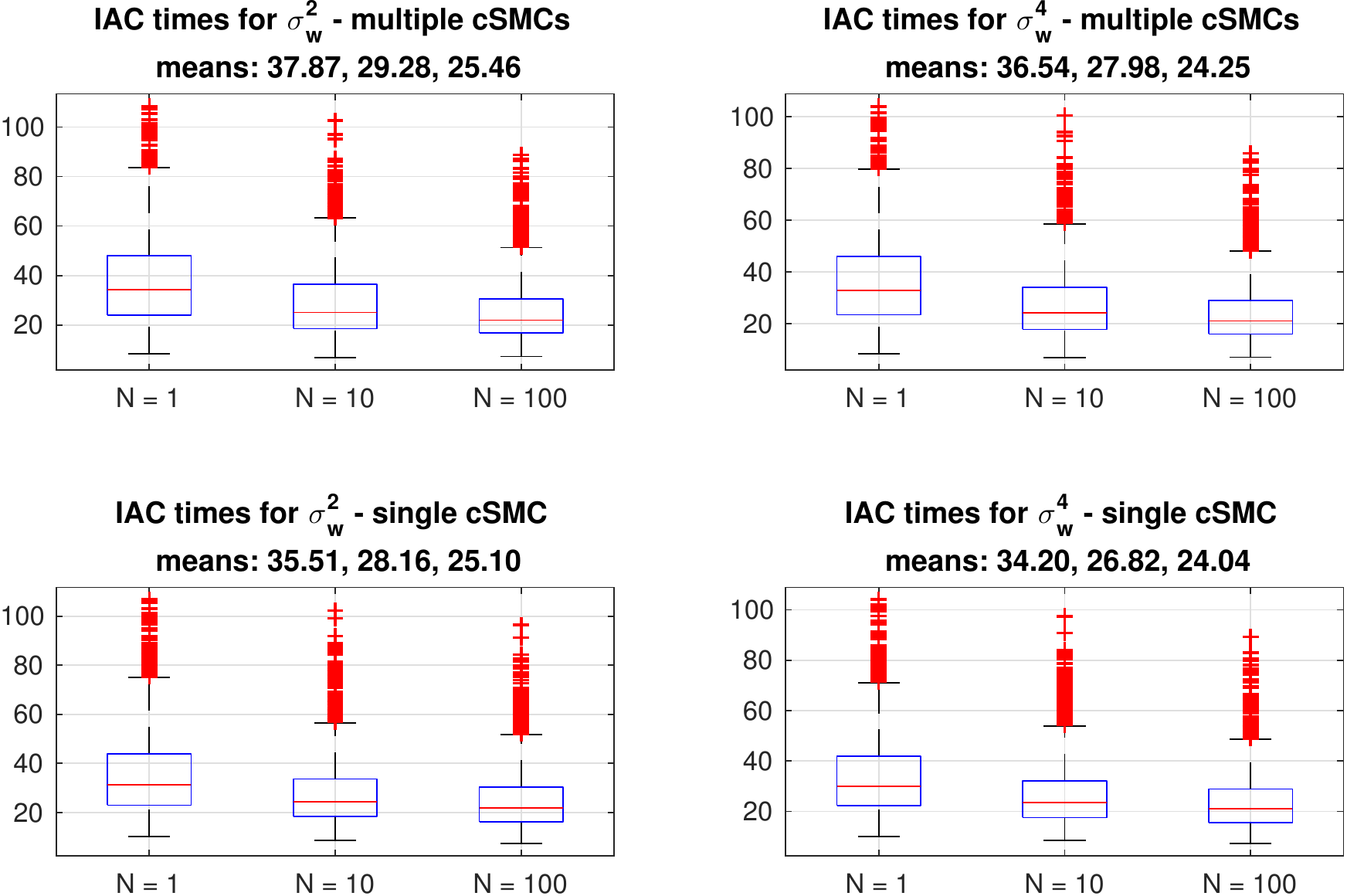}}
\protect\protect\caption{IAC times for $\theta=(\sigma_{v}^{2},\sigma_{w}^{2})$ vs $N$ in
Algorithm \ref{alg: MHAAR for SSM with cSMC - multiple paths from backward sampling}
in comparison with Algorithm \ref{alg: MHAAR for pseudo-marginal ratio in latent variable models}.}

\label{fig: IAC time vs N for the state-space model-2} 
\end{figure}
\end{example}
\begin{example}
 In this experiment, the true parameters are $\sigma_{v}^{2}=10$
and $\sigma_{w}^{2}=1$ and the data size is $T=500$. The prior and
proposal parameters are the same as the previous example. We ran Metropolis-within-Particle
Gibbs of \citet{Lindsten_and_Schon_2012} in Algorithm \ref{alg: Metropolis within particle Gibbs}.
Number of particles used in the cSMC moves is $M=100$. For each configuration,
$200$ Monte Carlo runs for 100000 iterations are performed and the
summary of the estimated IAC values from each run is reported in Figure
\ref{fig: IAC vs N for the state-space model}. One can see that increasing
the number of paths improves the results. However, the amount of improvement
(at least for this seemingly not very challenging model) vanishes
quickly after $N=10$; this is the reason we did not find necessary
to look at the performance of Algorithm \ref{alg: MHAAR for SSM with cSMC - Rao-Blackwellised backward sampling}
for this example. In addition, the results suggest that the scenario
$N=1$ seems useful in that the algorithm can beat Metropolis-within-Particle
Gibbs for the same order of computation. Note that the $N=1$ case
is also a recent algorithm, firstly proposed and analysed in \citet{Yildirim_et_al_2017},
with detailed comparisons with Metropolis-within-Particle Gibbs.

\begin{figure}[h]
\centerline{\includegraphics[scale=0.7]{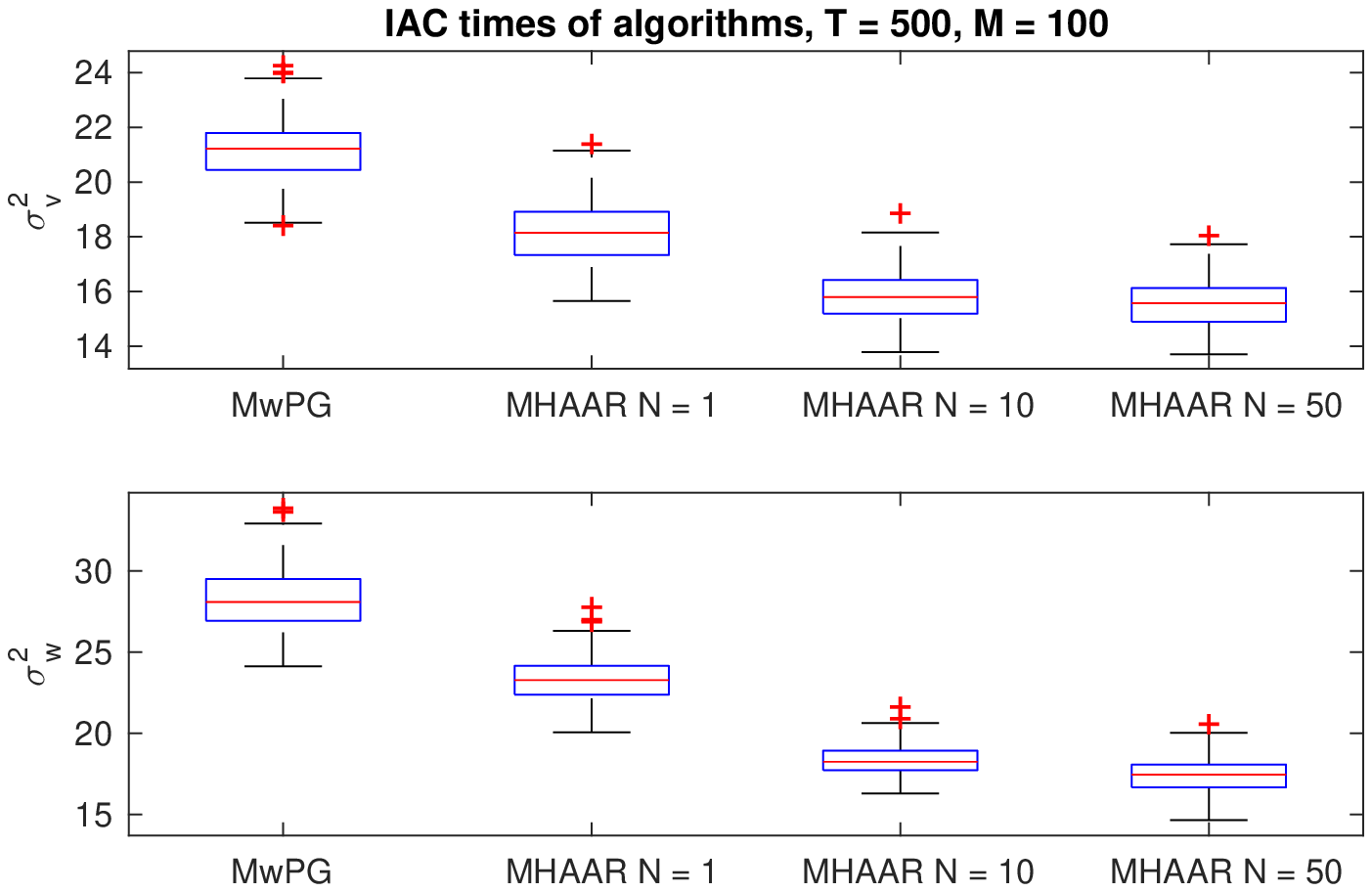}}
\protect\protect\caption{IAC for $\theta=(\sigma_{v}^{2},\sigma_{w}^{2})$ vs $N$ in Algorithm
\ref{alg: MHAAR for SSM with cSMC - multiple paths from backward sampling}
compared to Metropolis-within-Particle Gibbs (MwPG) in Algorithm \ref{alg: Metropolis within particle Gibbs}.}

\label{fig: IAC vs N for the state-space model} 
\end{figure}
 
\end{example}

\section{Discussion \label{sec: Discussion}}

In this paper, we exploit the ability to use more than one proposal
scheme within a MH update. We derive several useful MHAAR algorithms
that enable averaging multiple estimates of acceptance ratios, which
would not be valid by using a standard single proposal MH update.
The framework of MHAAR is rather general and provides a generic way
of improving performance of MH update based algorithm for a wide range
of problems. This is illustrated with doubly intractable models, general
latent variable models, trans-dimensional models, and general state-space
models. Although relevant in specific scenarios involving computations
on serial machines, MHAAR algorithms are particularly useful when
implemented on a parallel architecture since the computation required
to have an average acceptance ratio estimate can largely be parallelised.
In particular our experiments demonstrate significant reduction of
the burn in period required to reach equilibrium, an issue for which
very few generic approaches exist currently.

\subsection{Using SMC based estimators for the acceptance ratio \label{sec: Using SMC based estimators for the acceptance ratio}}

More broadly the framework of using asymmetric acceptance ratios allows
us to exploit even more general ratios of probabilities and plug them
into MCMCs. For example, a non-trivial interesting generalisation
of the algorithms presented earlier is possible by replacing AIS with
SMC. The generalisation is relevant when annealing is used, i.e.\ $T>0$
and it is available for both the scenario $\pi(x)=\pi(\theta)$ and
$\pi(x)=\pi(\theta,z)$. Notice that in Algorithms \ref{alg: MHAAR-AIS exchange algorithm}
to \ref{alg: MHAAR for pseudo-marginal ratio in latent variable models},
the acceptance ratios of the asymmetric MCMC algorithm contain the
factor 
\[
\frac{1}{N}\sum_{i=1}^{N}\prod_{t=0}^{T}\frac{f_{\theta,\theta',t+1}(u_{t}^{(i)})}{f_{\theta,\theta',t}(u_{t}^{(i)})}.
\]
This average actually serves as an AIS estimator of the ratio of the
normalising constants of the unnormalised densities $f_{\theta,\theta',0}$
and $f_{\theta,\theta',T+1}$ of the initial and the last densities
used in annealing. For doubly intractable models, this quantity is
$C_{\theta}/C_{\theta'}$, whereas in latent variable models, it is
$\pi(\theta')/\pi(\theta)$. Although SMC is a well known alternative
to AIS in estimating this ratio unbiasedly \citep{Del_Moral_et_al_2006},
it is not obvious whether or how we can substitute SMC for AIS in
proposal kernels $Q_{1}$ and $Q_{2}$ and still preserve the detailed
balance of the overall MCMC kernel with respect to $\pi$. It turns
out that this is possible by using a SMC in $Q_{1}$ and a series
of backward kernels followed by a cSMC in $Q_{2}$. For interested
readers, we present $Q_{1}$ and $Q_{2}$ with the corresponding acceptance
ratios, and the resulting algorithm in Appendix \ref{sec: Substituting SMC for AIS in the acceptance ratio in MHAAR}.

\subsection{Links to non-reversible algorithms \label{sec: Links to non-reversible algorithms}}

There has been recent interest in extending existing MCMC algorithms,
especially those based on MH, to algorithms having non-reversible
Markov chains preserving $\pi$ as their invariant distribution. The
motivation behind such algorithms is the desire to design proposals
based on the acceptance-rejection information of the previous iterations
so that the space $\mathsf{X}$ is explored more efficiently. For
example, it may be desirable to have a MH based Markov chain that
moves in a certain direction as long as the proposed values in that
direction are accepted. In case of rejection, the direction of the
proposal is altered and the Markov chain is made to choose a new direction
until the next rejection. 

These non-reversible MH algorithms can be interpreted as using acceptance
ratios involving two different proposal mechanism (e.g. for different
directions). Using two different proposals is inherent to our MHAAR
algorithms, and we briefly show how MHAAR algorithms can be turned
into non-reversible MCMC. Consider one pair of such proposal mechanisms
$Q_{1}(x,\mathrm{d}(y,u))$ and $Q_{2}(x,\mathrm{d}(y,u))$ as considered
throughout this paper. The acceptance ratios involved are denoted
$r_{1,u}(x,y)$ and $r_{2,u}(x,y)=1/r_{1,u}(y,x)$, depending no whether
$Q_{1}$ or $Q_{2}$ is on the numerator or denominator. The non-reversible
algorithm described in Algorithm \ref{alg: Non-reversible MHAAR}
targets the extended distribution $\pi(\mathrm{d}(x,a)):=\pi(\mathrm{d}x)\frac{1}{2}$,
where $a\in\{1,2\}$ and whose marginal is $\pi(\mathrm{d}x)$ as
desired, and generates realisations $\{\big(X_{n},A_{n}\big)\in\mathsf{X}\times\{1,2\},n\geq1\}$
where $A_{n}$ indicates which of $Q_{1}$ or $Q_{2}$ is to be used
at iteration $n+1$. 

\begin{algorithm}
\caption{Non-reversible MHAAR}
\label{alg: Non-reversible MHAAR}

\KwIn{Current sample and proposal state $X_{n}=x$, $A_{n}=a$} 

\KwOut{New sample and proposal state $X_{n+1}$, $A_{n+1}$} 

Sample $(y,u)\sim Q_{a}(x,\cdot)$ \\
Set $(X_{n+1},A_{n+1})=(x',a)$ with probability $\min\{1,r_{a,u}(x,y)\}$;
otherwise reject and set $(X_{n+1},A_{n+1})=(x,3-a)$.
\end{algorithm}

One iteration of the algorithm is a composition of two reversible
moves with respect to $\pi(\mathrm{d}(x,a))$: Given $(x,a)$, the
first move consists of proposing $y,a'$ (and $u$) from $Q_{a}(x,\mathrm{d}(y,u))\mathbb{I}_{3-a}(a')$,
accepting-rejecting with probability $\min\{1,r_{a,u}(x,y)\}$, which
is the corresponding asymmetric acceptance probability for $\pi(\mathrm{d}(x,a))$.
The second move simply switches the $a$-component: $a\rightarrow3-a$,
which is reversible. We do not investigate this further here.

\section{Acknowledgements}

CA and SY acknowledge support from EPSRC ``Intractable Likelihood:
New Challenges from Modern Applications (ILike)'' (EP/K014463/1)
and the Isaac Newton Institute for Mathematical Sciences, Cambridge,
for support and hospitality during the programme ``Scalable inference;
statistical, algorithmic, computational aspects'' where this manuscript
was finalised (EPSRC grant EP/K032208/1). AD acknowledges support
from EPSRC EP/K000276/1. NC is partially supported by a grant from
the French National Research Agency (ANR) as part of program ANR-11-LABEX-
0047. The authors would also like to thank Nick Whiteley for useful
discussions.

\bibliographystyle{plainnat}
\bibliography{myrefs_thesis}

\appendix

\section{A general framework for PMR and MHAAR algorithms \label{sec: A general framework for MPR and MHAAR algorithms}}

Assume $\pi$ is a probability distribution defined on the measurable
space $(\mathsf{X},\mathcal{X})$ and let $Q_{1}(\cdot,\cdot)$ and
$Q_{2}(\cdot,\cdot)$ be a pair of proposal kernels $Q_{1}(\cdot,\cdot),Q_{2}(\cdot,\cdot)\colon\mathsf{X}\times(\mathcal{U}\otimes\mathcal{X})\rightarrow[0,1]$,
where $\mathcal{U}$ is a sigma-algebra corresponding to an auxiliary
random variable $u$ defined on a measurable space $(\mathsf{U},\mathcal{U})$.
This variable may or may not be present, and is for example ignored
in the introductory Section \ref{subsec: Contribution}. We first
follow \citet{Tierney_1998} (in particular his treatment of \citet{Green_1995}'s
framework) and introduce the measure 
\begin{align*}
\nu_{i}\big({\rm d}(x,y,u)\big): & =\pi({\rm d}x)Q_{i}(x,{\rm d}(y,u))+\pi({\rm d}y)Q_{3-i}(y,{\rm d}(x,u))
\end{align*}
 and for $i\in\{1,2\}$ the densities $\eta_{i}(x,y,u):={\rm d}(\pi\otimes Q_{i})/{\rm d}\nu_{i}$
for $(x,y,u)\in\mathsf{X}^{2}\times\mathsf{U}$. Now define the measurable
set
\begin{equation}
\mathsf{S}:=\big\{(x,y,u)\in\mathsf{X}^{2}\times\mathsf{U}\colon\eta_{1}(x,y,u)>0\text{ and }\eta_{2}(y,x,u)>0\big\}\label{eq:def-ring-S}
\end{equation}
and let, for $i\in\{1,2\}$
\[
r_{i,u}(x,y):=\begin{cases}
\eta_{3-i}(y,x,u)/\eta_{i}(x,y,u) & \text{ for }(x,y,u)\in\mathsf{S},\\
0 & \text{otherwise}.
\end{cases}
\]
For ease of exposition, throughout the rest of the paper we may use
the notation
\[
\frac{\pi({\rm d}y)Q_{3-i}(y,{\rm d}(x,u))}{\pi({\rm d}x)Q_{i}(x,{\rm d}(y,u))}=:r_{i,u}(x,y),
\]
which should not lead to any confusion. Further, for $x\in\mathsf{X}$,
we define the rejection probabilities

\[
\rho_{i}(x):=1-\int_{\mathsf{X}\times\mathsf{U}}Q_{i}(x,{\rm d}(y,u))\min\{1,r_{i,u}(x,y)\},\quad i=1,2.
\]
In the theorem below we use the properties that for $i\in\{1,2\}$
and $(x,y)\in\mathsf{S}$ and $u\in\mathsf{U}$, then $r_{i}(x,y)r_{3-i}(y,x)=1$
and $\nu_{i}\big({\rm d}(x,y,u)\big)=\nu_{3-i}\big({\rm d}(y,x,u)\big)$.
 The following theorem serves as the basis for proving the reversibility
of all of the MHAAR algorithms developed in this paper.
\begin{thm}
\label{thm:asymmetricMH-1} Consider the Markov transition kernel
$\breve{P}\colon\mathsf{X}\times\mathcal{X}\rightarrow[0,1]$
\begin{equation}
\breve{P}(x,{\rm d}y):=\sum_{i=1}^{2}\frac{1}{2}\left[\int_{\mathsf{U}}Q_{i}(x,{\rm d}(y,u)\min\left\{ 1,r_{i,u}(x,y)\right\} +\delta_{x}({\rm d}y)\rho_{i}(x)\right],\quad x\in\mathsf{X}\label{eq: asymmetric MCMC acceptance kernel-1}
\end{equation}
then $\breve{P}$ satisfies the detailed balance for $\pi$.
\end{thm}
\begin{proof}
For any bounded measurable function $\phi$ on $\mathsf{X}^{2}$:
\begin{align*}
\int_{\mathsf{X}^{2}\times\mathsf{U}}\min\left\{ 1,r_{1,u}(x,y)\right\}  & \phi(x,y)\pi({\rm d}x)Q_{1}(x;{\rm d}(y,u))\\
= & \int_{\mathsf{X}^{2}\times\mathsf{U}}\phi(x,y)\min\left\{ 1,r_{1,u}(x,y)\right\} \eta_{1}(x,y,u)\nu_{1}\big({\rm d}(x,y,u)\big)\\
= & \int_{\mathsf{S}}\phi(x,y)\min\left\{ 1,r_{1,u}(x,y)\right\} r_{2,u}(y,x)\eta_{2}(y,x,u)\nu_{1}\big({\rm d}(x,y,u)\big)\\
= & \int_{\mathsf{X}^{2}\times\mathsf{U}}\phi(x,y)\min\left\{ 1,r_{1,u}(x,y)\right\} r_{2,u}(y,x)\eta_{2}(y,x,u)\nu_{2}\big({\rm d}(y,x,u)\big)\\
= & \int_{\mathsf{X}^{2}\times\mathsf{U}}\phi(x,y)\min\left\{ 1,r_{1,u}(x,y)\right\} r_{2,u}(y,x)\pi({\rm d}y)Q_{2}(y;{\rm d}(x,u))\\
= & \int_{\mathsf{X}^{2}\times\mathsf{U}}\phi(x,y)\min\left\{ r_{2,u}(y,x),1\right\} \pi({\rm d}y)Q_{2}(y;{\rm d}(x,u)).
\end{align*}
As a result for $\phi$ as above,
\begin{multline*}
\sum_{i=1}^{2}\frac{1}{2}\int_{\mathsf{X}^{2}\times\mathsf{U}}\phi(x,y)\min\left\{ 1,r_{i,u}(x,y)\right\} \pi({\rm d}x)Q_{i}(x;{\rm d}(y,u))\\
=\sum_{i=1}^{2}\frac{1}{2}\int_{\mathsf{X}^{2}\times\mathsf{U}}\phi(x,y)\min\left\{ 1,r_{i,u}(y,x)\right\} \pi({\rm d}y)Q_{i}(y;{\rm d}(x,u)),
\end{multline*}
and detailed balance hence follows.
\end{proof}
The following theorem validates the use of all the PMR algorithms
in this paper, specifically the algorithm corresponding to the kernel
presented in \eqref{eq:defPring} and the algorithms described in
Propositions \ref{prop: AIS MCMC for doubly intractable models},
\ref{prop:MHwithAISinside} and \ref{prop:extensionMHwithAISinside}.
\begin{thm}
\label{thm: pseudo-marginal ratio algorithms} Let $\varphi:\mathsf{U}\rightarrow\mathsf{U}$
be a measurable involution, that is such that $\varphi\circ\varphi(u)=u$
for all $u\in\mathsf{U}$ and with the set-up of Theorem \ref{thm:asymmetricMH-1}
for a given kernel $Q_{1}(\cdot,\cdot)$, let $Q_{2}(\cdot,\cdot)$
be defined such that for any measurable $\psi\colon\mathsf{X}^{2}\times\mathsf{U}\rightarrow[-1,1]$
\begin{align*}
\int\psi(x,y,u)\pi({\rm d}x)Q_{2}\big(x,\mathrm{d}(y,u)\big) & =\int\psi(x,y,\varphi(u))\pi({\rm d}x)Q_{1}\big(x,\mathrm{d}(y,u)\big).
\end{align*}
Then the Markov transition kernel
\[
\mathring{P}(x,\mathrm{d}y)=\int_{\mathsf{U}}Q_{1}(x,\mathrm{d}(y,u))\min\left\{ 1,r_{1,u}(x,y)\right\} +\rho_{1}(x)\delta_{x}(\mathrm{d}y)
\]
satisfies detailed balance with respect to $\pi$\textup{.}
\end{thm}
\begin{proof}
\[
\nu_{1}\big({\rm d}(x,y,u)\big)=\nu_{2}\big({\rm d}(y,x,u)\big)
\]
First we show that
\begin{equation}
\int\psi(x,y,u)\nu_{1}(\mathrm{d}(x,y,u))=\int\psi(x,y,\varphi(u))\nu_{1}(\mathrm{d}(y,x,u)).\label{eq:nu1_is_nu1_varphi}
\end{equation}
This is because for $\psi$ bounded and measurable 
\begin{align*}
\int\psi(x,y,u)\nu_{1}\big({\rm d}(x,y,u)\big) & =\int\psi(x,y,u)(\pi\otimes Q_{1})\big({\rm d}(x,y,u)\big)+\int\psi(x,y,\varphi(u))(\pi\otimes Q_{1})\big({\rm d}(y,x,u)\big)\\
 & =\int\psi(x,y,\varphi(u))(\pi\otimes Q_{2})\big({\rm d}(x,y,u)\big)+\int\psi(x,y,\varphi(u))(\pi\otimes Q_{1})\big({\rm d}(y,x,u)\big)\\
 & =\int\psi(x,y,\varphi(u))\nu_{1}\big({\rm d}(y,x,u)\big),
\end{align*}
where we have used our assumption on $\pi\otimes Q_{2}$ on the first
and second line, together with the fact that $\varphi$ is an involution,
and the definition of $\nu_{1}$ on the last line. As a result one
can establish that
\[
\eta_{2}(x,y,u)=\eta_{1}(x,y,\varphi(u)).
\]
Indeed, for $\psi$ bounded and measurable,
\begin{align*}
\int\psi(x,y,u)\eta_{2}(x,y,u)\nu_{2}\big({\rm d}(x,y,u)\big)= & \int\psi(x,y,u)\pi\otimes Q_{2}\big({\rm d}(x,y,u)\big)\\
= & \int\psi(x,y,\varphi(u))\pi\otimes Q_{1}\big({\rm d}(x,y,u)\big)\\
= & \int\psi(x,y,\varphi(u))\eta_{1}(x,y,u)\nu_{1}\big({\rm d}(x,y,u)\big)\\
= & \int\psi(x,y,u)\eta_{1}(x,y,\varphi(u))\nu_{1}\big({\rm d}(y,x,u)\big)\\
= & \int\psi(x,y,u)\eta_{1}(x,y,\varphi(u))\nu_{2}\left({\rm d}(x,y,u)\right),
\end{align*}
where we have used \eqref{eq:nu1_is_nu1_varphi} on the fourth line.
Now for $\phi\colon\mathsf{X}^{2}\rightarrow[-1,1]$ measurable
\begin{align*}
\int_{\mathsf{X}\times\mathsf{U}\times\mathsf{X}}\phi(x,y)\min & \left\{ 1,r_{1,u}(x,y)\right\} \pi({\rm d}x)Q_{1}(x,\mathrm{d}(y,u))\\
 & =\int_{\mathsf{S}}\phi(x,y)\min\left\{ 1,r_{1,u}(x,y)\right\} \frac{\eta_{1}(x,y,u)}{\eta_{1}(y,x,\varphi(u))}\eta_{1}(y,x,\varphi(u))\nu_{1}(\mathrm{d}(x,y,u))\\
 & =\int_{\mathsf{S}}\phi(x,y)\min\left\{ r_{1,\varphi(u)}(y,x),1\right\} \eta_{1}(y,x,\varphi(u))\nu_{1}(\mathrm{d}(x,y,u))\\
 & =\int_{\mathsf{S}}\phi(x,y)\min\left\{ r_{1,u}(y,x),1\right\} \eta_{1}(y,x,u)\nu_{1}(\mathrm{d}(y,x,u)),
\end{align*}
and reversibility follows.
\end{proof}

\subsection{Specialisation to two specific scenarios \label{subsec: Construction of the algorithms in the paper}}

Although the general framework described in Theorem \ref{thm:asymmetricMH-1}
is quite broad, our algorithms exploit it in specific ways. In this
subsection, we aim to provide some insight into the ways we exploit
these ideas in this paper. Recall that we either have $x=\theta$
in the single variable case or $x=(\theta,z)$ in the scenario where
the model involves latent variables. Throughout the paper, we design
algorithms where both $Q_{1}(\cdot,\cdot)$ and $Q_{2}(\cdot,\cdot)$
use the same proposal distribution for $\theta'$, that is $q(\theta,\cdot)$
and differ in the way they sample the auxiliary variables (and $z'$
in the latent variables scenario) such that $Q_{1}(\cdot,\cdot)$
and $Q_{2}(\cdot,\cdot)$ complement each other to produce acceptance
ratio estimators whose statistical properties increase with (parallelisable)
computations.

\subsubsection{Single variable scenario \label{subsec: Single variable}}

Here we have $x=\theta$. Let $\{Q_{\theta,\theta'}^{(1)}(\cdot)\colon\theta,\theta'\in\Theta\}$
and $\{Q_{\theta,\theta'}^{(2)}(\cdot)\colon\theta,\theta'\in\Theta\}$
be two families of probability distributions defined on $(\mathsf{U,\mathcal{U}})$
and $\omega_{\theta,\theta'}:\mathcal{\mathsf{U}}\rightarrow[0,\infty)$
satisfying the condition, for $\theta,\theta'\in\Theta$ and $u\in\mathsf{U}$,
\begin{equation}
Q_{\theta',\theta}^{(2)}(\mathrm{d}u)=Q_{\theta,\theta'}^{(1)}(\mathrm{d}u)\omega_{\theta,\theta'}(u),\label{eq: relation between Q1 and Q2}
\end{equation}
so that the expected value of $\omega_{\theta,\theta'}(\cdot)$ with
respect to $Q_{\theta,\theta'}^{(1)}(\cdot)$ (as well as the expected
value of $\omega_{\theta,\theta'}^{-1}(\cdot)$ with respect to $Q_{\theta',\theta}^{(2)}(\cdot)$
if $\omega_{\theta,\theta'}(\cdot)>0$) is $1$. Then the Radon-Nikodym
derivative evaluated for $(\theta,u,\theta')\in\mathsf{S}$ as defined
above,
\begin{equation}
r_{u}(\theta,\theta')=\frac{\pi(\mathrm{d}\theta')q(\theta',\mathrm{d}\theta)Q_{\theta',\theta}^{(2)}(\mathrm{d}u)}{\pi(\mathrm{d}\theta)q(\theta,\mathrm{d}\theta')Q_{\theta,\theta'}^{(1)}(\mathrm{d}u)}=r(\theta,\theta')\omega_{\theta,\theta'}(u).\label{eq: RN derivative for single variable}
\end{equation}
Note that this ratio is an unbiased estimator of the acceptance ratio
of the marginal distribution, $r(\theta,\theta')$; therefore useful
algorithm can be constructed if (i) $r(\theta,\theta')\omega_{\theta,\theta'}(u)$
can be evaluated, and (ii) the variance of of $\omega_{\theta,\theta'}$
can be controlled. It follows exactly from \eqref{eq: RN derivative for single variable}
and Theorem \ref{thm:asymmetricMH-1} that we can construct a reversible
Markov kernel using acceptance ratios involving $r_{u}(\theta,\theta')$
as in Theorem \ref{thm:asymmetricMH-1} with
\[
Q_{1}(\theta,\mathrm{d}(\theta',u))=q(\theta,\mathrm{d}\theta')Q_{\theta,\theta'}^{(1)}(\mathrm{d}u),\quad Q_{2}(\theta,\mathrm{d}(\theta',u))=q(\theta,\mathrm{d}\theta')Q_{\theta,\theta'}^{(2)}(\mathrm{d}u).
\]
If, in addition, for any measurable and bounded function $\phi$ we
have $\int\phi(u)Q_{\theta,\theta'}^{(2)}(\mathrm{d}u)=\int\phi\circ\varphi(u)Q_{\theta,\theta'}^{(1)}(\mathrm{d}u)$
for some involution $\varphi$, we are precisely in the frame of the
pseudo-marginal ratio algorithms discussed in \citet{Nicholls_et_al_2012},
whose transition kernel is given in Theorem \ref{thm: pseudo-marginal ratio algorithms}.

\subsubsection{Latent model scenario\label{subsec: Latent components} }

Here we have $x=(\theta,z)$ with $\pi(\mathrm{d}x)=\pi(\mathrm{d}\theta)\pi_{\theta}(\mathrm{d}z)$.
Let $\{Q_{\theta,\theta',z}^{(1)}(\cdot)\colon\theta,\theta'\in\Theta\}$
and $\{Q_{\theta,\theta',z}^{(2)}(\cdot)\colon\theta,\theta'\in\Theta,z\in\mathsf{Z}\}$
be two families of probability distributions defined on $(\mathsf{\mathsf{Z}\times U,\mathcal{Z}\otimes\mathcal{U}})$
and $\omega_{\theta,\theta'}:\mathcal{\mathsf{Z}}\times\mathsf{U}\rightarrow[0,\infty)$
satisfying the condition 
\[
\pi_{\theta'}(\mathrm{d}z')Q_{\theta',\theta,z'}^{(2)}(\mathrm{d}(z,u))=\pi_{\theta}(\mathrm{d}z)Q_{\theta,\theta',z}^{(1)}(\mathrm{d}(z',u))\omega_{\theta,\theta'}(z,u),
\]
so that the expected value of $\omega_{\theta,\theta'}(z,u)$ with
respect to $\pi_{\theta}(\mathrm{d}z)Q_{\theta,\theta',z}^{(1)}(\mathrm{d}(z',u))$
is $1$. Just as in the single variable case, consider the Radon-Nikodym
derivative again: 
\begin{equation}
r_{u}(x,x')=\frac{\pi(\mathrm{d}x')q(\theta',\mathrm{d}\theta)Q_{\theta',\theta,z'}^{(2)}(\mathrm{d}(z,u))}{\pi(\mathrm{d}x)q(\theta,\mathrm{d}\theta')Q_{\theta,\theta',z}^{(1)}(\mathrm{d}(z',u))}=r(\theta,\theta')\omega_{\theta,\theta'}(z,u).
\end{equation}
Note that this ratio is an unbiased estimator of the acceptance ratio
of the marginal distribution, $r(\theta,\theta')$; therefore useful
a algorithm can be constructed if (i) $r(\theta,\theta')\omega_{\theta,\theta'}(z,u)$
can be evaluated and (ii) the variance of of $\omega_{\theta,\theta'}$
can be controlled. We can construct a reversible Markov kernel using
$Q_{1}(\cdot,\cdot)$ and $Q_{2}(\cdot,\cdot)$ as: 
\[
Q_{1}(x,\mathrm{d}(x',u))=q(\theta,\mathrm{d}\theta')Q_{\theta,\theta',z}^{(1)}(\mathrm{d}(z',u)),\quad Q_{2}(x,\mathrm{d}(x',u))=q(\theta,\mathrm{d}\theta')Q_{\theta,\theta',z}^{(2)}(\mathrm{d}(z',u)).
\]
Similarly, if, in addition, for any bounded measurable function $\phi$
we have $\int\phi(z,u)\pi_{\theta}({\rm d}z)Q_{\theta,\theta',z}^{(2)}(\mathrm{d}(z',u))=\int\phi(z,\varphi(u))\pi_{\theta}({\rm d}z)Q_{\theta,\theta',z}^{(1)}(\mathrm{d}(z',u))$
for some some involution $\varphi\colon\mathsf{U}\rightarrow\mathsf{U}$,
we can use the transition kernel given in Theorem \ref{thm: pseudo-marginal ratio algorithms}
and we end up precisely in the framework of the pseudo-marginal ratio
algorithms for latent variable models discussed in Section \ref{sec: Pseudo-marginal ratio algorithms for latent variable models}.

\subsection{Generalisation and theoretical sub-optimality\label{subsec: Generalisation and suboptimality}}

One can be more general than having a single pair of proposal distributions
and sampling them with equal probabilities. In the following, we will
consider multiple pairs and sampling among proposal distributions
with state-dependent probabilities. Then we will investigate the statistical
properties of this scheme by comparing it to an ideal but non-implementable
algorithm in terms of Peskun order. For some $m\in\mathbb{N}$ let
$\{Q_{ij}(\cdot,\cdot),i,j\in\{1,\ldots,m\}\}$ be a family of proposal
kernels each from $(\mathsf{X},\mathcal{X})$ to $(\mathsf{X}\times\mathsf{U},\mathcal{X}\times\mathcal{U})$
and $\{\beta_{ij}:\mathsf{X}\rightarrow[0,1],i,j\in\{1,\ldots,m\}\}$
such that for any $x\in\mathsf{X}$, $\sum_{i,j=1}^{m}\beta_{ij}(x)=1$.
Define the Markov transition kernel 
\begin{equation}
\breve{P}(x,{\rm d}y):=\sum_{i=1}^{m}\sum_{j=1}^{m}\beta_{ij}(x)\left[\int_{\mathsf{U}}Q_{ij}(x,{\rm d}(y,u))\min\left\{ 1,r_{ij,u}(x,y)\right\} +\delta_{x}({\rm d}y)\rho_{ij}(x)\right],\quad x\in\mathsf{X}\label{eq: asymmetric MCMC acceptance kernel}
\end{equation}
where the acceptance ratios are 
\[
r_{ij,u}(x,y):=\frac{\pi({\rm d}y)Q_{ji}(y,{\rm d}(x,u))}{\pi({\rm d}x)Q_{ij}(x,{\rm d}(y,u))}\frac{\beta_{ji}(y)}{\beta_{ij}(x)},\quad i,j=1,\ldots,m,
\]
on some set $\mathring{\mathsf{S}}_{ij}\subset\mathsf{X}\times\mathsf{U}\times\mathsf{X}$
where the measures $\pi({\rm d}y)Q_{ji}(y,{\rm d}(x,u))$ and $\pi({\rm d}x)Q_{ij}(x,{\rm d}(y,u))$
are equivalent (see the beginning of Appendix \ref{sec: A general framework for MPR and MHAAR algorithms})
and $0<\beta_{ij}(x)\beta_{ji}(y)<\infty$ and set to zero otherwise,
while the rejection probabilities at $x\in\mathsf{X}$ corresponding
to all the updates are given by 
\[
\rho_{ij}(x):=1-\int_{\mathsf{U}\times\mathsf{X}}Q_{ij}(x,{\rm d}(y,u))\min\{1,r_{ij,u}(x,y)\},\quad i,j=1,\ldots,m.
\]
Reversibility of $\breve{P}$ can be proven very similarly to Theorem
\ref{thm:asymmetricMH-1}, therefore it is only stated as a corollary
below.

\begin{cor}
\label{cor: asymmetricMH generalisation}The MHAAR algorithm with
transition kernel $\breve{P}$ in \eqref{eq: asymmetric MCMC acceptance kernel}
satisfies detailed balance for $\pi$. 
\end{cor}
The standard MH algorithm is recovered, for example, in the situation
where $\beta_{11}(x)=1$. The single pair version is recovered with
$\beta_{12}(x)=\beta_{21}(x)=1/2$ and $Q_{12}(\cdot,\cdot)=Q_{1}(\cdot,\cdot)$,
$Q_{21}(\cdot,\cdot)=Q_{2}(\cdot,\cdot)$. Algorithm \ref{alg: Reversible multiple jump MCMC}
corresponds to the special case where $\beta_{12}(x)+\beta_{21}(x)=1$. 

The following interpretation of $\breve{P}$ points to a theoretical
sub-optimality of asymmetric MCMC a careful reader may point to. Indeed,
from \eqref{eq: asymmetric MCMC acceptance kernel}, the auxiliary
variable $u\in\mathsf{U}$ and the proposed value $y\in\mathsf{X}$
are sampled from 
\[
\breve{Q}(x,\cdot):=\sum_{i=1}^{m}\sum_{j=1}^{m}\beta_{ij}(x)Q_{ij}(x,\cdot),
\]
and the proposed value is accepted with probability 
\[
\breve{\alpha}_{u}(x,y):=\sum_{i=1}^{m}\sum_{j=1}^{m}\frac{\beta_{ij}(x)Q_{ij}(x,{\rm d}(y,u))}{\breve{Q}(x,{\rm d}(y,u))}\min\{1,r_{ij,u}(x,y)\}.
\]
Application of Jensen's inequality shows that for $x,y\in\mathsf{X}$,
$u\in\mathsf{U}$, we have 
\begin{align}
\breve{\alpha}_{u}(x,y) & \leq\min\left\{ 1,\sum_{i=1}^{m}\sum_{j=1}^{m}\frac{\beta_{ij}(x)Q_{ij}(x,{\rm d}(y,u))}{\breve{Q}(x,{\rm d}(y,u))}r_{ij,u}(x,y)\right\} \nonumber \\
 & =\min\left\{ 1,\sum_{i=1}^{m}\sum_{j=1}^{m}\frac{\beta_{ij}(x)Q_{ij}(x,{\rm d}(y,u))}{\breve{Q}(x,{\rm d}(y,u))}\frac{\beta_{ji}(y)\pi({\rm d}y)Q_{ji}(y,{\rm d}(x,u))}{\beta_{ij}(x)\pi({\rm d}x)Q_{ij}(x,{\rm d}(y,u))}\right\} \nonumber \\
 & =\min\left\{ 1,\frac{\pi({\rm d}y)\breve{Q}(y,{\rm d}(x,u))}{\pi({\rm d}x)\breve{Q}(x,{\rm d}(y,u))}\right\} =:\alpha_{u}(x,y),\label{eq: acceptance prob of MH with asymmetric proposal}
\end{align}
which is the acceptance probability of a pseudo-marginal ratio MH
algorithm $\mathring{P}$ with $Q_{1}(\cdot,\cdot)=\breve{Q}(\cdot,\cdot)$
and $\varphi(u)=u$ (see Theorem \ref{thm: pseudo-marginal ratio algorithms}).
From Peskun's result \citet{Tierney_1998} we deduce that for this
common proposal distribution $\breve{Q}$, the update $\breve{P}$
has worse performance properties in terms of both asymptotic variance
and right spectral gap than $\mathring{P}$. It is therefore natural
to question the interest of updates such as $\breve{P}$. An argument
already noted by \citet{tjelmeland-eidsvik-2004,andrieu2008tutorial},
is that computing the acceptance ratio in \eqref{eq: acceptance prob of MH with asymmetric proposal}
is generally substantially more computationally expensive than computing
$\breve{r}_{ij}(x,y)$, which may offset any theoretical advantage
in practice. It may also be that defining a desirable acceptance ratio
is theoretically impossible using the standard approach, or that practical
evaluation of the acceptance ratio is impossible. This is the case
for numerous examples, including Example \ref{ex:doublyintractaveraging},
for which
\begin{align*}
\alpha_{u}(\theta,\theta') & =r(\theta,\theta')\frac{\big[\prod_{i=1}^{N}g_{\theta}(u^{(i)})/C_{\theta}\,\mathring{r}_{u^{(k)}}(\theta',\theta)/\mathring{r}_{\mathfrak{u}}^{N}(\theta',\theta)+N^{-1}g_{\theta}(u^{(k)})/C_{\theta}\prod_{i\neq k}g_{\theta'}(u^{(i)})/C_{\theta'}\big]}{\big[\prod_{i=1}^{N}g_{\theta'}(u^{(i)})/C_{\theta'}\,\mathring{r}_{u^{(k)}}(\theta,\theta')/\mathring{r}_{\mathfrak{u}}^{N}(\theta,\theta')+N^{-1}g_{\theta'}(u^{(k)})/C_{\theta'}\prod_{i\neq k}g_{\theta}(u^{(i)})/C_{\theta}\big]},
\end{align*}
for $N\geq1$ and where we note that the unknown normalising constants
do not cancel.

\section{Justification of AIS and an extension\label{sec: A short justification of AIS and an extension}}

We provide here a short justification of the AIS of \citet{Crooks1998,Neal_2001},
as presented in \citet{Karagiannis_and_Andrieu_2013}, and an extension
of it that is useful in this paper. Here $\tau$ represents the number
of intermediate distributions introduced, while $\mu_{0}$ and $\mu_{\tau+1}$
are the distributions of which we want the normalising constants as
we assume that we only know them up to normalising constants i.e.\ we
know the unnormalised distributions $\nu_{0}=\mu_{0}Z_{0}$ and $\nu_{\tau+1}=\mu_{\tau+1}Z_{\tau+1}$. 
\begin{thm}
\label{thm:AIS}Let $\big\{\mu_{t},t=0,\ldots,\tau+1\big\}$ for some
$\tau\in\mathbb{N}$ be a family of probability distributions on some
measurable space $\big(\mathsf{E},\mathcal{E}\big)$ such that for
$t=0,\ldots,\tau$ $\mu_{t}\gg\mu_{t+1}$. Let $\big\{\Pi_{t},t=1,\ldots,\tau\big\}$
be a family of Markov transition kernels $\Pi_{t}:\mathsf{E}\times\mathcal{E}\rightarrow[0,1]$
such that for any $t=1,\ldots,\tau$, $\Pi_{t}$ is $\mu_{t}-$reversible.
Let us define the following probability distributions on $\big(\mathsf{E}^{\tau+1},\mathcal{E}^{\tau+1}\big)$,
$\overleftarrow{\varPi}:=\mu_{\tau+1}\times\Pi_{\tau}\times\cdots\times\Pi_{1}$
and $\overrightarrow{\varPi}:=\mu_{0}\times\Pi_{1}\times\cdots\times\Pi_{\tau}$
for $\tau\geq1$ and $\overleftarrow{\varPi}:=\mu_{\tau+1}$ and $\overrightarrow{\varPi}:=\mu_{0}$
for $\tau=0$. Then for any $x_{0:\tau}\in\mathsf{E}^{\tau+1}$ 
\begin{align*}
\overleftarrow{\varPi}\big({\rm d}(x_{\tau},\ldots,x_{0})\big) & =\prod_{t=0}^{\tau}\frac{\mu_{t+1}\big({\rm d}x_{t}\big)}{\mu_{t}\big({\rm d}x_{t}\big)}\overrightarrow{\varPi}\big({\rm d}(x_{0},\ldots,x_{\tau})\big).
\end{align*}
\end{thm}
\begin{proof}
The case $\tau=0$ is direct. Assume $\tau\geq1$, we show by induction
that for any $j=1,\ldots,\tau$ 
\begin{multline*}
\overleftarrow{\varPi}\big({\rm d}(x_{\tau},\ldots,x_{0})\big)\\
=\left[\prod_{t=1}^{j}\frac{\mu_{\tau-t+2}\big({\rm d}x_{\tau+1-t}\big)}{\mu_{\tau-t+1}\big({\rm d}x_{\tau+1-t}\big)}\Pi_{\tau+1-t}\big(x_{\tau-t},{\rm d}x_{\tau-t+1}\big)\right]\mu_{\tau-j+1}\big({\rm d}x_{\tau-j}\big)\prod_{t=j+1}^{\tau}\Pi_{\tau-t+1}\big(x_{\tau-t+1},{\rm d}x_{\tau-t}\big),
\end{multline*}
with the convention $\prod_{t=\tau+1}^{\tau}=I$. First we check the
result for $j=1$ 
\begin{align*}
\overleftarrow{\varPi}\big({\rm d}(x_{\tau},\ldots,x_{0})\big) & =\mu_{\tau+1}\big({\rm d}x_{\tau}\big)\prod_{t=1}^{\tau}\Pi_{\tau-t+1}\big(x_{\tau-t+1},{\rm d}x_{\tau-t}\big)\\
 & =\frac{\mu_{\tau+1}\big({\rm d}x_{\tau}\big)}{\mu_{\tau}\big({\rm d}x_{\tau}\big)}\mu_{\tau}\big({\rm d}x_{\tau}\big)\Pi_{\tau}\big(x_{\tau},{\rm d}x_{\tau-1}\big)\prod_{t=2}^{\tau}\Pi_{\tau-t+1}\big(x_{\tau-t+1},{\rm d}x_{\tau-t}\big)\\
 & =\frac{\mu_{\tau+1}\big({\rm d}x_{\tau}\big)}{\mu_{\tau}\big({\rm d}x_{\tau}\big)}\mu_{\tau}\big({\rm d}x_{\tau-1}\big)\Pi_{\tau}\big(x_{\tau-1},{\rm d}x_{\tau}\big)\prod_{t=2}^{\tau}\Pi_{\tau-t+1}\big(x_{\tau-t+1},{\rm d}x_{\tau-t}\big),
\end{align*}
where we have used $\mu_{\tau}\gg\mu_{\tau+1}$ and the fact that
$\Pi_{\tau}$ is $\mu_{\tau}-$reversible. Now assume the result true
for some $1\leq j\leq\tau-1$ and $\tau\geq2$, then using similar
arguments as above, 
\begin{align*}
\mu_{\tau-j+1}\big({\rm d}x_{\tau-j}\big)\Pi_{\tau-j}\big(x_{\tau-j}, & {\rm d}x_{\tau-j-1}\big)\\
 & =\frac{\mu_{\tau-j+1}\big({\rm d}x_{\tau-j}\big)}{\mu_{\tau-j}\big({\rm d}x_{\tau-j}\big)}\mu_{\tau-j}\big({\rm d}x_{\tau-j}\big)\Pi_{\tau-j}\big(x_{\tau-j},{\rm d}x_{\tau-j-1}\big)\\
 & =\frac{\mu_{\tau-j+1}\big({\rm d}x_{\tau-j}\big)}{\mu_{\tau-j}\big({\rm d}x_{\tau-j}\big)}\mu_{\tau-j}\big({\rm d}x_{\tau-j-1}\big)\Pi_{\tau-j}\big(x_{\tau-j-1},{\rm d}x_{\tau-j}\big),
\end{align*}
from which the intermediate claim follows for $j+1$. Now for $j=\tau$
we obtain the claimed result after a change of variables $t\leftarrow\tau+1-t$
in the product. 
\end{proof}
\begin{cor}
Assume that we have access to unnormalised versions of the probability
distributions, say $\nu_{t}=\mu_{t}Z_{t}$. Then 
\[
\overleftarrow{\varPi}\big({\rm d}(x_{\tau},\ldots,x_{0})\big)=\prod_{t=0}^{\tau}\frac{Z_{t}}{Z_{t+1}}\frac{\nu_{t+1}\big({\rm d}x_{t}\big)}{\nu_{t}\big({\rm d}x_{t}\big)}\overrightarrow{\varPi}\big({\rm d}(x_{0},\ldots,x_{\tau})\big),
\]
and therefore 
\begin{align*}
\prod_{t=0}^{\tau}\frac{\nu_{t+1}\big({\rm d}x_{t}\big)}{\nu_{t}\big({\rm d}x_{t}\big)}\overrightarrow{\varPi}\big({\rm d}(x_{0},\ldots,x_{\tau})\big) & =\overleftarrow{\varPi}\big({\rm d}(x_{\tau},\ldots,x_{0})\big)\prod_{t=0}^{\tau}\frac{Z_{t+1}}{Z_{t}}\\
 & =\overleftarrow{\varPi}\big({\rm d}(x_{\tau},\ldots,x_{0})\big)\frac{Z_{\tau+1}}{Z_{0}},
\end{align*}
which suggests and justifies the AIS estimator. 
\end{cor}
The following extension of the result above turns out to be of practical
interest.
\begin{thm}
\label{thm:extensionAIS} Let $\tau\geq2$, $\big\{\mu_{t},t=1,\ldots,\tau\big\}$
and $\big\{\Pi_{t},t=2,\ldots,\tau-1\big\}$ be as in Theorem \ref{thm:AIS}
above but assume now that $\mu_{0}$ and $\mu_{\tau+1}$ are defined
on a potentially different measurable space $(\mathsf{F},\mathcal{F})$.
Further let $\overrightarrow{\Pi}_{1},\overleftarrow{\Pi}_{\tau}:\mathsf{F}\times\mathcal{E}\rightarrow[0,1]$
and $\overleftarrow{\Pi}_{1},\overrightarrow{\Pi}_{\tau}:\mathsf{E}\times\mathcal{F}\rightarrow[0,1]$
be Markov kernels satisfying the following properties 
\[
\mu_{0}\big({\rm d}x_{0}\big)\overrightarrow{\Pi}_{1}\big(x_{0},{\rm d}x_{1}\big)=\mu_{1}\big({\rm d}x_{1}\big)\overleftarrow{\Pi}_{1}\big(x_{1},{\rm d}x_{0}\big),
\]
and
\[
\mu_{\tau}\big({\rm d}x_{\tau-1}\big)\overrightarrow{\Pi}_{\tau}\big(x_{\tau-1},{\rm d}x_{\tau}\big)=\mu_{\tau+1}\big({\rm d}x_{\tau}\big)\overleftarrow{\Pi}_{\tau}\big(x_{\tau},{\rm d}x_{\tau-1}\big).
\]
Define
\[
\overrightarrow{\varPi}:=\mu_{0}\times\overrightarrow{\Pi}_{1}\times\Pi_{2}\cdots\times\overrightarrow{\Pi}_{\tau},
\]
and 
\[
\overleftarrow{\varPi}:=\mu_{\tau+1}\times\overleftarrow{\Pi}_{\tau}\times\Pi_{\tau-1}\cdots\times\overleftarrow{\Pi}_{1}.
\]
Then
\[
\overleftarrow{\varPi}\big({\rm d}(x_{\tau},\ldots,x_{0})\big)=\prod_{t=1}^{\tau-1}\frac{\mu_{t+1}\big({\rm d}x_{t}\big)}{\mu_{t}\big({\rm d}x_{t}\big)}\overrightarrow{\varPi}\big({\rm d}(x_{0},\ldots,x_{\tau})\big).
\]
\end{thm}
\begin{proof}
The proof follows from manipulations similar to those of Theorem \ref{thm:AIS}.
We have,
\begin{align*}
\overleftarrow{\varPi} & \big({\rm d}(x_{\tau},\ldots,x_{0})\big)=\mu_{\tau+1}\big({\rm d}x_{\tau}\big)\overleftarrow{\Pi}_{\tau}\big(x_{\tau},{\rm d}x_{\tau-1}\big)\left[\prod_{t=2}^{\tau-1}\Pi_{\tau-t+1}\big(x_{\tau-t+1},{\rm d}x_{\tau-t}\big)\right]\overleftarrow{\Pi}_{1}(x_{1},\mathrm{d}x_{0})\\
 & =\left[\frac{\mu_{\tau}\big({\rm d}x_{\tau-1}\big)}{\mu_{\tau-1}\big({\rm d}x_{\tau-1}\big)}\overrightarrow{\Pi}_{\tau}\big(x_{\tau-1},{\rm d}x_{\tau}\big)\right]\left[\mu_{\tau-1}\big({\rm d}x_{\tau-1}\big)\prod_{t=2}^{\tau-1}\Pi_{t}\big(x_{t},{\rm d}x_{t-1}\big)\right]\overleftarrow{\Pi}_{1}(x_{1},\mathrm{d}x_{0})\\
 & =\left[\frac{\mu_{\tau}\big({\rm d}x_{\tau-1}\big)}{\mu_{\tau-1}\big({\rm d}x_{\tau-1}\big)}\overrightarrow{\Pi}_{\tau}\big(x_{\tau-1},{\rm d}x_{\tau}\big)\right]\left[\left\{ \prod_{t=2}^{\tau-1}\frac{\mu_{t}\big({\rm d}x_{t-1}\big)}{\mu_{t-1}\big({\rm d}x_{t-1}\big)}\right\} \prod_{t=2}^{\tau-1}\Pi_{t}\big(x_{t-1},{\rm d}x_{t}\big)\mu_{1}\big({\rm d}x_{1}\big)\right]\overleftarrow{\Pi}_{1}(x_{1},\mathrm{d}x_{0})\\
 & =\left\{ \prod_{t=1}^{\tau-1}\frac{\mu_{t+1}\big({\rm d}x_{t}\big)}{\mu_{t}\big({\rm d}x_{t}\big)}\right\} \left[\overrightarrow{\Pi}_{\tau}\big(x_{\tau-1},{\rm d}x_{\tau}\big)\prod_{t=2}^{\tau-1}\Pi_{\tau-t+1}\big(x_{\tau-t},{\rm d}x_{\tau-t+1}\big)\overrightarrow{\Pi}_{1}\big(x_{0},{\rm d}x_{1}\big)\mu_{0}\big({\rm d}x_{0}\big)\right]\\
 & =\left\{ \prod_{t=1}^{\tau-1}\frac{\mu_{t+1}\big({\rm d}x_{t}\big)}{\mu_{t}\big({\rm d}x_{t}\big)}\right\} \overrightarrow{\varPi}(\mathrm{d}(x_{0},\ldots,x_{\tau}))
\end{align*}
where on the second and fourth line we have used the two conditions
on the arrowed kernels, and the third line is obtained by applying
Theorem \ref{thm:AIS}.
\end{proof}
\begin{rem}
The additional conditions are satisfied, for example, if $\mathsf{F}=\mathsf{E}$,
$\mu_{0}=\mu_{1}$, $\mu_{\tau}=\mu_{\tau+1}$, $\overrightarrow{\Pi}_{1}=\overleftarrow{\Pi}_{1}$
is $\mu_{1}-$ reversible and likewise $\overrightarrow{\Pi}_{\tau}=\overleftarrow{\Pi}_{\tau}$
is $\mu_{\tau}-$reversible, taking us to the standard AIS setting
(with repeats at the ends). However, the generalisation obtained by
those additional conditions allow for more general scenarios of interest,
in particular for annealing to occur on a space different from that
where $\mu_{0}$ and $\mu_{\tau+1}$ are defined. The application
of our methodology for trans-dimensional models indeed requires this
generalisation; see Sections \ref{subsec: Generalisations of pseudo-marginal asymmetric MCMC}
and\ref{subsec: An application: trans-dimensional distributions}.
\end{rem}

\section{Auxiliary results and proofs for cSMC based algorithms \label{sec: Auxiliary results and proofs for cSMC based algorithms }}

First, we lay out some useful results on SMC, cSMC, for the state-space
model defined in Section \ref{subsec: State-space models and conditional SMC}
dropping $\theta$ from the notation. For notational simplicity we
will consider the bootstrap particle filter where the particles are
initiated from the initial distribution and propagated from the state
transition, so that $h({\rm d}\zeta_{1}^{(i)})=\mu({\rm d}\zeta_{1}^{(i)})$
and $H(\zeta_{t-1}^{a_{t-1}^{(i)}},{\rm d}\zeta_{t}^{(i)})=f(\zeta_{t-1}^{a_{t-1}^{(i)}},{\rm d}\zeta_{t}^{(i)})$
in Algorithm \ref{alg: Conditional SMC} and the particle weight is
simply the observation density, $w_{t}(\zeta_{t}^{(i)})=g(y_{t}|\zeta_{t}^{(i)})$.
Note that our results can be extended to other choices of $h$ and
$H$.

It is standard that the law of a particle filter with $M$ particles
and multinomial resampling for $\zeta\in\mathsf{Z}^{MT}$ and $a\in[M]^{M(T-1)}$
is \citep{Andrieu_et_al_2010}

\begin{align*}
\psi\big({\rm d}(\zeta,a)\big)= & \prod_{i=1}^{M}\mu({\rm d}\zeta_{1}^{(i)})\prod_{t=2}^{T}\left\{ \prod_{i=1}^{M}\frac{w_{t-1}(\zeta_{t-1}^{(a_{t-1}^{(i)})})}{\sum_{j=1}^{M}w_{t-1}(\zeta_{t-1}^{(j)})}f(\zeta_{t-1}^{(a_{t-1}^{(i)})},{\rm d}\zeta_{t}^{(i)})\right\} .
\end{align*}
What is important for us is that the marginal distribution $\psi\big({\rm d}\zeta\big)$
has a simple form
\[
\psi\big({\rm d}\zeta\big)=\prod_{i=1}^{M}\mu({\rm d}\zeta_{1}^{(i)})\prod_{t=2}^{T}\left\{ \prod_{i=1}^{M}\frac{\sum_{j=1}^{M}w_{t-1}(\zeta_{t-1}^{(j)})f(\zeta_{t-1}^{(j)},{\rm d}\zeta_{t}^{(i)})}{\sum_{j=1}^{M}w_{t-1}(\zeta_{t-1}^{(j)})}\right\} .
\]
Now, letting $C:=\ell(y)$ (recall $y=y_{1:T}$) and its estimator
$\hat{C}(\zeta):=\prod_{t=1}^{T}\frac{1}{M}\sum_{i=1}^{M}w_{t}(\zeta_{t}^{(i)})$,
we introduce 
\begin{equation}
\bar{\psi}\big({\rm d}\zeta\big):=\psi\big({\rm d}\zeta\big)\frac{\hat{C}(\zeta)}{C}.\label{eq: SMC to cSMC probability law}
\end{equation}
We know from \citet{Andrieu_et_al_2010} that this is a probability
distribution, and is a way of justifying that $\hat{C}(\zeta)$ is
an unbiased estimator of $C$\textendash note that the ancestral history
is here integrated out. For $\zeta\in\mathsf{Z}^{MT}$ and $\mathbf{k}=(k_{1},\ldots,k_{T})\in[M]^{T}$,
let $\zeta^{(\mathbf{k})}=(\zeta_{1}^{(k_{1})},\ldots,\zeta_{T}^{(k_{T})})$.
Furthermore, for $z\in\mathsf{Z}^{T}$ and $\zeta\in\mathsf{Z}^{MT}$,
define the (extended) cSMC kernel
\[
\Phi(z,{\rm d}(\mathbf{k},\zeta)):=\frac{1}{M^{T}}\delta_{z}({\rm d}\zeta^{(\mathbf{k})})\prod_{i\neq k_{1}}^{M}\mu({\rm d}\zeta_{1}^{(i)})\prod_{t=2}^{T}\left\{ \prod_{i=1,i\neq k_{t}}^{M}\frac{\sum_{j=1}^{M}w_{t-1}(\zeta_{t-1}^{(j)})f(\zeta_{t-1}^{(j)},{\rm d}\zeta_{t}^{(i)})}{\sum_{j=1}^{M}w_{t-1}(\zeta_{t-1}^{(j)})}\right\} ,
\]
with its marginal $\Phi(z,{\rm d}\zeta)=\sum_{\mathbf{k}\in[M]^{T}}\Phi(z,{\rm d}(\mathbf{k},\zeta))$.
Recall the law of the indices used in the backward-sampling procedure
in order to draw a path $\zeta^{(\mathbf{k})}$,
\[
\phi(\mathbf{k}|\zeta):=\frac{w_{T}(\zeta_{T}^{(k_{T})})}{\sum_{i=1}^{M}w_{T}(\zeta_{T}^{(i)})}\prod_{t=2}^{T}\frac{w_{t-1}(\zeta_{t-1}^{(k_{t-1})})f(\zeta_{t-1}^{(k_{t-1})},\zeta_{t}^{(k_{t})})}{\sum_{i=1}^{M}w_{t-1}(\zeta_{t-1}^{(i)})f(\zeta_{t-1}^{(i)},\zeta_{t}^{(k_{t})})}.
\]
with a convention for $f$ when $t=1$. Finally, define the joint
distribution of the indices and path drawn via backward sampling,
\[
\check{\Phi}(\zeta,{\rm d}(\mathbf{k},z))=\phi(\mathbf{k}\mid\zeta)\delta_{\zeta^{(\mathbf{k})}}({\rm d}z).
\]
and its marginal $\check{\Phi}(\zeta,{\rm d}z)=$$\sum_{\mathbf{k}\in[M]^{T}}\check{\Phi}(\zeta,{\rm d}(\mathbf{k},z))$.
\begin{lem}
\label{lem: cSMC semi-reversibility}For any $z\in\mathsf{Z}^{T}$,
$\mathbf{k}\in[M]^{T},$ and $\zeta\in\mathsf{Z}^{MT}$,
\[
\pi({\rm d}z)\Phi(z,{\rm d}(\mathbf{k},\zeta))=\bar{\psi}\big({\rm d}\zeta\big)\check{\Phi}(\zeta,{\rm d}(\mathbf{k},z)).
\]
\end{lem}
\begin{proof}
For the left hand side, we have
\begin{multline*}
\pi\big(\mathrm{d}z\big)\Phi(z,\mathrm{d}(\mathbf{k},\zeta))=\frac{1}{M^{T}}\pi(\mathrm{d}z)\delta_{z}(\mathrm{d}\zeta^{(\mathbf{k})})\\
\times\prod_{i=1,i\neq k_{1}}^{M}\mu({\rm d}\zeta_{1}^{(i)})\prod_{t=2}^{T}\left\{ \prod_{i=1,i\neq k_{t}}^{M}\frac{\sum_{j=1}^{M}w_{t-1}(\zeta_{t-1}^{(j)})f(\zeta_{t-1}^{(j)},{\rm d}\zeta_{t}^{(i)})}{\sum_{j=1}^{M}w_{t-1}(\zeta_{t-1}^{(j)})}\right\} .
\end{multline*}
For the right hand side, first we note the identity
\[
\bar{\psi}\big({\rm d}\zeta\big)\phi(\mathbf{k}\mid\zeta)=\frac{1}{M^{T}}\pi(\mathrm{d}\zeta^{(\mathbf{k})})\prod_{i=1,i\neq k_{1}}^{M}\mu({\rm d}\zeta_{1}^{(i)})\prod_{t=2}^{T}\left\{ \prod_{i=1,i\neq k_{t}}^{M}\frac{\sum_{j=1}^{M}w_{t-1}(\zeta_{t-1}^{(j)})f(\zeta_{t-1}^{(j)},{\rm d}\zeta_{t}^{(i)})}{\sum_{j=1}^{M}w_{t-1}(\zeta_{t-1}^{(j)})}\right\} 
\]
so that we get
\begin{multline*}
\bar{\psi}\big({\rm d}\zeta\big)\check{\Phi}(\zeta,{\rm d}(\mathbf{k},z))=\frac{1}{M^{T}}\pi(\mathrm{d}\zeta^{(\mathbf{k})})\delta_{\zeta^{(\mathbf{k})}}(\mathrm{d}z)\\
\times\prod_{i=1,i\neq k_{1}}^{M}\mu({\rm d}\zeta_{1}^{(i)})\prod_{t=2}^{T}\left\{ \prod_{i=1,i\neq k_{t}}^{M}\frac{\sum_{j=1}^{M}w_{t-1}(\zeta_{t-1}^{(j)})f(\zeta_{t-1}^{(j)},{\rm d}\zeta_{t}^{(i)})}{\sum_{j=1}^{M}w_{t-1}(\zeta_{t-1}^{(j)})}\right\} 
\end{multline*}
which is equal to $\pi\big(\mathrm{d}z\big)\Phi(z,\mathrm{d}(\mathbf{k},\zeta)).$ 
\end{proof}
{} Lemma \ref{lem: cSMC semi-reversibility} immediately leads to the
following corollaries which will be useful in the subsequent proofs.
\begin{cor}
\label{cor: cSMC semi-reversibility}For any $z\in\mathsf{Z}^{T}$and
$\zeta\in\mathsf{Z}^{MT}$,
\[
\pi({\rm d}z)\Phi(z,{\rm d}\zeta)=\bar{\psi}\big({\rm d}\zeta\big)\check{\Phi}(\zeta,{\rm d}z).
\]
\end{cor}
\begin{cor}
\label{cor:exchangeability-SSM}For any $N\geq1$ $(z,u^{(1)},\ldots,u^{(N)})\in\mathsf{Z}^{(N+1)T}$,
$(\mathbf{k},\mathbf{k}^{(1)},\ldots,\mathbf{k}^{(N)})\in[M]^{(N+1)T}$,
and $\zeta\in\mathsf{Z}^{MT}$,
\[
\pi({\rm d}z)\Phi\big(z,{\rm d}(\mathbf{k},\zeta)\big)\prod_{i=1}^{N}\check{\Phi}(\zeta,{\rm d}(\mathbf{k}^{(i)},u^{(i)}))=\bar{\psi}\big({\rm d}\zeta\big)\check{\Phi}(\zeta,{\rm d}(\mathbf{k},z))\prod_{i=1}^{N}\check{\Phi}(\zeta,{\rm d}(\mathbf{k}^{(i)},u^{(i)}))
\]
which establishes that $z,u^{(1)},\ldots,u^{(N)}$ are exchangeable
under the joint distribution 
\[
\pi({\rm d}z)\int_{\zeta}\Phi\big(z,{\rm d}\zeta\big)\prod_{i=1}^{N}\check{\Phi}(\zeta,{\rm d}u^{(i)})=\bar{\psi}\big({\rm d}\zeta\big)\int_{\zeta}\check{\Phi}(\zeta,{\rm d}z)\prod_{i=1}^{N}\check{\Phi}(\zeta,{\rm d}u^{(i)}).
\]
\end{cor}

\subsection{Proof of unbiasedness for the acceptance/likelihood ratio estimator
of Algorithm \ref{alg: MHAAR for SSM with cSMC - Rao-Blackwellised backward sampling}
\label{subsec: Proof of unbiasedness for the Rao-Blackwellised estimator}}

From here on, we have $\theta$ back in the notation. Let $F:\mathsf{Z}^{T}\rightarrow\mathbb{R}$
be a real-valued function and given $\zeta\in\mathsf{Z}^{TM}$, denote
$\check{\Phi}_{\theta}(\zeta,F)$ its conditional expectation with
respect to the backward sampling distribution $\check{\Phi}_{\theta}(\zeta,\cdot)$,
\[
\check{\Phi}_{\theta}(\zeta,F)=\sum_{\mathbf{k}\in[M]^{T}}F(\zeta^{(k)})\phi_{\theta}(\mathbf{k}|\zeta).
\]
which is a function of $\zeta$. It is a result from \citet[Theorem 5.2]{del2010backward}
that for any $F:\mathsf{Z}^{T}\rightarrow\mathbb{R}$, the expectation
of $\check{\Phi}_{\theta}(\zeta,F)$, scaled by $\hat{C}_{\theta}(\zeta)/C_{\theta}$,
with respect to the law of SMC, $\psi_{\theta}$ is $\pi_{\theta}(F)$:
\[
\psi_{\theta}\left(\frac{\hat{C}_{\theta}(\zeta)}{C_{\theta}}\check{\Phi}_{\theta}(\zeta,F)\right)=\pi_{\theta}(F).
\]
The crucial point here is that we can rewrite the identity above in
terms of $\bar{\psi}_{\theta}$ as
\begin{equation}
\bar{\psi}_{\theta}\left(\check{\Phi}_{\theta}(\zeta,F)\right)=\pi_{\theta}(F).\label{eq: SMC Del Morals result rewritten in terms of cSMC}
\end{equation}
owing to \eqref{eq: SMC to cSMC probability law}. Now, we have the
necessary intermediate results to prove Theorem \ref{thm: SMC unbiased estimator of acceptance ratio}.
\begin{proof}
(Theorem \ref{thm: SMC unbiased estimator of acceptance ratio}) Let
$\gamma_{\theta}(z):=p_{\theta}(z,y)$ be the unnormalised density
for $\pi_{\theta}(z)$ so that $\gamma_{\theta}(z)=\pi(\theta)\ell_{\theta}(y)$.
We can write the estimator in \eqref{eq: SMC acceptance ratio estimator all paths}
as
\[
\mathring{r}_{z,\zeta}(\theta,\theta';\tilde{\theta})=\frac{q(\theta',\theta)}{q(\theta,\theta')}\frac{\eta(\theta')}{\eta(\theta)}\frac{\gamma_{\tilde{\theta}}(z)}{\gamma_{\theta}(z)}\check{\Phi}_{\tilde{\theta}}\left(\zeta,\frac{\gamma_{\theta'}}{\gamma_{\tilde{\theta}}}(\cdot)\right).
\]
The expectation of $\mathring{r}_{z,\zeta}(\theta,\theta';\tilde{\theta})$
with respect to the law of the mechanism described in Theorem \ref{thm: SMC unbiased estimator of acceptance ratio}
that generates $\mathring{r}_{z,\zeta}(\theta,\theta;\tilde{\theta})$
is 
\[
\int\pi_{\theta}(\mathrm{d}z)\Phi_{\tilde{\theta}}(z,\mathrm{d}\zeta)\mathring{r}_{z,\zeta}(\theta,\theta';\tilde{\theta}).
\]
To see that this is indeed $r(\theta,\theta')$, firstly observe that
\begin{equation}
\pi_{\theta}(\mathrm{d}z)\frac{\gamma_{\tilde{\theta}}(z)}{\gamma_{\theta}(z)}=\frac{\gamma_{\tilde{\theta}}(z)}{\gamma_{\theta}(z)}\frac{\pi_{\theta}(z)}{\pi_{\tilde{\theta}}(z)}\pi_{\tilde{\theta}}(\mathrm{d}z)=\frac{\ell_{\tilde{\theta}}(y)}{\ell_{\theta}(y)}\pi_{\tilde{\theta}}(\mathrm{d}z).\label{eq: SMC unbiasedness proof first step}
\end{equation}
Secondly, using Corollary \ref{cor: cSMC semi-reversibility}, we
have
\begin{equation}
\pi_{\tilde{\theta}}(\mathrm{d}z)\Phi_{\tilde{\theta}}(z,\mathrm{d}\zeta)\check{\Phi}_{\tilde{\theta}}\left(\zeta,\frac{\gamma_{\theta'}}{\gamma_{\tilde{\theta}}}(\cdot)\right)=\bar{\psi}_{\tilde{\theta}}(\mathrm{d}\zeta)\check{\Phi}_{\tilde{\theta}}\left(\zeta,\frac{\gamma_{\theta'}}{\gamma_{\tilde{\theta}}}(\cdot)\right)\check{\psi}(\zeta,\mathrm{d}z).\label{eq: SMC unbiasedness proof second step}
\end{equation}
Therefore, we have
\begin{align*}
\int\pi_{\theta}(\mathrm{d}z)\Phi_{\tilde{\theta}}(z,\mathrm{d}\zeta)\mathring{r}_{z,\zeta}(\theta,\theta';\tilde{\theta}) & =\frac{q(\theta',\theta)}{q(\theta,\theta')}\frac{\eta(\theta')}{\eta(\theta)}\frac{\ell_{\tilde{\theta}}(y)}{\ell_{\theta}(y)}\int\bar{\psi}_{\tilde{\theta}}(\mathrm{d}\zeta)\check{\Phi}_{\tilde{\theta}}\left(\zeta,\frac{\gamma_{\theta'}}{\gamma_{\tilde{\theta}}}(\cdot)\right)\check{\psi}(\zeta,\mathrm{d}z)\\
 & =\frac{q(\theta',\theta)}{q(\theta,\theta')}\frac{\eta(\theta')}{\eta(\theta)}\frac{\ell_{\tilde{\theta}}(y)}{\ell_{\theta}(y)}\pi_{\tilde{\theta}}\left(\frac{\gamma_{\theta'}}{\gamma_{\tilde{\theta}}}(\cdot)\right)\\
 & =r(\theta,\theta')
\end{align*}
where the first line is due to \eqref{eq: SMC unbiasedness proof first step}
and \eqref{eq: SMC unbiasedness proof second step}, the second line
follows from \eqref{eq: SMC Del Morals result rewritten in terms of cSMC}
and the last line is due to the identity $\pi_{\tilde{\theta}}\left(\gamma_{\theta'}/\gamma_{\tilde{\theta}}\right)=\ell_{\theta'}(y)/\ell_{\tilde{\theta}}(y)$.
\end{proof}

\subsection{Proof of reversibility for Algorithms \ref{alg: MHAAR for SSM with cSMC - Rao-Blackwellised backward sampling}
and \ref{alg: MHAAR for SSM with cSMC - multiple paths from backward sampling}
\label{subsec: Proof of reversibility for Algorithms}}

First we show the reversibility of Algorithm \ref{alg: MHAAR for SSM with cSMC - Rao-Blackwellised backward sampling}
that uses the Rao-Blackwellised estimator of the acceptance ratio
\begin{thm}
The transition probability of Algorithm \ref{alg: MHAAR for SSM with cSMC - Rao-Blackwellised backward sampling}
satisfies the detailed balance with respect to $\pi(\mathrm{d}(\theta,z))$.
\end{thm}
\begin{proof}
Let $u=(\mathbf{k},\zeta,\mathbf{k}')\in[M]^{T}\times\mathrm{Z}^{TM}\times[M]^{T}$.
The proposal kernels that correspond to the moves of Algorithm \ref{alg: MHAAR for SSM with cSMC - Rao-Blackwellised backward sampling}
are
\begin{align*}
Q_{1}^{M}(x,\mathrm{d}(y,u)) & =q(\theta,\mathrm{d}\theta')\Phi_{\tilde{\theta}_{1}(\theta,\theta')}(z,\mathrm{d}(\mathbf{k},\zeta))\frac{\phi_{\tilde{\theta}_{1}(\theta,\theta')}(\mathbf{k}'|\zeta)\mathring{r}_{z,\zeta^{(\mathbf{k}')}}(\theta,\theta';\tilde{\theta}_{1}(\theta,\theta'))}{\sum_{\mathbf{l}\in[M]^{T}}\phi_{\tilde{\theta}_{1}(\theta,\theta')}(\mathbf{l}|\zeta)\mathring{r}_{z,\zeta^{(\mathbf{l})}}(\theta,\theta';\tilde{\theta}_{1}(\theta,\theta'))}\delta_{\zeta^{(\mathbf{k}')}}(\mathrm{d}z'),\\
Q_{2}^{M}(x,\mathrm{d}(y,u)) & =q(\theta,\mathrm{d}\theta')\Phi_{\tilde{\theta}_{2}(\theta,\theta')}(z,\mathrm{d}(\mathbf{k},\zeta))\check{\Phi}_{\tilde{\theta}_{2}(\theta,\theta')}(\zeta,{\rm d}(\mathbf{k}',z')).
\end{align*}
First, observe that, for any $z,z'\in\mathsf{Z}^{T}$, and $\theta,\theta',\tilde{\theta}\in\Theta$,
equation \eqref{eq: AIS acceptance ratio for SSM} can be rewritten
as 
\begin{align}
\mathring{r}_{z,z'}(\theta,\theta';\tilde{\theta}) & =\frac{q(\theta',\theta)}{q(\theta,\theta')}\frac{\eta(\theta')}{\eta(\theta)}\frac{\pi(\mathrm{d}(\theta',z'))}{\pi(\mathrm{d}(\tilde{\theta},z))}\frac{\pi(\mathrm{d}(\tilde{\theta},z'))}{\pi(\mathrm{d}(\theta,z))}\nonumber \\
 & =r(\theta,\theta')\frac{\pi_{\theta'}(\mathrm{d}z')}{\pi_{\tilde{\theta}}(\mathrm{d}z')}\frac{\pi_{\tilde{\theta}}(\mathrm{d}z)}{\pi_{\theta}(\mathrm{d}z)}.\label{eq: acceptance ratio modified}
\end{align}
From Corollary \ref{cor:exchangeability-SSM}, for $(\theta,z)\in\mathrm{X}$,
$u=(\mathbf{k},\zeta,\mathbf{k}')\in[M]^{T}\times\mathsf{Z}^{TM}\times[M]^{T}$,
and $z'\in\mathsf{Z}{}^{T}$ we have
\begin{align}
\pi_{\theta}(z)\Phi_{\theta}(z,\mathrm{d}(\mathbf{k},\zeta))\check{\Phi}_{\theta}(\zeta,{\rm d}(\mathbf{k}',z')) & =\pi_{\theta}(\mathrm{d}z')\Phi_{\theta}(z',\mathrm{d}(\mathbf{k}',\zeta))\check{\Phi}_{\theta}(\zeta,{\rm d}(\mathbf{k},z)).\label{eq: result from corollary for exchangeability}
\end{align}
Using those relations, and letting $\tilde{\theta}=\tilde{\theta}_{1}(\theta,\theta')=\tilde{\theta}_{2}(\theta',\theta)$,
we arrive the Radon-Nikodym derivative
\begin{align}
\frac{\pi(\mathrm{d}\theta')\pi_{\theta'}(\mathrm{d}z')Q_{2}^{M}(y,\mathrm{d}(x,u))}{\pi(\mathrm{d}\theta)\pi_{\theta}(\mathrm{d}z)Q_{1}^{M}(x,\mathrm{d}(y,u))} & =r(\theta,\theta')\frac{\pi_{\theta'}(\mathrm{d}z')\Phi_{\tilde{\theta}}(z',\mathrm{d}(\mathbf{k}',\zeta))\check{\Phi}_{\tilde{\theta}}(\zeta,\mathrm{d}(\mathbf{k},z))}{\pi_{\theta}(\mathrm{d}z)\Phi_{\tilde{\theta}}(z,\mathrm{d}(\mathbf{k},\zeta))\frac{\phi_{\tilde{\theta}}(\mathbf{k}'|\zeta)\mathring{r}_{z,\zeta^{(\mathbf{k}')}}(\theta,\theta';\tilde{\theta})}{\sum_{\mathbf{l}\in[M]^{T}}\phi_{\tilde{\theta}}(\mathbf{l}|\zeta)\mathring{r}_{z,\zeta^{(\mathbf{l})}}(\theta,\theta';\tilde{\theta})}\delta_{\zeta^{(\mathbf{k}')}}(\mathrm{d}z')}\nonumber \\
 & =\frac{r(\theta,\theta')}{\mathring{r}_{z,z'}(\theta,\theta',\tilde{\theta})}\frac{\pi_{\theta'}(\mathrm{d}z')}{\pi_{\tilde{\theta}}(\mathrm{d}z')}\frac{\pi_{\tilde{\theta}}(\mathrm{d}z)\Phi_{\tilde{\theta}}(z,\mathrm{d}(\mathbf{k},\zeta))\check{\Phi}_{\tilde{\theta}}(\zeta,\mathrm{d}(\mathbf{k}',z'))}{\pi_{\theta}(\mathrm{d}z)\Phi_{\tilde{\theta}}(z,\mathrm{d}(\mathbf{k},\zeta))\check{\Phi}_{\tilde{\theta}}(\zeta,\mathrm{d}(\mathbf{k}',z'))}\mathring{r}_{z,\zeta}(\theta,\theta';\tilde{\theta})\nonumber \\
 & =\frac{r(\theta,\theta')}{\mathring{r}_{z,z'}(\theta,\theta',\tilde{\theta})}\frac{\pi_{\theta'}(\mathrm{d}z')\pi_{\tilde{\theta}}(\mathrm{d}z)}{\pi_{\tilde{\theta}}(\mathrm{d}z')\pi_{\theta}(\mathrm{d}z)}\frac{\pi_{\theta}(\mathrm{d}z)\Phi_{\tilde{\theta}}(z,\mathrm{d}(\mathbf{k},\zeta))\check{\Phi}_{\tilde{\theta}}(\zeta,\mathrm{d}(\mathbf{k}',z'))}{\pi_{\theta}(\mathrm{d}z)\Phi_{\tilde{\theta}}(z,\mathrm{d}(\mathbf{k},\zeta))\check{\Phi}_{\tilde{\theta}}(\zeta,\mathrm{d}(\mathbf{k}',z'))}\mathring{r}_{z,\zeta}(\theta,\theta';\tilde{\theta})\nonumber \\
 & =\mathring{r}_{z,\zeta}(\theta,\theta';\tilde{\theta}).\label{eq: RN derivative aMCMC for HMM all paths}
\end{align}
\end{proof}
The analysis in the proof above not only bears an alternative proof
of Theorem \ref{thm: SMC unbiased estimator of acceptance ratio}
on the unbiasedness of \eqref{eq: SMC acceptance ratio estimator all paths}
but also implicitly proves Corollary \ref{cor: SMC unbiased estimator of acceptance ratio};
as we show below.
\begin{proof}
(Theorem \ref{thm: SMC unbiased estimator of acceptance ratio}) Equation
\eqref{eq: RN derivative aMCMC for HMM all paths} can be modified
to obtain
\[
\pi_{\theta}(\mathrm{d}z)\Phi_{\tilde{\theta}}(z,\mathrm{d}(\mathbf{k},\zeta))\frac{\phi_{\tilde{\theta}}(\mathbf{k}'|\zeta)\mathring{r}_{z,\zeta^{(\mathbf{k}')}}(\theta,\theta';\tilde{\theta})}{\sum_{\mathbf{l}\in[M]^{T}}\phi_{\tilde{\theta}}(\mathbf{l}|\zeta)\mathring{r}_{z,\zeta^{(\mathbf{l})}}(\theta,\theta';\tilde{\theta})}\delta_{\zeta^{(\mathbf{k}')}}(\mathrm{d}z')\frac{\mathring{r}_{z,\zeta}(\theta,\theta';\tilde{\theta})}{r(\theta,\theta')}=\pi_{\theta'}(\mathrm{d}z')\Phi_{\tilde{\theta}}(z',\mathrm{d}(\mathbf{k}',\zeta))\check{\Phi}_{\tilde{\theta}}(\zeta,{\rm d}(\mathbf{k},z)).
\]
Integrating both sides with respect to all the variables except $\theta$
and $\theta'$ leads to
\[
\int\pi_{\theta}(\mathrm{d}z)\Phi_{\tilde{\theta}}(z,\mathrm{d}\zeta)\mathring{r}_{z,\zeta}(\theta,\theta';\tilde{\theta})=r(\theta,\theta')
\]
upon noticing that $\mathring{r}_{z,\zeta}(\theta,\theta';\tilde{\theta})$
does not depend on $\mathbf{k}'$ or $z'$ and the right hand side
is a probability distribution. Noting that $\pi_{\theta}(\mathrm{d}z)\Phi_{\tilde{\theta}}(z,\mathrm{d}\zeta)$
is exactly the distribution of the mechanism described in Theorem
\ref{thm: SMC unbiased estimator of acceptance ratio} that generates
$\mathring{r}_{z,\zeta}(\theta,\theta';\tilde{\theta})$, we prove
Theorem \ref{thm: SMC unbiased estimator of acceptance ratio}. 
\end{proof}
\begin{proof}
(Corollary \ref{cor: SMC unbiased estimator of acceptance ratio})
Similarly to the previous proof, we can write
\begin{multline*}
\pi_{\theta}(\mathrm{d}z)\Phi_{\tilde{\theta}}(z,\mathrm{d}(\mathbf{k},\zeta))\check{\Phi}_{\tilde{\theta}}(\zeta,{\rm d}(\mathbf{k}',z'))\frac{r(\theta',\theta)}{\mathring{r}_{z',\zeta}(\theta',\theta;\tilde{\theta})}\\
=\pi_{\theta}(\mathrm{d}z')\Phi_{\tilde{\theta}}(z',\mathrm{d}(\mathbf{k}',\zeta))\frac{\phi_{\tilde{\theta}}(\mathbf{k}|\zeta)\mathring{r}_{z',\zeta^{(\mathbf{k})}}(\theta',\theta;\tilde{\theta})}{\sum_{\mathbf{l}\in[M]^{T}}\phi_{\tilde{\theta}}(\mathbf{l}|\zeta)\mathring{r}_{z',\zeta^{(\mathbf{l})}}(\theta',\theta;\tilde{\theta})}\delta_{\zeta^{(\mathbf{k}')}}(\mathrm{d}z').
\end{multline*}
Again, integrating both sides with respect to all the variables except
$\theta$ and $\theta'$ leads to
\[
\int\pi_{\theta}(\mathrm{d}z)\Phi_{\tilde{\theta}}(z,\mathrm{d}\zeta)\check{\Phi}_{\tilde{\theta}}(\zeta,{\rm d}z')(1/\mathring{r}_{z',\zeta}(\theta',\theta;\tilde{\theta}))=r(\theta,\theta').
\]
Since $1/\mathring{r}_{z',\zeta}(\theta',\theta;\tilde{\theta})$
is the estimator in question in Corollary \ref{cor: SMC unbiased estimator of acceptance ratio}
and $\pi_{\theta}(\mathrm{d}z)\Phi_{\tilde{\theta}}(z,\mathrm{d}\zeta)\check{\Phi}_{\tilde{\theta}}(\zeta,{\rm d}z')$
is exactly the distribution of the described mechanism that generates
it, we prove Corollary \ref{cor: SMC unbiased estimator of acceptance ratio}.
\end{proof}
Next, we show the reversibility of Algorithm \ref{alg: MHAAR for SSM with cSMC - multiple paths from backward sampling}
that uses a subsampled version of the Rao-Blackwellised acceptance
ratio estimator.
\begin{thm}
The transition probability of Algorithm \ref{alg: MHAAR for SSM with cSMC - multiple paths from backward sampling}
satisfies the detailed balance with respect to $\pi(\mathrm{d}x)$.
\end{thm}
\begin{proof}
For any $\theta\in\Theta$ and $z\in\mathsf{Z}^{T},$ define the kernel
on ($\mathsf{Z}^{TN},\mathscr{Z}^{\otimes TN}$) for $N$ paths drawn
via backward sampling following cSMC at $\theta$ conditioned on $z$
\[
R_{\theta}(z,\mathrm{d}(u^{(1)},\ldots,u^{(N)}))=\int_{\zeta}\Phi_{\theta}\big(z,{\rm d}\zeta\big)\prod_{i=1}^{N}\check{\Phi}_{\theta}(\zeta,{\rm d}u^{(i)}).
\]
By the exchangeability result of Corollary \ref{cor:exchangeability-SSM},
it holds for any $0\leq k\leq N$ that
\[
\pi_{\theta}({\rm d}u^{(0)})R_{\theta}(u^{(0)},\mathrm{d}u^{(1:N)})=\pi_{\theta}({\rm d}u^{(k)})R_{\theta}(u^{(k)},\mathrm{d}u^{(-k)}),
\]
where $u^{(-k)}=(u^{(0)},\ldots,u^{(k-1)},u^{(k+1)},\ldots,u^{(N)})$,
and therefore
\begin{equation}
\pi_{\theta}({\rm d}z)\delta_{z}(\mathrm{d}u^{(0)})R_{\theta}(u^{(0)},{\rm d}u^{(1:N)})\delta_{u^{(k)}}(\mathrm{d}z')=\pi_{\theta}({\rm d}z')\delta_{z'}(\mathrm{d}u^{(k)})R_{\theta}(u^{(k)},{\rm d}u^{(-k)})\delta_{u^{(0)}}(\mathrm{d}z).\label{eq: exchangeability of auxiliary variables}
\end{equation}
Letting $\mathfrak{u}=(u^{(0)},\ldots,u^{(N)})$, the proposal kernels
that correspond to the moves of Algorithm \ref{alg: MHAAR for SSM with cSMC - multiple paths from backward sampling}
are
\begin{align*}
Q_{1}^{M,N}\big(x;{\rm d}(y,\mathfrak{u},k)\big) & =q(\theta,{\rm d}\theta')\delta_{z}(\mathrm{d}u^{(0)})R_{\tilde{\theta}_{1}(\theta,\theta')}(u^{(0)},\mathrm{d}u^{(1:N)})\frac{\mathring{r}_{z,u^{(k)}}(\theta,\theta';\tilde{\theta}_{1}(\theta,\theta'))}{\sum_{i=1}^{N}\mathring{r}_{z,u^{(i)}}(\theta,\theta';\tilde{\theta}_{1}(\theta,\theta'))}\delta_{u^{(k)}}({\rm d}z'),\\
Q_{2}^{M,N}\big(x;{\rm d}(y,\mathfrak{u},k)\big) & =q(\theta,{\rm d}\theta')\frac{1}{N}\delta_{z}(\mathrm{d}u^{(k)})R_{\tilde{\theta}_{2}(\theta',\theta)}(u^{(k)},\mathrm{d}u^{(-k)})\delta_{u^{(0)}}({\rm d}z).
\end{align*}
Now we use \eqref{eq: exchangeability of auxiliary variables}, letting
$\tilde{\theta}=\tilde{\theta}_{1}(\theta,\theta')=\tilde{\theta}_{2}(\theta',\theta)$,
we can write
\begin{align*}
\pi_{\theta'}(\mathrm{d}z')Q_{2}^{M,N}\big(y;{\rm d}(x,\mathfrak{u},k)\big) & =q(\theta',{\rm d}\theta)\frac{\pi_{\theta'}(\mathrm{d}z')}{\pi_{\tilde{\theta}}(\mathrm{d}z')}\frac{1}{N}\pi_{\tilde{\theta}}({\rm d}z')\delta_{z'}(\mathrm{d}u^{(k)})R_{\tilde{\theta}}(u^{(k)},{\rm d}u^{(-k)})\delta_{u^{(0)}}(\mathrm{d}z)\\
 & =q(\theta',{\rm d}\theta)\frac{\pi_{\theta'}(\mathrm{d}z')}{\pi_{\tilde{\theta}}(\mathrm{d}z')}\frac{1}{N}\pi_{\tilde{\theta}}({\rm d}z)\delta_{z}(\mathrm{d}u^{(0)})R_{\tilde{\theta}}(u^{(0)},{\rm d}u^{(1:N)})\delta_{u^{(k)}}(\mathrm{d}z').
\end{align*}
Exploiting the relation between above and
\begin{multline*}
\pi_{\theta}({\rm d}z)Q_{1}^{M,N}\big(x;{\rm d}(y,\mathfrak{u},k)\big)\\
=q(\theta,{\rm d}\theta')\frac{\pi_{\theta}(\mathrm{d}z)}{\pi_{\tilde{\theta}}(\mathrm{d}z)}\pi_{\tilde{\theta}}({\rm d}z)\delta_{z}(\mathrm{d}u^{(0)})R_{\tilde{\theta}}(u^{(0)},{\rm d}u^{(1:N)})\delta_{u^{(k)}}(\mathrm{d}z')\frac{\mathring{r}_{z,z'}(\theta,\theta';\tilde{\theta})}{\sum_{i=1}^{N}\mathring{r}_{z,u^{(i)}}(\theta,\theta';\tilde{\theta})},
\end{multline*}
and finally noting \eqref{eq: acceptance ratio modified}, we conclude
\[
\frac{\pi({\rm d}y)Q_{2}^{M,N}\big(y;{\rm d}(x,\mathfrak{u},k)\big)}{\pi({\rm d}x)Q_{1}^{M,N}\big(x;{\rm d}(y,\mathfrak{u},k)\big)}=\frac{1}{N}\sum_{i=1}^{N}\mathring{r}_{z,u^{(i)}}(\theta,\theta';\tilde{\theta}).
\]
\end{proof}

\section{Substituting SMC for AIS in the acceptance ratio in MHAAR \label{sec: Substituting SMC for AIS in the acceptance ratio in MHAAR}}

To avoid repeats, we restrict ourselves to the description of the
generalising Algorithm \ref{alg: MHAAR for pseudo-marginal ratio in latent variable models},
i.e.\ when $\pi(x)=\pi(\theta,z)$. Therefore, let us go back to
the setting in Section \ref{sec: Pseudo-marginal ratio algorithms for latent variable models},
where we have the joint distribution $\pi(x)=\pi(\theta,z)$, the
unnormalised densities for the intermediate steps of AIS, $\pi_{\theta,\theta',t}\propto f_{\theta,\theta',t}$,
$t=0,\ldots,T+1$, $R_{\theta},$ and $R_{\theta,\theta',t}$, $t=1,\ldots,T$,
as detailed in Proposition \ref{prop:MHwithAISinside}. 

Consider $Q_{1}$ of Algorithm \ref{alg: MHAAR for pseudo-marginal ratio in latent variable models}.
Instead of AIS in Algorithm \ref{alg: MHAAR for pseudo-marginal ratio in latent variable models},
we want the sample paths $u_{0:T}^{(i)}$, $i=1,\ldots,N$ to interact
via an SMC algorithm that uses resampling in the annealing steps.
Recalling the definition for $Q_{\theta,\theta',z}$ in equation \eqref{eq: M},
the SMC algorithm that executes this change has the following unnormalised
target distribution 
\begin{equation}
\hat{A}_{\theta,\theta'z}(\mathrm{d}u)=Q_{\theta,\theta',z}(\mathrm{d}u)\prod_{t=0}^{T}\frac{f_{\theta,\theta't+1}(u_{t})}{f_{\theta,\theta't}(u_{t})}.\label{eq: particle AIS SMC target}
\end{equation}
Let us define $C_{\theta,\theta',z}:=\int\hat{A}_{\theta,\theta',z}(\mathrm{d}u)$
so that the normalised target distribution of the SMC is 
\begin{equation}
A_{\theta,\theta',z}(\mathrm{d}u)=\frac{\hat{A}_{\theta,\theta'z}(\mathrm{d}u)}{C_{\theta,\theta',z}}.\label{eq: particle AIS normalised SMC target}
\end{equation}
One important observation is that 
\[
\int\pi_{\theta}(\mathrm{d}z)C_{\theta,\theta',z}=\frac{\pi(\theta')}{\pi(\theta)}.
\]
Denote all the particles generated by the SMC by $\mathfrak{\zeta}=u_{0:T}^{(1:N)}$
and let $\psi_{\theta,\theta',z}$ be the law of $\mathfrak{\zeta}$
with respect to the SMC that targets $A_{\theta,\theta',z}$. Notice
that the ratio 
\[
\hat{C}_{\theta,\theta',z}(\zeta)=\prod_{t=0}^{T}\frac{1}{N}\sum_{i=1}^{N}\frac{f_{\theta,\theta',t+1}(u_{t}^{(i)})}{f_{\theta,\theta',t}(u_{t}^{(i)})},
\]
is the unbiased SMC estimator of $C_{\theta,\theta',z}$, so that
\[
\int\pi_{\theta}(\mathrm{d}z)\psi_{\theta,\theta',z}(\mathrm{d}\zeta)\hat{C}_{\theta,\theta',z}(\zeta)=\frac{\pi(\theta')}{\pi(\theta)}.
\]
Then, a sensible candidate for the acceptance ratio would be
\[
\mathring{r}_{\zeta}^{N}(\theta,\theta'):=\frac{q(\theta',\theta)}{q(\theta,\theta')}\hat{C}_{\theta,\theta',z}(\zeta).
\]
It turns out that we can develop an SMC based MHAAR algorithm that
uses $\mathring{r}_{\zeta}^{N}(\theta,\theta')$; this is shown in
Algorithm \ref{alg: MHAAR-SMC for general latent variable models}.
We prove its reversibility in the subsequent theorem.

\begin{algorithm}[!h]
\caption{MHAR for averaged SMC PMR estimators for general latent variable models}
\label{alg: MHAAR-SMC for general latent variable models}

\KwIn{Current sample $X_{n}=x=(\theta,z)$}

\KwOut{New sample $X_{n+1}$}

Sample $\theta'\sim q(\theta,\cdot)$ and $v\sim\mathcal{U}(0,1)$.
\\
\If{$v\leq1/2$}{

\For{$i=1,\ldots,N$ }{ 

Sample $u_{0}^{(i)}\sim R_{\theta}(z,\cdot)$.

}

\For{$t=1,\ldots,T$}{

\For{$i=1,\ldots,N$}{

Sample $u_{t}^{(i)}\sim\sum_{j=1}^{N}\frac{w_{t-1}^{(j)}}{\sum_{l=1}^{N}w_{t-1}^{(l)}}R_{\theta,\theta',t}(u_{t-1}^{(j)},\cdot)$,
where $w_{t-1}^{(j)}=\frac{f_{\theta,\theta',t}(u_{t-1}^{(j)})}{f_{\theta,\theta',t-1}(u_{t-1}^{(j)})}$.

}

Sample $k\sim\mathcal{P}\big(w_{T}^{(1)},\ldots,w_{T}^{(N)}\big)$
and $z'\sim R_{\theta'}(u_{T}^{(k)},\cdot)$.\\
Set $X_{n+1}=(\theta',z')$ with probability $\min\{1,\mathring{r}_{\zeta}^{N}(\theta,\theta')\},$
otherwise set $X_{n+1}=x$.

}

}\Else{ 

Sample $u_{T}^{(1)}\sim R_{\theta'}(z,\cdot)$.\\
\For{ $t=T,\ldots,1$}{

Sample $u_{t-1}^{(1)}\sim R_{\theta,\theta',z}(u_{t}^{(1)},\cdot)$.

} 

Sample $z'\sim R_{\theta'}(u_{0}^{(1)},\cdot).$\\
\For{ $i=2,\ldots,N$}{

Sample $u_{0}^{(i)}\sim R_{\theta'}(z',\cdot)$.

}

\For{ $t=1,\ldots,T$}{

\For{ $i=2,\ldots,N$}{

Sample $u_{t}^{(i)}\sim\sum_{j=1}^{N}\frac{w_{t-1}^{(j)}}{\sum_{l=1}^{N}w_{t-1}^{(l)}}R_{\theta',\theta,t}(u_{t-1}^{(j)},\cdot)$,
where $w_{t}^{(i)}=\frac{f_{\theta',\theta,t}(u_{t-1}^{(i)})}{f_{\theta',\theta,t-1}(u_{t-1}^{(i)})}$.

}

}

Set $X_{n+1}=(\theta',z')$ with probability $\min\{1,1/\mathring{r}_{\zeta}^{N}(\theta',\theta)\}$,
otherwise set $X_{n+1}=x$. 

}
\end{algorithm}
\begin{thm}
\label{thm: reversibility of particle AIS}The transition kernel of
Algorithm \ref{alg: MHAAR-SMC for general latent variable models}
satisfies the detailed balance with respect to $\pi$. 
\end{thm}
\begin{proof}
Since $\hat{C}_{\theta,\theta'z}(\zeta)$ is an unbiased SMC estimator
of $C_{\theta,\theta',z}$, we can define the probability distribution
\begin{equation}
\bar{\psi}_{\theta,\theta',z}(\mathrm{d}\zeta)=\frac{\hat{C}_{\theta,\theta',z}(\zeta)}{C_{\theta,\theta',z}}\psi_{\theta,\theta',z}(\mathrm{d}\zeta).\label{eq: Particle AIS derivation SMC to cSMC}
\end{equation}
Denote the law of all the particles in cSMC conditioned on $u$ by
$\Phi_{\theta,\theta'}(u,\cdot)$ and the law of the path obtained
by backward sampling given particles $\zeta$ by $\check{\Phi}_{\theta,\theta'}(\zeta,\cdot)$.
Then, $Q_{1}$ and $Q_{2}$ of Algorithm \eqref{alg: MHAAR-SMC for general latent variable models}
can be written as
\begin{align*}
Q_{1}(x,\mathrm{d}(y,\zeta,u)) & =q(\theta,\mathrm{d}\theta')\psi_{\theta,\theta',z}(\mathrm{d}\zeta)\check{\Phi}_{\theta,\theta'}(\zeta,\mathrm{d}u)R_{\theta'}(u_{T},\mathrm{d}z'),\\
Q_{2}(x,\mathrm{d}(y,\zeta,u)) & =q(\theta,\mathrm{d}\theta')\bar{Q}_{\theta,\theta',z}(\mathrm{d}u)R_{\theta'}(u_{0},\mathrm{d}z')\Phi_{\theta',\theta}(u,\mathrm{d}\zeta).
\end{align*}
where $\bar{Q}_{\theta,\theta',z}$ is defined with the involution
$\varphi(u_{0},\ldots,u_{T})=(u_{T},\ldots,u_{0})$ as before. Note
that in practice, we do not need to generate or store all the variables
involved in $Q_{1}$ and $Q_{2}$. This is reflected in Algorithm
\ref{alg: MHAAR-SMC for general latent variable models} where there
is no direct reference to $u$ in $Q_{1}(x,\mathrm{d}(y,\zeta,u))$
and $Q_{2}(x,\mathrm{d}(y,\zeta,u))$. Indeed, in the calculation
of the acceptance ratio we only use $(\theta,z$), $\zeta$, and $(\theta',z')$.
A similar shortcut taken in the implementation of the algorithm is
in the labelling of the conditioned path in $Q_{2}$. However; we
choose to formally define $Q_{1}$ and $Q_{2}$ as above, since with
those definitions it is straightforward to show the detailed balance. 

Using equation \eqref{eq: Particle AIS derivation SMC to cSMC}, Corollary
\ref{cor: cSMC semi-reversibility}, and \eqref{eq: particle AIS SMC target}
in order, we have
\begin{align*}
\psi_{\theta,\theta',z}(\mathrm{d}\zeta)\check{\Phi}_{\theta,\theta'}(\zeta,\mathrm{d}u) & =\frac{C_{\theta,\theta',z}}{\hat{C}_{\theta,\theta',z}(\zeta)}\bar{\psi}_{\theta,\theta',z}(\mathrm{d}\zeta)\check{\Phi}_{\theta,\theta'}(\zeta,\mathrm{d}u)\\
 & =\frac{C_{\theta,\theta',z}}{\hat{C}_{\theta,\theta',z}(\zeta)}A_{\theta,\theta',z}(\mathrm{d}u)\Phi_{\theta,\theta'}(u,\mathrm{d}\zeta)\\
 & =\frac{1}{\hat{C}_{\theta,\theta',z}(\zeta)}\hat{A}_{\theta,\theta',z}(\mathrm{d}u)\Phi_{\theta,\theta'}(u,\mathrm{d}\zeta)\\
 & =\frac{1}{\hat{C}_{\theta,\theta',z}(\zeta)}\prod_{t=0}^{T}\frac{f_{\theta,\theta't+1}(u_{t})}{f_{\theta,\theta't}(u_{t})}Q_{\theta,\theta',z}(\mathrm{d}u)\Phi_{\theta,\theta'}(u,\mathrm{d}\zeta)
\end{align*}
Therefore, we arrive at the Radon-Nikodym derivative
\begin{align*}
\frac{\pi(\mathrm{d}y)Q_{2}(y,\mathrm{d}(x,\zeta,u))}{\pi(\mathrm{d}x)Q_{1}(x,\mathrm{d}(y,\zeta,u))} & =\frac{q(\theta',\theta)}{q(\theta,\theta')}\frac{\pi(\theta')}{\pi(\theta)}\frac{\pi_{\theta'}(\mathrm{d}z')\bar{Q}_{\theta',\theta,z'}(\mathrm{d}u)R_{\theta}(u_{0},\mathrm{d}z)}{\pi_{\theta}(\mathrm{d}z)Q_{\theta,\theta',z}(\mathrm{d}u)R_{\theta'}(u_{0},\mathrm{d}z')}\frac{\hat{C}_{\theta,\theta',z}(\zeta)}{\prod_{t=0}^{T}\frac{f_{\theta,\theta't+1}(u_{t})}{f_{\theta,\theta't}(u_{t})}}\\
 & =\frac{q(\theta',\theta)}{q(\theta,\theta')}\frac{\pi(\theta')}{\pi(\theta)}\frac{\pi(\theta)}{\pi(\theta')}\prod_{t=0}^{T}\frac{f_{\theta,\theta't+1}(u_{t})}{f_{\theta,\theta't}(u_{t})}\frac{\hat{C}_{\theta,\theta',z}(\zeta)}{\prod_{t=0}^{T}\frac{f_{\theta,\theta't+1}(u_{t})}{f_{\theta,\theta't}(u_{t})}}\\
 & =\frac{q(\theta',\theta)}{q(\theta,\theta')}\hat{C}_{\theta,\theta',z}(\zeta),
\end{align*}
we have used the identity in \eqref{eq:AISbasedAcceptRatio}.
\end{proof}
It should now be clear how the generalisation introduced in Algorithm
\ref{alg: MHAAR-SMC for general latent variable models} can be modified
for Algorithms \ref{alg: MHAAR-AIS exchange algorithm} and \ref{alg: MHAAR-AIS exchange algorithm - reduced computation},
which were developed for $\pi(x)=\pi(\theta)$. Let us remark the
main difference that in the algorithms of Section \ref{sec: Improving pseudo-marginal ratio algorithms for doubly intractable models},
$u_{0}^{(i)}$'s are sampled from the initial distribution of the
annealing schedule directly, whereas for the algorithms in Section
\ref{sec: Pseudo-marginal ratio algorithms for latent variable models},
we exploit the $z$ component of the current sample to start the SMC
(remember $u_{0}^{(i)}\sim R_{\theta}(z,\cdot)$, which becomes $u_{0}^{(i)}=z$
when $R_{\theta}(z,\cdot)=\delta_{z}(\cdot)$).
\end{document}